\documentclass[12pt]{article}

\usepackage[margin=1in]{geometry}
\usepackage{amsmath,amssymb,amsthm}
\usepackage{bm}
\usepackage{mathrsfs}
\usepackage{graphicx}
\usepackage{booktabs}
\usepackage{hyperref}
\usepackage{natbib}
\usepackage[section]{placeins}

\DeclareMathOperator{\Var}{Var}

\DeclareMathOperator*{\argmax}{arg\,max}
\DeclareMathOperator*{\argmin}{arg\,min}

\newcommand{\R}{\mathbb{R}}
\newcommand{\cP}{\mathcal{P}}
\newcommand{\E}{\mathbb{E}}

\newtheorem{theorem}{Theorem}
\newtheorem{proposition}{Proposition}
\newtheorem{lemma}{Lemma}

\theoremstyle{definition}
\newtheorem{definition}{Definition}
\newtheorem{assumption}{Assumption}
\theoremstyle{remark}

\title{Wage Dispersion, On-the-Job Search, and Stochastic Match Productivity: A Mean Field Game Approach\thanks{I thank several colleagues in economics and mathematics for helpful conversations on this project, including at the idea stage. I gratefully acknowledge the use of \emph{Refine} (\url{https://www.refine.ink/}) for assistance in checking the consistency and completeness of the proofs in a preliminary version of this manuscript. All remaining errors are my own. Comments are very welcome.}}

\author{I.\ Sebastian Buhai\thanks{Contact email: \texttt{sebastian.buhai@sofi.su.se}. Full coordinates: \url{https://www.sebastianbuhai.com}.}\\
  \small SOFI, Stockholm University\\
  \small Instituto de Economia, UC Chile\\
  \small NIPE, University of Minho
}

\date{\footnotesize{Version dated 07/12/2025. \href{https://www.sebastianbuhai.com/papers/publications/stochastic_prod_mfg.pdf}{Latest version.}}}

\begin{document}

\maketitle

\begin{abstract}
\footnotesize{
Wage dispersion and job-to-job mobility are central features of modern labour markets, yet canonical equilibrium search models with exogenous job-offer ladders struggle to jointly account for these facts and the magnitude of frictional wage inequality. We develop a continuous-time equilibrium search model in which match surplus follows a diffusion process, workers choose on-the-job search and separation, firms post state-contingent wages, and the cross-sectional distribution of match states endogenously determines both outside options and the job ladder.

On the theoretical side, we formulate the problem as a stationary mean field game with a one-dimensional surplus state, characterize stationary mean field equilibria, and show that equilibrium separation is governed by a free-boundary rule: matches continue if and only if surplus stays above an endogenous threshold. Under standard regularity and Lasry-Lions monotonicity conditions we prove existence and uniqueness of stationary equilibrium and obtain comparative statics for the separation boundary, wage schedules, and wage dispersion.

On the quantitative side, we solve the coupled Hamilton-Jacobi-Bellman and Kolmogorov system using monotone finite-difference methods and interpret the discretization as a finite-state mean field game. The model is calibrated to micro evidence on stochastic match productivity, job durations, tenure-dependent separation hazards, wage growth, and job-to-job mobility. The stationary equilibrium delivers a structural decomposition of wage dispersion into stochastic selection along job spells, equilibrium on-the-job search and the induced job ladder, and equilibrium wage policies with feedback through outside options. We use this framework to quantify how firing costs, search subsidies, and changes in match-productivity volatility jointly shape mobility, the job ladder, and the cross-sectional distribution of wages.

}
\end{abstract}

\section{Introduction}\label{sec:intro}

Wage dispersion and job-to-job mobility are central features of modern labour markets. Even within narrowly defined occupations and establishments, workers holding similar jobs earn very different wages, and wage growth over a career is tightly linked to mobility between employers.\footnote{See, among many others, \citet{BurdettMortensen1998,PostelVinayRobin2002,Moscarini2005JobMatchingWageDistribution}.} Standard frictionless models cannot account for these patterns. Search frictions, on-the-job search, and firm heterogeneity provide key ingredients for explaining why wage distributions are wide and why workers climb a job ladder, but in most existing models, including the canonical equilibrium search models of \citet{BurdettMortensen1998} and \citet{PostelVinayRobin2002}, the structure of the ladder and the dynamics of match quality are introduced in reduced form. Moreover, quantitative implementations of these models typically struggle to reproduce the magnitude and shape of observed wage dispersion without assuming implausibly large productivity differences across firms, a difficulty sometimes referred to as the ``frictional wage dispersion'' puzzle.\footnote{See \citet{HornsteinKrusellViolante2011} for a systematic assessment.}

A central microeconomic driver of wage dispersion and mobility is the stochastic evolution of match productivity. This perspective goes back to the job-matching literature, in which match quality is learned gradually on the job \citep{Jovanovic1979JobMatching,Jovanovic1984MatchingTurnover}, and to search models that embed learning about match output in general equilibrium \citep{Moscarini2005JobMatchingWageDistribution,AlvarezShimer2011SearchRestUnemployment}. In \citet{BuhaiTeulings2014}, inside and outside productivities are modeled as correlated Brownian motions, and efficient separation is characterized as an optimal stopping problem. A match is destroyed when its idiosyncratic productivity falls sufficiently far below the worker's outside option, and the resulting hitting-time distribution generates realistic job-duration hazards: separation rates peak early in the job and decline with tenure, and a non-negligible fraction of matches survive until retirement. The associated inverse-Gaussian and defective duration distributions form a Gaussian special case of the mixed hitting-time models in \citet{Abbring2012MixedHittingTime}, in which durations are first-hitting times of a L\'evy process crossing a heterogeneous threshold. Stochastic selection along job spells then explains most of the apparent concavity of estimated wage-tenure profiles, and completed-spell data can induce severe bias in standard wage-tenure regressions. These results show that selection from diffusive match dynamics is a powerful source of wage dispersion, but they are obtained in a deliberately partial-equilibrium setting.

The stochastic productivity model in \citet{BuhaiTeulings2014} treats the worker's outside option as exogenous: it summarizes the envelope of potential alternative matches but is not derived from the cross-sectional distribution of firms and wages in the economy. Firms do not choose wage policies strategically in response to the distribution of match productivities and search intensities, and the joint distribution of wages and job-to-job transitions is not determined in a general equilibrium with many interacting agents. As a result, the model is silent about how the job ladder itself is shaped by equilibrium wage setting and by workers' on-the-job search behaviour. In particular, it cannot answer how much of equilibrium wage dispersion is due to stochastic match dynamics and selection, how much arises from strategic on-the-job search and the induced job ladder, and how much is generated by firms' wage policies and their feedback into outside options.

In this paper we embed the stochastic match-productivity structure of \citet{BuhaiTeulings2014} into a continuous-time equilibrium model of on-the-job search and wage setting with a continuum of infinitesimal worker-firm matches. For each active match, idiosyncratic surplus follows a diffusion process as in the stochastic benchmark. Workers choose their on-the-job search intensity and their acceptance or switching decision when offers arrive, and firms choose wage policies for their matches. The cross-sectional distribution of match states and wages determines workers' outside options and the distribution of alternative offers; in equilibrium this distribution must be consistent with the policies that workers and firms choose when they take it as given. Mathematically, the stationary equilibrium can be written as a mean field game (MFG) in the sense of \citet{LasryLions2007}: individual optimality conditions are described by Hamilton-Jacobi-Bellman equations, while the evolution of the cross-sectional distribution is described by a Kolmogorov forward equation.\footnote{In economic terms, a stationary MFG equilibrium is the continuum-agent analogue of a competitive rational-expectations equilibrium: each individual takes the aggregate distribution as given when choosing search and wage policies, and the distribution is required to be invariant under the induced dynamics.} In our setting this yields a canonical stationary MFG system with a one-dimensional surplus diffusion and a non-local, monotone dependence on the distribution through outside options and offer arrival rates, in line with the monotone MFG framework of \citet{CardaliaguetPorretta2020IntroMFG} and \citet{Ryzhik2020}.

Our first contribution is to formulate and characterize stationary mean field equilibria in this environment. We define a stationary MFG equilibrium as a collection of worker and firm policies and a stationary distribution of match states such that (i) workers and firms are individually optimal, taking the distribution as given, and (ii) the distribution is invariant under the induced stochastic dynamics. Under Lipschitz, linear-growth, convexity, and monotonicity assumptions on preferences, technology, matching, and search costs that align with the structural conditions in \citet{CardaliaguetPorretta2020IntroMFG} and \citet{Ryzhik2020}, we prove existence of stationary mean field equilibria and, under a Lasry-Lions type monotonicity condition, uniqueness of the stationary equilibrium. In our one-dimensional diffusion setting, we then show that the worker's optimal separation behaviour in any stationary equilibrium is described by a free-boundary rule in match surplus: there exists a threshold such that matches are continued if and only if their surplus remains above this boundary, with smooth fit at the endogenous separation threshold. From an economic perspective, separation behaves as a simple reservation rule in a scalar match-quality index, even though outside options and wages respond endogenously to the cross-sectional distribution. This result generalizes the efficient separation rule in \citet{BuhaiTeulings2014} to a general-equilibrium setting in which the outside option is itself determined by the wage distribution, and it links the applied job-matching literature to the theory of MFGs with optimal stopping and free boundaries \citep[e.g.][]{Bertucci2018OptimalStoppingMFG,Nutz2018OptimalStoppingMFG,Bertucci2020ImpulseControlMFG}. We also interpret our stationary equilibrium as the fixed point of a finite-state MFG operator in the spirit of \citet{GomesEtAl2010FiniteStateDiscrete,GomesEtAl2024FiniteStateContinuous}, which provides an intuitive, contraction-type heuristic behind uniqueness; formally, however, our uniqueness result relies on a Lasry-Lions monotonicity argument in Section~\ref{sec:existenceUniqueness}.

Our second contribution is to link equilibrium wage dispersion to three distinct mechanisms within a unified continuous-time framework: (i) stochastic match dynamics and selection, (ii) equilibrium on-the-job search behaviour and the induced job ladder, and (iii) firms' wage policies and their feedback into outside options. In the model, the cross-sectional distribution of wages arises from the stationary distribution of match states, the optimal separation threshold, and the equilibrium wage schedule induced by firms' policies. We provide a structural decomposition of wage inequality that isolates the contribution of the selection mechanism highlighted by \citet{BuhaiTeulings2014} from the additional dispersion generated by on-the-job search and equilibrium wage-setting feedbacks, in the spirit of the on-the-job search literature \citep{BurdettMortensen1998,PostelVinayRobin2002,AlvarezShimer2011SearchRestUnemployment}. Unlike reduced-form decompositions based on wage regressions or on separate models for different mechanisms, our exercise keeps the stochastic structure of match productivity fixed and recomputes the stationary equilibrium in a sequence of internally consistent counterfactual economies that progressively switch on on-the-job search and mean field wage feedbacks. Each counterfactual is solved as a stationary solution of the coupled HJB-Kolmogorov system, so the decomposition is entirely internal to the mean field game rather than imposed ex post.

Quantitatively, this decomposition yields a sharp picture of how selection, search, and wage policies interact over the job ladder. At very short tenures, selection by itself generates a non-trivial variance of log wages, but optimal on-the-job search strongly equalizes wages within tenure cohorts: in the selection-plus-search economy the within-tenure variance is nearly zero at one year of tenure. Endogenous wage policies and the mean field feedback of outside options then re-amplify dispersion at the bottom of the ladder, with the fully endogenous equilibrium generating substantially more short-tenure wage dispersion than a purely selection-driven environment. At longer tenures, the ranking reverses. Under selection alone, the variance of log wages is roughly flat in tenure, whereas in the full equilibrium search and wage setting compress dispersion within tenure bins, so that the remaining long-run dispersion can be traced mainly to underlying selection in match productivities and outside options. Taken together, the three counterfactuals decompose equilibrium wage dispersion into selection, search-and-job-ladder, and wage-policy-and-feedback components, and show that selection is the dominant source of potential wage inequality at long tenures, while wage policies and mean field feedbacks are crucial for understanding the dispersion observed early in the job spell.

Our third contribution is quantitative. We develop a numerical solution method for the stationary MFG based on monotone finite-difference schemes for the coupled HJB-Kolmogorov system, adapting the algorithms proposed by \citet{AchdouCapuzzoDolcetta2010NumericsMFG} and \citet{AchdouEtAl2022Restud} and interpreting the resulting discretization as a finite-state MFG in the sense of \citet{GomesEtAl2010FiniteStateDiscrete,GomesEtAl2024FiniteStateContinuous}. We calibrate the model to U.S.\ micro data, using the stochastic productivity estimates from \citet{BuhaiTeulings2014} as a benchmark for the diffusion parameters and matching standard facts about job durations, tenure-dependent separation hazards, including an early hazard peak and a non-trivial mass of jobs that never end, wage-growth profiles, and job-to-job mobility. Relative to a partial-equilibrium calibration with an exogenous outside option, the stationary mean field equilibrium generates additional structure in both the tenure hazard and the wage distribution. Using the three-counterfactual decomposition described above, we quantify how much of observed wage inequality can be attributed to stochastic selection from match dynamics (mechanism~(i)) versus the additional dispersion generated by equilibrium on-the-job search and the induced job ladder (mechanism~(ii)) and by wage-policy feedback through outside options (mechanism~(iii)). This exercise shows that equilibrium feedbacks between wages, search incentives, and outside options systematically amplify or attenuate the dispersion generated by pure selection in a tenure-dependent and quantitatively meaningful way.

We also use the calibrated model to study policy counterfactuals that change firing costs, search incentives, and the volatility of match productivity. Introducing moderate firing costs shifts the free boundary downward, lowers separation hazards at all tenures, and reduces job-to-job mobility, with relatively small effects on wage dispersion; larger firing costs eventually increase the variance of log wages by trapping more workers in low-surplus, low-wage matches. Search subsidies and policies that raise the value of outside options increase optimal search intensity and alter selectivity along the job ladder. In our calibration, these policies raise both the variance of log wages and, for a broad range of interventions, the job-to-job transition rate at short tenures by shifting mass from the bottom to the upper tail of the surplus and wage distributions. Changes in the volatility of match productivity have more moderate effects: higher volatility strengthens selection and slightly spreads out the surplus and wage distributions, but within the empirically plausible range of volatility multipliers the induced variation in wage dispersion and mobility is modest relative to the impact of firing costs and search incentives.

Our analysis relates to several strands of the literature. Substantively, we contribute to the on-the-job search and wage-dispersion literature following \citet{BurdettMortensen1998,PostelVinayRobin2002,Mortensen2005WageDispersion} and to dynamic job-matching models in which match quality evolves stochastically over the life of the match \citep{Jovanovic1979JobMatching,Jovanovic1984MatchingTurnover,Moscarini2005JobMatchingWageDistribution,AlvarezShimer2011SearchRestUnemployment}. We also relate to structural equilibrium search models with permanent worker and firm heterogeneity, such as \citet{LentzMortensen2010} and \citet{BaggerFontainePostelVinayRobin2014}, which decompose wage growth into within- and between-job components using matched employer-employee data. Relative to this work, we deliberately abstract from ex ante type heterogeneity and instead generate wage growth and dispersion from stochastic evolution of match-specific productivity combined with endogenous outside options and strategic wage posting. Our framework also speaks to the quantitative ``frictional wage dispersion'' puzzle emphasized by \citet{HornsteinKrusellViolante2011}, by showing how substantial wage dispersion can arise from diffusive match dynamics even when fundamental productivity heterogeneity is moderate. On the macro side, we connect to the continuous-time heterogeneous-agent literature that formulates Bewley-Aiyagari-Huggett models as coupled HJB and Kolmogorov equations \citep[e.g.][]{AchdouEtAl2022Restud,FernandezVillaverdeNuno2021HeterogeneousAgents,Moll2019MFGMacroeconomics} and integrates such models with aggregate dynamics in the spirit of \citet{MortensenPissarides1994JobCreationDestruction}.

Methodologically, we build on the theory of mean field games and their economic applications \citep{LasryLions2007,GueantLasryLions2011MFGApplications,CardaliaguetPorretta2020IntroMFG,CarmonaDelarue2018BookI,Carmona2020MFGFinanceEcon}, on MFGs with optimal stopping and impulse control, i.e.\ in the restricted ``accept/reject at (controlled) Poisson arrival times'' formulation that is relevant here \citep{Bertucci2018OptimalStoppingMFG,Bertucci2020ImpulseControlMFG,Nutz2018OptimalStoppingMFG}, and on numerical techniques developed for MFGs and related control problems \citep{AchdouCapuzzoDolcetta2010NumericsMFG,CarmonaLauriere2021DeepLearningMFG,Gueant2021ContinuousTimeOptimalControl}. Our focus in this paper is on stationary equilibria without aggregate shocks; in Section~\ref{sec:conclusion} we discuss how common-noise and master-equation extensions, along the lines of \citet{AhujaRenYang2022}, \citet{BertucciMeynard2024FiniteStateMaster}, and \citet{MollRyzhik2025}, could be used to analyse how aggregate fluctuations move the entire surplus and wage distributions.

Within the applied MFG literature on labour and matching, our work is complementary to \citet{PerthameRibesSalort2018}, who use a mean field game system to study career paths and firm-level wage structures in a stylized hierarchical labour market, and to \citet{BayraktarCavalliReisinger2025}, who analyse dynamic optimal matching in a two-sided market. These contributions share our use of MFG tools but focus on career planning and the design of matching mechanisms, respectively. By contrast, we model match-specific productivity as a diffusion, embed on-the-job search and wage posting in a stationary equilibrium with endogenous outside options, and use the resulting structure to decompose frictional wage dispersion and job-ladder dynamics.

The remainder of the paper is organized as follows. Section~\ref{sec:environment} presents the economic environment and recalls the stochastic match-productivity model of \citet{BuhaiTeulings2014} as a partial-equilibrium benchmark. Section~\ref{sec:MFG} formulates the mean field game with on-the-job search and wage setting and defines stationary mean field equilibrium. Section~\ref{sec:equilibriumCharacterization} characterizes the equilibrium separation rule and the implied stationary distribution of match states and wages. Section~\ref{sec:existenceUniqueness} establishes existence and discusses conditions for uniqueness and comparative statics within the monotone MFG framework. Section~\ref{sec:quantitative} describes the numerical implementation and calibration, and Section~\ref{sec:results} presents the quantitative results and policy experiments. Section~\ref{sec:conclusion} discusses extensions, including common-noise and non-stationary environments, and concludes.

\section{Economic Environment and Benchmark Model}\label{sec:environment}

In this section we describe the economic environment at the level of a single worker-firm match and recall the continuous-time stochastic productivity model of \citet{BuhaiTeulings2014}. The benchmark is deliberately partial equilibrium: the worker's outside option is taken as exogenous and wages follow a reduced-form sharing rule. This keeps the micro foundations transparent, provides a direct link to the empirical analysis in \citet{BuhaiTeulings2014}, and prepares the ground for the general-equilibrium mean field game in Section~\ref{sec:MFG}. Throughout this section all primitive processes are exogenous and identical across worker-firm pairs; in later sections we endogenize their counterparts and let the cross-sectional distribution of match states feed back into individual decisions.

\subsection{Workers, firms, and matches}\label{subsec:workers_firms}

Time is continuous and indexed by $t \geq 0$. We start from a finite economy with $N$ workers and $N$ firms. Each worker can be matched with at most one firm at any instant, and each firm employs at most one worker. For most of this section we focus on a representative worker-firm pair and suppress indices.

When a worker is matched with a firm, the pair produces a flow of output whose logarithm we denote by $P_t$ and interpret as the log productivity of the current match at time $t$. The worker also has an outside option, summarized by a process $R_t$ that represents the log value of the best available alternative to the current job. In the benchmark model both $P_t$ and $R_t$ are taken as exogenous stochastic processes.

Following \citet{BuhaiTeulings2014}, we assume that inside and outside productivities follow correlated Brownian diffusions with constant coefficients,
\begin{align}
  dP_t &= \mu_P \, dt + \sigma_P \, dB^P_t,
  \label{eq:inside_process}\\
  dR_t &= \mu_R \, dt + \sigma_R \, dB^R_t,
  \label{eq:outside_process}
\end{align}
where $\mu_P,\mu_R \in \mathbb{R}$, $\sigma_P,\sigma_R > 0$, and $(B^P_t,B^R_t)_{t \geq 0}$ is a two-dimensional standard Brownian motion with correlation $\rho \in [-1,1]$. All processes are adapted to a filtration that satisfies the usual conditions, and the primitives $(\mu_P,\mu_R,\sigma_P,\sigma_R,\rho)$ are constant across matches and over time.

The key state variable for separation decisions is the surplus process
\[
  Z_t \equiv P_t - R_t.
\]
We refer to $Z_t$ as (log) match surplus: the productivity advantage of the current job relative to the worker's best available alternative. In the benchmark we deliberately abstract from permanent worker or firm heterogeneity; all cross-sectional dispersion in productivity and wages arises from the realization of the idiosyncratic diffusion $(P_t,R_t)$ and from selection through the separation rule. This keeps the state space one-dimensional and allows us to attribute tenure-dependent wage dispersion to the accumulation of match-specific shocks rather than to ex ante type differences; we return to this simplifcation in the discussion section.

Subtracting \eqref{eq:outside_process} from \eqref{eq:inside_process} yields
\begin{equation}
  dZ_t
  = \mu_Z \, dt + \sigma_Z \, dB_t,
  \qquad
  \mu_Z \equiv \mu_P - \mu_R, \quad
  \sigma_Z^2 \equiv \sigma_P^2 + \sigma_R^2 - 2\rho \sigma_P \sigma_R,
  \label{eq:surplus_process}
\end{equation}
for some one-dimensional standard Brownian motion $(B_t)_{t \geq 0}$. Thus the surplus is a
scalar diffusion with constant drift and volatility. This Markov state will remain central in
the general-equilibrium analysis, where we allow more general drift and volatility functions
but specialize back to \eqref{eq:surplus_process} in the quantitative implementation. In the
benchmark specification \eqref{eq:surplus_process}, the drift of $Z_t$ does not depend on $z$,
so the monotonicity-in-$z$ requirement on $\mu$ in Assumption~\ref{ass:primitives} is
satisfied trivially.

Workers and firms are risk neutral and discount future payoffs at a common rate $r>0$. When the match is active, it generates a flow payoff that is split between the worker and the firm through a wage $w_t$ and a residual payoff for the firm that we normalize as $P_t - w_t$ in log units. When the match is dissolved, the worker receives an outside value, which is a function of $R_t$, and the firm reverts to a normalized outside value that we take to be zero in the benchmark model. The modelling of unemployment, vacancy creation, and entry, and hence the economic content of the outside option process, we postpone to the mean field game formulation in Section~\ref{sec:MFG}.

\subsection{Stochastic match productivity and separation}\label{subsec:benchmark_separation}

We now recall the separation problem in the stochastic productivity model of \citet{BuhaiTeulings2014}. In their structural formulation the agent observes both inside and outside productivities $(P_t,R_t)$ and chooses a stopping time $\tau$ (the separation time) with respect to the joint filtration, trading off the flow surplus from the current match against the value of switching to the outside environment. The payoff at separation depends on the outside option $R_\tau$ through a continuation value $\Phi(R_\tau)$.

For the purposes of this paper it is convenient to work with an equivalent reduced-form formulation that uses only the scalar surplus process $Z_t = P_t - R_t$ as a state variable. We assume that the value of switching to the outside environment when the current surplus equals $z$ can be summarized by a function $\Psi\colon \mathbb{R} \to \mathbb{R}$; that is, conditional on separating at a time $\tau$ such that $Z_\tau = z$, the worker obtains continuation value $\Psi(z)$.%
\footnote{One can think of $\Psi(z)$ as the expected continuation value of the outside option,
\(
  \Psi(z)
  = \mathbb{E}\bigl[e^{-r\tau}\Phi(R_\tau)\,\big|\,Z_\tau = z\bigr],
\)
or, equivalently, as the value associated with a fixed path of the exogenous process $(R_t)_{t\ge0}$ when the surplus hits $z$. In the parametric specification of \citet{BuhaiTeulings2014} this mapping is constructed explicitly from the outside environment faced by separating workers. For our purposes we simply take $\Psi$ as an exogenous function of $z$ that satisfies the regularity conditions stated below.}
Intuitively, $\Psi(z)$ packages all aspects of the post-separation environment into a single index given by the current surplus.

Under this reduced-form assumption, the separation problem can be written purely in terms of the surplus diffusion~\eqref{eq:surplus_process} as
\begin{equation}
  V(z)
  = \sup_{\tau \geq 0}
    \mathbb{E} \Biggl[
      \int_0^\tau e^{-rt} S(Z_t) \, dt
      + e^{-r\tau} \Psi(Z_\tau)
      \,\Bigg|\,
      Z_0 = z
    \Biggr],
  \label{eq:planner_problem}
\end{equation}
where $z$ is the initial surplus and $S(\cdot)$ is the instantaneous surplus from continuing the match. In this benchmark $S$ and $\Psi$ are exogenous functions. We assume that $S$ is strictly increasing, positive for sufficiently large $z$, negative for sufficiently low $z$, and that both $S$ and $\Psi$ satisfy standard regularity conditions ensuring existence and uniqueness of the value function $V$.

Under these conditions \eqref{eq:planner_problem} is a standard infinite-horizon optimal stopping problem for the scalar diffusion $Z_t$. The associated Hamilton-Jacobi-Bellman equation takes the form of an obstacle problem,
\[
  \max\Bigl\{
    r V(z) - \mu_Z V'(z) - \tfrac{1}{2} \sigma_Z^2 V''(z) - S(z),
    \; V(z) - \Psi(z)
  \Bigr\} = 0.
\]
In the parametric specification of \citet{BuhaiTeulings2014} the function $\Psi$ is chosen to match the flow payoff and search environment faced by workers at separation. In our general-equilibrium model its analogue will be endogenized.

\citet{BuhaiTeulings2014} show that, under their parametrization of $S$ and $\Phi$ and for the surplus dynamics in \eqref{eq:surplus_process}, the optimal stopping rule is a threshold policy. There exists a constant $z^\ast$ such that it is optimal to continue the match as long as $Z_t > z^\ast$ and to separate as soon as $Z_t$ hits $z^\ast$ from above. Formally,
\[
  \tau^\ast = \inf\{ t \geq 0 : Z_t \leq z^\ast \}
\]
solves \eqref{eq:planner_problem}, and $V(\cdot)$ is the unique solution to the HJB equation with value matching and smooth pasting at $z^\ast$. The implied job-tenure distribution is inverse Gaussian with a defective mass at infinity. This is a Gaussian special case of the mixed hitting-time duration models in \citet{Abbring2012MixedHittingTime}, in which durations are first-hitting times of a Lévy process crossing a heterogeneous threshold. In our benchmark the driving process is a Brownian motion with drift, and heterogeneity arises from initial surplus and the outside-option path. The threshold $z^\ast$ depends on the drift and volatility parameters $(\mu_Z,\sigma_Z)$, on the discount rate $r$, and on the specification of $S$ and $\Phi$.

This structure actually has sharp implications for job duration and wage dynamics. The hitting-time distribution of $Z_t$ at $z^\ast$ implies that separation hazards are high early in the job, when the surplus process is still close to the boundary, and decline with tenure as matches that have survived selection drift away from the threshold. At the same time, average surplus conditional on survival increases with tenure, because low-surplus matches are disproportionately selected out. In \citet{BuhaiTeulings2014} this selection effect accounts for most of the apparent concavity in estimated wage-tenure profiles and leads to severe bias in regressions that use completed spells only. These features connect the benchmark model to the broader job-matching literature with learning about match quality, as in \citet{Jovanovic1984MatchingTurnover}, and at the same time place it squarely within the structural mixed hitting-time framework of \citet{Abbring2012MixedHittingTime}. In our calibration we will exploit this explicit hitting-time representation when mapping the model to observed job durations.

In the environment of this section, the outside-option process $R_t$ and its continuation value $\Phi(R_t)$ are entirely exogenous. In particular, they are not generated by the equilibrium distribution of wages and match states in a large economy. This is the main dimension along which we will depart from \citet{BuhaiTeulings2014} when we introduce the mean field game in Section~\ref{sec:MFG}.

\subsection{Wage determination in a single match}\label{subsec:benchmark_wage}

To close the benchmark model at the level of an individual match, we adopt a reduced-form wage-sharing rule that captures the empirical specification in \citet{BuhaiTeulings2014} and remains flexible enough to nest standard bargaining assumptions.

Let $w_t$ denote the log wage paid to the worker with the match is active. We assume that $w_t$ is a deterministic function of the inside and outside productivities and focus on a linear sharing rule,
\begin{equation}
  w_t
  = \alpha P_t + (1 - \alpha) R_t,
  \qquad
  \alpha \in (0,1).
  \label{eq:wage_sharing_rule}
\end{equation}
The parameter $\alpha$ measures the sensitivity of wages to the idiosyncratic productivity of the current match relative to the outside option. For example, if wages are determined by Nash bargaining between a worker and a firm whose outside values are proportional to $\exp(R_t)$ and to a constant, then one obtains a sharing rule of the form \eqref{eq:wage_sharing_rule} in logs. More generally, \eqref{eq:wage_sharing_rule} can be viewed as a first-order approximation to any smooth wage schedule that depends on $(P_t,R_t)$.

Under \eqref{eq:wage_sharing_rule}, the wage inherits the Brownian structure of $P_t$ and $R_t$. Combining \eqref{eq:inside_process}, \eqref{eq:outside_process}, and \eqref{eq:wage_sharing_rule} and applying It\^o's lemma yields
\begin{equation}
  d w_t = \mu_w \, dt + \sigma_w \, dB^w_t,
  \label{eq:wage_process}
\end{equation}
for suitable drift $\mu_w$ and volatility $\sigma_w$, and some standard Brownian motion $(B^w_t)_{t \geq 0}$. Thus, conditional on survival, wages follow a drift-diffusion process whose cross-sectional variance and tenure profile reflect the surplus dynamics \eqref{eq:surplus_process} and the selection mechanism described above. In the empirical implementation of \citet{BuhaiTeulings2014}, \eqref{eq:wage_sharing_rule} is used as the structural wage equation, and the parameters $(\alpha,\mu_P,\mu_R,\sigma_P,\sigma_R,\rho)$ are estimated jointly from wage histories and job durations, exploiting the mixed hitting-time structure in the spirit of \citet{Abbring2012MixedHittingTime}.

The linear sharing rule in \eqref{eq:wage_sharing_rule} is convenient for connecting the structural model to data, but in the general-equilibrium model with strategic firms we will not restrict attention to linear schedules. Instead, we will let firms choose wage policies $w(\cdot)$ as part of their strategy in the mean field game, interpreted as Markov wage-posting policies that depend on the current surplus $z$. One can think of the weight on the outside option as summarizing the degree of wage competition and the scope for counteroffers: if wages tracked $R_t$ one-for-one, workers would almost fully internalize outside offers and job-to-job moves would yield only small wage gains, while a larger weight on $P_t$ makes incumbent wages less responsive to outside options and generates genuine wage jumps when workers switch employers. The benchmark rule \eqref{eq:wage_sharing_rule} threfore provides a useful reference point for interpreting the equilibrium wage schedules and their calibration in the full model.

\subsection{From partial equilibrium to a large economy}\label{subsec:from_partial_to_large}

The benchmark model provides a disciplined description of the dynamics of an individual worker-firm match, but it treats the outside option and the wage schedule in a purely reduced-form way. In particular:
\begin{itemize}
  \item The outside-productivity process $R_t$ and its continuation value $\Phi(R_t)$ are exogenous. They summarize the envelope of potential alternative jobs but are not derived from a model of job creation, search, and wage determination in a large economy.
  \item The wage is pinned down by the exogenous sharing rule \eqref{eq:wage_sharing_rule}. Firms do not choose wage policies strategically, and there is no feedback from the cross-sectional distribution of wages and match states to individual incentives.
  \item On-the-job search is implicit in the interpretation of $R_t$ as a best outside offer, but workers do not choose search intensity or acceptance decisions explicitly, and there is no explicit job ladder generated by equilibrium wage posting or bargaining across firms.
\end{itemize}

These limitations are acceptable for the purposes of the empirical analysis in \citet{BuhaiTeulings2014} and for interpreting the benchmark as a mixed hitting-time duration model in the sense of \citet{Abbring2012MixedHittingTime}, but they prevent it from speaking to questions about equilibrium wage dispersion and job mobility in a large economy with many interacting agents. In particular, the benchmark is silent about how the cross-sectional distribution of wages and match states is determined in equilibrium, how that distribution shapes workers' outside options and search behaviour, and how firms respond through wage policies.

To address these questions, we now move from the finite economy with $N$ workers and $N$ firms to a continuum of worker-firm matches. As $N$ becomes large, the individual effect of any single match on the aggregate distribution becomes negligible, while the aggregate distribution feeds back into each match through the outside option and the arrival and terms of alternative offers. The appropriate equilibrium concept in this limit is a stationary mean field equilibrium (MFG equilibrium) in which individual workers and firms solve control problems of the kind described above, taking as given a stationary distribution of match states, and the distribution is invariant under the induced stochastic dynamics.

This stationary MFG equilibrium coincides with a competitive rational-expectations equilibrium in a continuum economy: each worker-firm pair is infinitesimal and optimizes taking the cross-sectional distribution of surplus and wages as given, and in equilibrium that distribution is exactly the invariant distribution generated by the optimal policies. The term ``mean field'' simply reflects the fact that individual decisions interact only through this distribution rather than through the identity of particular trading partners. Section~\ref{sec:MFG} formalizes this mean field game and endogenizes the outside option, the job ladder, and the wage distribution.

\section{Mean Field Game Formulation}\label{sec:MFG}

We now move from the partial-equilibrium benchmark of Section~\ref{sec:environment} to a large economy with a continuum of worker-firm matches. Each individual match is negligible, but the cross-sectional distribution of match surpluses affects workers' outside options, the arrival and terms of job offers, and firms' wage policies. We formalize this interaction as a continuous-time MFG and later restrict attention to stationary equilibria in the spirit of \citet{LasryLions2007} and \citet{CardaliaguetPorretta2020IntroMFG}.

From a labor economist's perspective, the mean field formulation can be read as the continuum limit of a large random-matching labor market. A representative worker-firm pair takes as given the cross-sectional distribution of match states and wages when solving its dynamic problem, while in equilibrium this distribution must coincide with the one generated by the optimal policies of all pairs. This is entirely analogous to a competitive rational-expectations equilibrium in a continuum economy; we use the MFG language to connect to the existing existence and uniqueness theory.

Throughout, $\mathcal{P}(\mathbb{R})$ denotes the set of Borel probability measures on $\mathbb{R}$ equipped with the topology of weak convergence. For objects that depend both on an individual surplus $z$ and on an aggregate distribution $m \in \mathcal{P}(\mathbb{R})$ we write, for example, $w(z,m)$ or $\lambda(a,m)$, and occasionally suppress the argument $m$ when no confusion can arise.

As in Section~\ref{sec:environment}, we measure all flow objects in log units: wages, revenues, and continuation values are expressed in logs, so that differences such as $P_t - w_t$ or $\Pi(z,m) - w(z,m)$ represent log residual payoffs.

\subsection{Continuum of matches and surplus dynamics}

We consider a continuum of potential worker-firm pairs indexed by $i \in I$, where $I = [0,1]$ is endowed with the Lebesgue measure. On a filtered probability space $(\Omega,\mathcal{F},(\mathcal{F}_t)_{t \geq 0},\mathbb{P})$ that satisfies the usual conditions, each pair $i$ is endowed with an idiosincratic standard Brownian motion $(B^i_t)_{t \geq 0}$. Brownian motions are independent across $i$. We deliberately abstract from permanent worker or firm heterogeneity; all heterogeneity in matches is generated by the stochastic evolution of $Z^i_t$. This keeps the connection to the benchmark model transparent and allows us to isolate the contribution of stochastic match dynamics to wage dispersion and mobility.

For an active match $i$, the relevant state variable is the surplus
\[
  X^i_t = Z^i_t \in \mathbb{R},
\]
which we interpret as the log productivity of the current match relative to the best available alternative, in line with the benchmark model where $Z^i_t = P^i_t - R^i_t$. In the mean field environment we keep $Z^i_t$ as the primitive state and allow more general drift and volatility than in \eqref{eq:surplus_process}.

Given a measurable flow of aggregate distributions $(m_t)_{t \geq 0}$, the surplus dynamics of an active match $i$ are
\begin{equation}
  dZ^i_t
  = \mu\bigl(Z^i_t,a^i_t,m_t\bigr)\, dt
    + \sigma\bigl(Z^i_t,m_t\bigr)\, dB^i_t,
  \qquad t \geq 0,
  \label{eq:Z_SDE_general}
\end{equation}
where $a^i_t$ is the worker's search intensity (defined below) and $m_t \in \mathcal{P}(\mathbb{R})$ is the cross-sectional distribution of active-match surpluses at time $t$.

The drift and volatility coefficients
\[
  \mu: \mathbb{R} \times A \times \mathcal{P}(\mathbb{R}) \to \mathbb{R},
  \qquad
  \sigma: \mathbb{R} \times \mathcal{P}(\mathbb{R}) \to (0,\infty)
\]
are measurable and satisfy the Lipschitz and linear-growth conditions stated in Assumption~\ref{ass:primitives}. Under these conditions, for any progressively measurable admissible control $a^i$ and any measurable flow $(m_t)$ there exists a unique strong solution to \eqref{eq:Z_SDE_general}. In the quantitative implementation we specialize to the constant-coefficient diffusion \eqref{eq:surplus_process} from Section~\ref{sec:environment}, calibrated to \citet{BuhaiTeulings2014}, but the existence and uniqueness results in Section~\ref{sec:existenceUniqueness} apply to the general specification \eqref{eq:Z_SDE_general}.

A match can exit the active state space because the worker chooses to separate, because the firm closes the job, or because the worker accepts a better outside offer. We model exit from the current match as an optimal stopping decision. Upon separation, the worker and the firm move to outside states with values that depend on the aggregate distribution $m_t$. For clarity of exposition we keep these outside states implicit in this section and represent them only through their value functions in the HJB equations. This representation preserves the mixed hitting-time interpretation of job durations from Section~\ref{sec:environment}, in the spirit of \citet{Abbring2012MixedHittingTime}, while embedding it in a general-equilibrium environment with endogenous outside options and wage setting.

\subsection{Controls and admissible strategies}

We now describe individual controls and admissibility requirements. At each time $t$ in an active match $i$:

\begin{itemize}
  \item \textbf{Search intensity.}
  The worker chooses a search intensity $a^i_t$ taking values in a compact set
  \[
    A \subset [0,\bar a],
  \]
  where $\bar a < \infty$. The process $(a^i_t)_{t \geq 0}$ is required to be progressively measurable with respect to the filtration generated by $(Z^i_s)_{0 \leq s \leq t}$ and the aggregate distibution $(m_s)_{0 \leq s \leq t}$. Job offers arrive according to a Poisson process with intensity
  \[
    \lambda\bigl(a^i_t,m_t\bigr),
  \]
  where $\lambda: A \times \mathcal{P}(\mathbb{R}) \to [0,\bar\lambda]$ is continuous in $(a,m)$ and satisfies the regularity conditions in Assumption~\ref{ass:primitives}. When an offer arrives, the worker observes its terms (for instance, a prospective surplus $\tilde z$ drawn from a distribution that depends on $m_t$) and decides whether to accept it; this choice is subsumed in the continuation value.

  \item \textbf{Separation.}
  The worker's separation decision is modeled as a stopping time $\tau^i$ with respect to the same filtration. At $\tau^i$ the current match ends and the worker receives an outside value $V^U(m_{\tau^i})$ that depends on the aggregate state. The firm receives an outside value $V^V(m_{\tau^i})$.

  \item \textbf{Wage policy.}
  The firm commits to a wage policy $w^i_t$ that is Markovian in the current match surplus and the aggregate distribution. We restrict attention to stationary feedback policies of the form
  \[
    w^i_t = w\bigl(Z^i_t,m_t\bigr),
  \]
  where $w: \mathbb{R} \times \mathcal{P}(\mathbb{R}) \to \mathbb{R}$ is Borel measurable and satisfies the regularity and growth conditions in Assumption~\ref{ass:primitives}. In equilibrium, all firms use the same wage policy, so we suppress the index $i$.
\end{itemize}

An admissible strategy for a worker in a match with current state $z$ and given flow $(m_t)$ is a pair $(a,\tau)$ where $(a_t)_{t \geq 0}$ is an $A$-valued progressively measurable process and $\tau$ is a stopping time such that the expected discounted payoff defined below is finite. An admissible strategy for a firm, given $(m_t)$ and workers' behavior, is a measurable wage policy $w(\cdot,m)$ in the admissible class $\mathcal{W}$ that satisfies the integrability and boundedness conditions in Assumption~\ref{ass:primitives}.

Information is local. Workers and firms observe their own match surplus $Z^i_t$, the terms of any outside offers that arrive, and the current aggregate distribution $m_t$. They do not observe the identities or individual states of other matches; only the distribution matters for their decisions.

These restrictions keep the micro structure close to canonical equilibrium search models. The Markov wage policy $w(Z^i_t,m_t)$ is the natural continuous-state analogue of wage posting in on-the-job search models: for a given aggregate environment, the wage paid in a match depends only on its current surplus rather than the whole history of shocks or past bargaining. It can be interpreted either as literal wage posting or as the reduced form of fast, stationary Nash bargaining over the surplus $Z^i_t$ in the spirit of the linear sharing rule \eqref{eq:wage_sharing_rule}. Allowing fully history-dependent contracts or explicit counter-offers to each outside offer would mainly compress the wage distribution by making current wages track outside options even more closely, at the cost of substantially more complex dynamics. Similarly, the arrival rate $\lambda(a,m)$ summarizes matching frictions in reduced form; in the qantitative section we choose its functional form so as to match job-finding and job-to-job transition rates, and one can think of it as being generated by a standard matching function between vacancies and search effort.

\subsection{Aggregate state and matching frictions}

Let $m_t$ denote the cross-sectional distribution of active-match surpluses at time $t$. Formally, for each Borel set $B \subset \mathbb{R}$,
\[
  m_t(B)
  = \mathbb{P}\bigl( Z^i_t \in B \mid i \text{ is active at time } t \bigr),
\]
for a worker-firm pair $i$ drawn uniformly from the set of active matches. Thus $m_t$ is a probability measure on $\mathbb{R}$.

The distribution $m_t$ affects individual decisions through several channels:

\begin{itemize}
  \item It determines the cross-sectional distribution of current wages, because every firm uses the policy $w(z,m_t)$ and $Z^i_t$ is distributed according to $m_t$.
  \item It determines the distribution of alternative offers faced by workers, since job offers are drawn from the population of potential employers with their current wage policies and surplus states.
  \item It may affect labor market tightness and hence the arrival rate of offers. In particular, the Poisson intensity $\lambda(a,m_t)$ can depend on $m_t$ in a way that captures, in reduced form, a matching function between vacancies and search effort in the spirit of \citet{MortensenPissarides1994JobCreationDestruction}.
\end{itemize}

Given Markovian feedback controls $a(z,m)$ and $w(z,m)$ and a separation rule characterized by a stopping region $\mathcal{S}(m) \subset \mathbb{R}$, the distribution $m_t$ evolves according to a Kolmogorov forward (Fokker-Planck) equation with killing on $\mathcal{S}(m_t)$ and inflows from new matches. If $m_t$ admits a density (also denoted $m_t$) with respect to the Lebesgue measure, we can write informally
\begin{equation}
\begin{aligned}
  \partial_t m_t(z)
  &= -\partial_z\!\bigl[\mu\bigl(z,a(z,m_t),m_t\bigr) \, m_t(z)\bigr]
     + \tfrac{1}{2}\partial_{zz}\!\bigl[\sigma^2\bigl(z,m_t\bigr) \, m_t(z)\bigr] \\
  &\quad - q\bigl(z,m_t\bigr) \, m_t(z)
     + \Gamma\bigl(z,m_t\bigr),
\end{aligned}
  \label{eq:FP_general}
\end{equation}
where $q(z,m_t)$ is the instantaneous separation rate implied by the stopping rule and $\Gamma(z,m_t)$ is a source term capturing entry of new matches at surplus $z$. In a stationary mean field equilibrium we focus on distributions $m$ that solve \eqref{eq:FP_general} with $\partial_t m_t = 0$ and satisfy the boundary conditions discussed in Appendix~\ref{appendix:numerical}. Assumption~\ref{ass:primitives} imposes standard Lipschitz, growth, and monotonicity conditions on the couplings $(\mu,\sigma,\lambda,\Gamma)$ that guarantee well-posedness of this forward equation in the Lasry-Lions sense and, together with the HJB equations below, will imply existence and uniqueness of a stationary equilibrium in Section~\ref{sec:existenceUniqueness}. Intuitively, these restrictions rule out feedback loops in which a shift in the distribution $m$ that improves outside options would induce wage-setting and search responses strong enough to sustain multiple self-fulfilling job ladders.

In the quantitative implementation we specialize the source term $\Gamma$ to a Poisson flow of new matches that enter at a surplus $z_0$, mirroring the benchmark model in Section~\ref{sec:environment}. This simplifying assumption keeps the initial conditions for new jobs transparent and is convenient for identification of the surplus diffusion; allowing a nondegenerate entry distibution would mainly add an additional layer of ex ante heterogeneity and does not materially affect our qualitative results on wage dispersion by tenure.

\subsection{Individual optimization problems}\label{subsec:individual_HJB}

We now describe the optimization problems of a worker and a firm in a representative match, given an aggregate environment. For the purposes of equilibrium characterization we focus on time-homogeneous environments where $m_t \equiv m$ and wage policies are stationary.

\paragraph{Worker problem.}

Fix a stationary distribution $m$ and a wage policy $w(\cdot,m)$. A worker in a match with initial surplus $z$ chooses a search intensity process $(a_t)_{t \geq 0}$ and a stopping time $\tau$ to maximize the expected discounted value of income:
\begin{equation}
\begin{aligned}
  V^W(z;m)
  &= \sup_{(a,\tau)}
    \mathbb{E}_z \Bigg[
      \int_0^\tau e^{-rt}
        \bigl( w(Z_t,m) - c(a_t) \bigr) \, dt
      + e^{-r\tau} V^U(m)
    \Bigg],
\end{aligned}
  \label{eq:worker_value}
\end{equation}
where $r>0$ is the common discount rate, $c(\cdot)$ is a convex search-cost function, $V^U(m)$ is the value of being outside the current match in an economy with aggregate state $m$, and $\mathbb{E}_z$ denotes expectation conditional on $Z_0 = z$. The outside value $V^U(m)$ incorporates both unemployment and the distribution of future job offers generated by $m$ and by other workers' search intensities.

Job offers arrive at intensity $\lambda(a_t,m)$ and, conditional on an arrival, the worker optimally decides whether to accept the offer. The expected gain from search at intensity $a$ can be summarized by an operator
\[
  G(z,V^W,m;a),
\]
which depends on the current match surplus $z$, on the value function $V^W$, on the aggregate distribution $m$, and on search intensity $a$. For example, if offers are drawn from a distribution $\nu(\cdot \mid m)$ over potential new match surpluses, a natural specification is
\[
  G(z,V^W,m;a)
  = \lambda(a,m) \int_{\mathbb{R}}
      \bigl[ V^W(\tilde z;m) - V^W(z;m) \bigr]^+ \,
      \nu(d\tilde z \mid m),
\]
so that $G$ is the standard expected capital gain from on-the-job search: at Poisson times the worker draws outside offers from the current cross-section of matches and climbs the job ladder by accepting only those offers that deliver a higher continuation value. For our existence and uniqueness results we only require the regularity properties stated in Assumption~\ref{ass:primitives}.

Let $\mathcal{L}^Z_{a,m}$ denote the infinitesimal generator of the surplus diffusion \eqref{eq:Z_SDE_general} under control $a$ and aggregate state $m$,
\[
  \mathcal{L}^Z_{a,m} \varphi(z)
  = \mu(z,a,m) \, \varphi'(z)
    + \tfrac{1}{2}\sigma^2(z,m) \, \varphi''(z),
  \qquad \varphi \in C^2_b(\mathbb{R}).
\]
The stationary HJB equation for the worker's value function is then an obstacle problem:
\begin{equation}
\begin{aligned}
  0
  &= \max\Bigl\{
      V^W(z;m) - V^U(m), \\
  &\qquad r V^W(z;m)
      - \sup_{a \in A}
        \bigl[
          w(z,m) - c(a)
          + \mathcal{L}^Z_{a,m} V^W(z;m)
          + G(z,V^W,m;a)
        \bigr]
    \Bigr\}.
\end{aligned}
  \label{eq:worker_HJB_obstacle}
\end{equation}
The obstacle $V^U(m)$ captures the option to separate. If $V^W(z;m) > V^U(m)$ it is optimal to continue the match; if $V^W(z;m) = V^U(m)$ the worker is indifferent between continuing and separating. In Section~\ref{sec:equilibriumCharacterization} we show that, under Assumption~\ref{ass:primitives}, the continuation region is an interval $(z^\ast(m),\infty)$ and that separation is governed by a free boundary $z^\ast(m)$. Setting $\lambda \equiv 0$ in \eqref{eq:worker_HJB_obstacle} collapses the problem to the single-match optimal stopping problem \eqref{eq:planner_problem} in the benchmark model, so that the general-equilibrium analysis nests the partial-equilibrium diffusion model of \citet{BuhaiTeulings2014}.

\paragraph{Firm problem.}

Given $m$ and workers' optimal search and stopping behavior, the firm in a match with initial surplus $z$ chooses a wage policy $w(\cdot,m)$ to maximize expected discounted profits:
\begin{equation}
\begin{aligned}
  V^F(z;m)
  &= \sup_{w \in \mathcal{W}}
    \mathbb{E}_z \Bigg[
      \int_0^\tau e^{-rt}
        \bigl( \Pi(Z_t,m) - w(Z_t,m) \bigr) \, dt
      + e^{-r\tau} V^V(m)
    \Bigg],
\end{aligned}
  \label{eq:firm_value}
\end{equation}

where $\Pi(z,m)$ is the \emph{log} flow revenue generated by a match of surplus $z$ when the aggregate state is $m$ (in the same log units as $w(z,m)$ and the processes $P_t$ and $w_t$ in Section~\ref{subsec:workers_firms}), $V^V(m)$ is the value of holding a vacant position in an economy with aggregate state $m$, and $\tau$ is the separation time implied by the worker's stopping rule and exogenous shocks. Thus $\Pi(z,m) - w(z,m)$ is the firm's log residual payoff.

The wage policy affects the worker's continuation value and hence endogenous separation and job-to-job mobility. We collect the contribution of this channel in an operator
\[
  H(z,V^F,m;w),
\]
which captures the effect of the wage schedule on the effective drift and killing rate of the match from the firm's perspective. A precise definition in terms of the joint dynamics of $(Z_t,1_{\{t<\tau\}})$ is given in Appendix~\ref{appendix:proofs}; for the main text it is enough to note that $H$ satisfies the regularity and monotonicity conditions in Assumption~\ref{ass:primitives}.

In a stationary environment, the firm's value function solves
\begin{equation}
\begin{aligned}
  r V^F(z;m)
  &= \sup_{w \in \mathcal{W}}
    \Bigl[
      \Pi(z,m) - w(z,m)
      + \mathcal{L}^Z_{a^\ast(z,m),m} V^F(z;m)
      - H\bigl(z,V^F,m;w\bigr)
    \Bigr],
\end{aligned}
  \label{eq:firm_HJB}
\end{equation}
where $a^\ast(z,m)$ is the worker's optimal search policy derived from \eqref{eq:worker_HJB_obstacle}. In the quantitative analysis we restrict attention to stationary Markov wage policies $w(z,m)$ that depend only on the current surplus and on the stationary distribution $m$. As discussed above, such policies can be interpreted either as posted wages or as the reduced form of fast wage renegotiation; making wages fully contingent on each outside offer would correspond to policies in which $w(z,m)$ tracks $V^U(m)$ very closely and would tend to compress wage dispersion.

\subsection{Stationary mean field equilibrium}\label{subsec:MFG_equilibrium_def}

The mean field framework closes the model by requiring consistency between individual optimization and the aggregate distribution. We focus on stationary mean field Nash equilibria.

\begin{definition}[Stationary mean field equilibrium]\label{def:MFG_equilibrium}
A stationary mean field equilibrium is a collection
\[
  \bigl(
    m,\,
    a^\ast(\cdot,m),\, \mathcal{S}^\ast(m),\,
    w^\ast(\cdot,m),\,
    V^W(\cdot;m),\, V^F(\cdot;m),\,
    V^U(m),\, V^V(m)
  \bigr)
\]
such that:
\begin{enumerate}
  \item \textbf{Worker optimality.}
  Given the stationary distribution $m$ and the wage policy $w^\ast(\cdot,m)$, the worker's search and stopping strategy $(a^\ast,\mathcal{S}^\ast)$ and the value function $V^W(\cdot;m)$ solve the worker's control-and-stopping problem~\eqref{eq:worker_value} and satisfy the HJB obstacle problem~\eqref{eq:worker_HJB_obstacle}. The stopping region $\mathcal{S}^\ast(m)$ is the set of surplus levels $z$ for which separation is optimal.

  \item \textbf{Firm optimality.}
  Given $m$ and workers' optimal strategy $(a^\ast,\mathcal{S}^\ast)$, the wage policy $w^\ast(\cdot,m)$ and the value function $V^F(\cdot;m)$ solve the firm's problem~\eqref{eq:firm_value} and satisfy the HJB equation~\eqref{eq:firm_HJB}.

  \item \textbf{Consistency of outside values.}
  The outside values $V^U(m)$ and $V^V(m)$ equal the expected discounted values obtained by workers and firms who are not currently matched but participate in the same stationary environment induced by $(m,a^\ast,w^\ast)$.

  \item \textbf{Stationarity of the distribution.}
  The distribution $m$ is invariant under the surplus dynamics~\eqref{eq:Z_SDE_general} induced by the controls $(a^\ast,w^\ast)$ and the separation rule $\mathcal{S}^\ast(m)$. Equivalently, $m$ solves the stationary Kolmogorov equation associated with~\eqref{eq:FP_general} with $\partial_t m_t = 0$, together with the boundary conditions implied by $\mathcal{S}^\ast(m)$ and the entry term $\Gamma(\cdot,m)$.
\end{enumerate}
\end{definition}

Conditions (i) and (ii) ensure that workers and firms are individually optimal given the conjectured stationary environment; condition (iii) guarantees that the outside values used in the HJB equations coincide with those generated by the same environment; and condition (iv) imposes the fixed-point requirement that the conjectured distribution $m$ be invariant under the induced dynamics. Thus Definition~\ref{def:MFG_equilibrium} is the stationary counterpart of the dynamic mean field Nash equilibrium of \citet{LasryLions2007} and \citet{CarmonaDelarue2018BookI}, and can be interpreted as the steady state of a large random-matching economy in which each match is small.

In Section~\ref{sec:equilibriumCharacterization} we introduce the dynamic mean field Nash equilibrium formally (Definition~\ref{def:dynamic_MFNE}), relate it to Definition~\ref{def:MFG_equilibrium}, characterize the free boundary $z^\ast(m)$ that governs separation, and describe the equilibrium wage distribution and the joint stationary distribution of match surpluses and wages implied by the mean field game.

\section{Mean Field Equilibrium and Separation Rules}\label{sec:equilibriumCharacterization}

This section formalizes the mean field Nash equilibrium of the dynamic game in Section~\ref{sec:MFG}, specializes to stationary environments, and spells out the associated Hamilton-Jacobi-Bellman and Kolmogorov equations. We then characterize the worker's separation behaviour as a free boundary in the surplus state and define equilibrium wage dispersion as a functional of the stationary distribution and the wage policy. These objects form the theoretical backbone of the quantitative analysis in Sections~\ref{sec:quantitative} and~\ref{sec:results}.

Throughout, we write $\mathcal{P}(\mathbb{R})$ for the set of Borel probability measures on $\mathbb{R}$ equipped with the topology of weak convergence, and $\mathcal{P}_2(\mathbb{R})$ for the subset of measures with finite second moment. The existence and uniqueness analysis in Section~\ref{sec:existenceUniqueness} is conducted on $\mathcal{P}_2(\mathbb{R})$, but in this section we write $m$ for generic elements of $\mathcal{P}(\mathbb{R})$ whenever no confusion can arise.

As also stated earlier, a mean field equilibrium can be actually read as the rational-expectations steady state of a very large random-matching labour market. Each infinitesimal worker-firm match takes the cross-sectional distribution $m$ of surplus as given when choosing search effort and wage policies, while in equilibrium that same distribution must be reproduced by the stochastic evolution of match-specific surplus under those policies.

\subsection{Dynamic mean field Nash equilibrium}

For completeness, we start from a general, possible time-dependent, notion of mean field Nash equilibrium in the sense of \citet{LasryLions2007} and \citet{CarmonaDelarue2018BookI}. The stationary concept used in the rest of the paper is obtained as a special case.

Let $t \mapsto m_t \in \mathcal{P}(\mathbb{R})$ be a measurable flow of cross-sectional distributions of match surpluses, and let
\[
  w_t \colon \mathbb{R} \to \mathbb{R},
  \qquad
  a_t \colon \mathbb{R} \to A
\]
be time-dependent Markov feedback policies for wages and search effort. Conditional on $(m_t)$ and $(a_t,w_t)$, an individual surplus $Z_t$ follows the controlled diffusion in~\eqref{eq:Z_SDE_general}, with separation and re-entry as described in Section~\ref{sec:MFG}.

\begin{definition}[Dynamic mean field Nash equilibrium]\label{def:dynamic_MFNE}
A \emph{dynamic mean field Nash equilibrium} is a collection
\[
  \bigl(
    (m_t)_{t \geq 0},
    a_t^\ast(\cdot), w_t^\ast(\cdot),
    V^W(t,\cdot), V^F(t,\cdot),
    V^U(t), V^V(t)
  \bigr)
\]
such that:
\begin{enumerate}
  \item \textbf{Worker optimality.}
  For each initial time $t_0$ and surplus $z$, given the flow $(m_t)_{t \geq t_0}$ and the wage policy $(w_t^\ast)_{t \geq t_0}$, the worker's search-and-stopping problem with state $(t_0,z)$ has value $V^W(t_0,z)$. It is solved by admissible controls $(a_s^\ast,\tau^\ast)$ that satisfy the time-dependent HJB obstacle problem associated with~\eqref{eq:Z_SDE_general}, with running payoff $w_s^\ast(Z_s) - c(a_s)$ and outside value $V^U(s)$.

  \item \textbf{Firm optimality.}
  For each $(t_0,z)$, given $(m_t)$ and the worker's optimal strategy, the firm's problem in a match with state $(t_0,z)$ has value $V^F(t_0,z)$. It is solved by the wage policy $(w_t^\ast)_{t \geq t_0}$, which satisfies the time-dependent HJB equation with running payoff $\Pi(Z_t,m_t) - w_t^\ast(Z_t)$ and outside value $V^V(t)$.

  \item \textbf{Consistency of outside values.}
  For each $t$, the outside values $V^U(t)$ and $V^V(t)$ equal the expected discounted values obtained by workers and firms who are not currently matched but participate in the same environment induced by $\bigl((m_s)_{s \geq t},a_s^\ast,w_s^\ast\bigr)$.

  \item \textbf{Law of motion for $(m_t)$.}
  The flow $(m_t)_{t \geq 0}$ is the law of $Z^i_t$ when $i$ is drawn uniformly from the set of active matches. It satisfies the Kolmogorov forward equation associated with the controlled diffusion~\eqref{eq:Z_SDE_general}, the optimal controls $a_t^\ast(\cdot)$ and $w_t^\ast(\cdot)$, and the induced separation and entry rates.
\end{enumerate}
\end{definition}

In simple words, a dynamic mean field Nash equilibrium is a path of distributions $(m_t)$ and associated Markov policies such that (i) workers and firms choose search, stopping, and wage policies optimally given $(m_t)$, and (ii) $(m_t)$ coincides with the cross-sectional law generated by those optimal policies. Items (1)-(4) are the usual best-response and consistency conditions of a competitive rational-expectations equilibrium, written in continuous time and with the cross-sectional distribution as the aggregate state variable.

We restrict attention to Markovian feedback strategies that depend only on the current time, the current surplus, and the current distribution $m_t$. In the one-dimensional, time-homogeneous environment considered here, this restriction is without loss for the equilibria of interest; see \citet{CardaliaguetPorretta2020IntroMFG} and \citet{CarmonaDelarue2018BookI}. It reflects that, in a large anonymous labour market, history beyond the current surplus and the aggregate state does not convey additional payoff-relevant information.

\subsection{Stationary equilibrium}\label{subsec:stationary_equilibrium}

Our focus is on stationary mean field equilibria. This is empirically natural for the medium-run wage and mobility patterns we target and theoretically convenient, because stationarity yields a sharp characterization of the separation rule and a tractable numerical system.

A stationary mean field equilibrium is a dynamic equilibrium in the sense of Definition~\ref{def:dynamic_MFNE} with the following additional properties:
\begin{itemize}
  \item The aggregate distribution is time invariant:
  \[
    m_t \equiv m \in \mathcal{P}_2(\mathbb{R})
    \quad \text{for all } t \geq 0.
  \]
  \item Worker and firm policies are time-homogeneous and Markovian in the surplus and in $m$:
  \[
    a_t^\ast(z) \equiv a^\ast(z,m),
    \qquad
    w_t^\ast(z) \equiv w^\ast(z,m).
  \]
  \item The value functions do not depend explicitly on time:
  \[
    V^W(t,z) \equiv V^W(z;m), \quad
    V^F(t,z) \equiv V^F(z;m), \quad
    V^U(t) \equiv V^U(m), \quad
    V^V(t) \equiv V^V(m).
  \]
\end{itemize}

Thus, in a stationary equilibrium each match faces a time-invariant environment summarized by $m$, chooses time-invariant Markov policies, and the cross-sectional distribution generated by these policies is itself constant over time. This is precisely the notion formalized in Definition~\ref{def:MFG_equilibrium} in Section~\ref{subsec:MFG_equilibrium_def}. In a stationary equilibrium the worker and firm problems reduce to the time-independent HJB equations~\eqref{eq:worker_HJB_obstacle} and~\eqref{eq:firm_HJB}, and the distribution $m$ solves the stationary Kolmogorov equation associated with~\eqref{eq:FP_general} with $\partial_t m_t = 0$.

In particular, for a given stationary distribution $m$ and policies $a^\ast(\cdot,m)$ and $w^\ast(\cdot,m)$, the worker's continuation value $V^W(\cdot;m)$ solves the HJB obstacle problem
\begin{equation}
\begin{aligned}
  0 
  = \max \Bigl\{
    & V^W(z;m) - V^U(m), \\
    & r V^W(z;m)
      - \sup_{a \in A} \Bigl[
          w^\ast(z,m) - c(a)
          + \mathcal{L}^Z_{a,m} V^W(z;m)
          + G\bigl(z,V^W,m;a\bigr)
        \Bigr]
    \Bigr\},
\end{aligned}
\label{eq:worker_HJB_stationary}
\end{equation}
where $\mathcal{L}^Z_{a,m}$ is the infinitesimal generator of~\eqref{eq:Z_SDE_general} under control $a$ and aggregate state $m$, as defined in Section~\ref{subsec:individual_HJB}, and $G$ captures the expected gains from on-the-job search. The first term inside the maximum compares the continuation value in the current job to the outside option $V^U(m)$; the second term describes the usual flow-payoff plus discounted-drift trade-off when the worker chooses to continue searching on the job rather than separating immediately.

The firm's stationary value function $V^F(\cdot;m)$ solves
\begin{equation}
  r V^F(z;m)
  = \sup_{w \in \mathcal{W}}
      \Bigl[
        \Pi(z,m) - w(z,m)
        + \mathcal{L}^Z_{a^\ast(z,m),m} V^F(z;m)
        - H\bigl(z,V^F,m;w\bigr)
      \Bigr],
  \label{eq:firm_HJB_stationary}
\end{equation}
where $H$ summarizes how the wage policy affects the worker's separation behaviour and thus effective match duration, as in Section~\ref{subsec:individual_HJB} and Appendix~\ref{appendix:proofs}. The firm trades off paying a higher wage today, which reduces the current surplus $\Pi-w$ but tends to prolong the match and raise future profits through $H$, against the option of keeping wages low and letting low-surplus matches separate more quickly.

Given these policies, the stationary distribution $m$ solves the Kolmogorov equation
\begin{equation}
\begin{aligned}
  0
  &= -\partial_z\!\bigl[
       \mu\bigl(z,a^\ast(z,m),m\bigr)\, m(z)
     \bigr]
     + \tfrac{1}{2}\,\partial_{zz}\!\bigl[
       \sigma^2(z,m)\, m(z)
     \bigr] \\[0.5em]
  &\quad - q(z,m)\, m(z)
     + \Gamma(z,m),
\end{aligned}
\label{eq:FP_stationary_equilibrium}
\end{equation}
where $q(z,m)$ is the separation intensity implied by the optimal stopping rule and $\Gamma(z,m)$ is the entry term. The first line captures drift and diffusion of surplus among ongoing matches, the term $q(z,m)m(z)$ removes mass because of separation, and $\Gamma(z,m)$ injects new matches at different surplus levels. The coupling is nonlocal in the sense of \citet{LasryLions2007}: the coefficients $\mu$, $\sigma$, $w$, $\Pi$, $\lambda$, and $\Gamma$ depend on $m$ through aggregate statistics (such as mean surplus, employment rates, and wage moments) rather than pointwise evaluations of $m$.

Equations~\eqref{eq:worker_HJB_stationary}-\eqref{eq:FP_stationary_equilibrium} constitute the canonical HJB-Kolmogorov system that underpins the existence and uniqueness analysis in Section~\ref{sec:existenceUniqueness}, in line with the framework of \citet{CardaliaguetPorretta2020IntroMFG} and \citet{CarmonaDelarue2018BookI}. Readers unfamiliar with mean field games can therefore think of a stationary equilibrium as a steady state of this coupled ``backward'' HJB (optimal behaviour) and ``forward'' Kolmogorov (distribution dynamics) system.

\subsection{Free-boundary separation rule}\label{subsec:free_boundary}

A distinctive feature of our environment is that separation is an optimal stopping decision at the match level. In a stationary mean field equilibrium this leads to a free-boundary characterization of the worker's separation rule in terms of the surplus state, in the spirit of \citet{Jovanovic1979JobMatching} or \citet{BuhaiTeulings2014}.

For a fixed stationary distribution $m$ and wage policy $w^\ast(\cdot,m)$, define the surplus of staying in the match relative to immediate separation as
\[
  \Delta(z;m) := V^W(z;m) - V^U(m).
\]
Intuitively, $\Delta(z;m)$ measures the gain from keeping the curent job when the match-specific surplus is $z$, taking as given the stationary outside option generated by $m$. Under the monotonicity and single-crossing conditions on the primitives stated in Sections~\ref{sec:MFG} and~\ref{sec:existenceUniqueness}, $\Delta(\cdot;m)$ is continuous and strictly increasing in $z$: higher-surplus matches yield higher current wages and better continuation values, so the benefit from staying rises with $z$. For sufficiently low surplus levels, continuation is dominated by the outside option, while for sufficiently high surplus levels the continuation value strictly exceeds the outside value.

\begin{proposition}[Free-boundary separation rule]\label{prop:free_boundary}
Fix a stationary distribution $m$ and a wage policy $w^\ast(\cdot,m)$, and suppose Assumption~\ref{ass:primitives} holds. Then there exists a unique threshold $z^\ast(m) \in \mathbb{R}$ such that:
\begin{enumerate}
  \item The stopping region is an interval of the form
  \[
    \mathcal{S}(m) = (-\infty,z^\ast(m)],
  \]
  and the continuation region is $(z^\ast(m),\infty)$.

  \item The value function satisfies the value-matching condition
  \[
    V^W\bigl(z^\ast(m);m\bigr) = V^U(m)
  \]
  and the smooth-fit condition
  \[
    \partial_z V^W\bigl(z^\ast(m);m\bigr) = 0.
  \]

  \item On $(z^\ast(m),\infty)$, $V^W(\cdot;m)$ solves the linear elliptic HJB equation obtained from~\eqref{eq:worker_HJB_stationary} with the stopping constraint inactive, while on $(-\infty,z^\ast(m)]$ the value is constant and equal to $V^U(m)$.
\end{enumerate}
\end{proposition}

The proof, given in Appendix~\ref{appendix:proofs}, exploits the one-dimensional diffusion structure, the regularity of $\mu$ and $\sigma$, and the convexity properties implied by Assumption~\ref{ass:primitives}. In particular, smooth fit follows from standard arguments for optimal stopping of one-dimensional diffusions; see, for example, \citet{PeskirShiryaev2006OptimalStopping}. Proposition~\ref{prop:free_boundary} says that in any stationay environment with given $m$ and $w^\ast$, there is a unique surplus level below which jobs are destroyed and above which they are retained, generalizing the partial-equilibrium efficient separation rule in \citet{BuhaiTeulings2014} to a setting with endogenous outside options and on-the-job search.

Separation in the Kolmogorov equation~\eqref{eq:FP_stationary_equilibrium} can be represented in two equivalent ways:
\begin{itemize}
  \item In the pure hitting-time benchmark that mirrors \citet{BuhaiTeulings2014}, separation occurs only when $Z_t$ hits $z^\ast(m)$ from above. In this case we set $q(z,m) \equiv 0$ and impose an absorbing boundary at $z^\ast(m)$, so that separation is captured by the probability flux through the boundary. The associated hitting-time distribution generates the inverse-Gaussian tenure hazards discussed in Sections~\ref{sec:environment} and~\ref{sec:quantitative}.

  \item In more general specifications, $q(z,m)$ can capture additional Poisson separation inside the stopping region, for example due to exogenous destruction shocks. Then $q(z,m) > 0$ for $z \leq z^\ast(m)$, while separation at the boundary is still governed by the absorbing condition at $z^\ast(m)$.
\end{itemize}

In the baseline calibration, separation is driven entirely by the hitting time of $Z_t$ at the endogenously determined threshold $z^\ast(m)$, which is the general-equilibrium analogue of the stopping rule in the benchmark model of \citet{BuhaiTeulings2014}. Section~\ref{sec:quantitative} shows how the location of this free boundary responds to changes in primitives such as surplus volatility and workers' bargaining strength.

\subsection{Equilibrium wage distribution and dispersion}\label{subsec:wage_dispersion}

A stationary mean field equilibrium $(m,a^\ast,w^\ast)$ delivers an endogenous joint distribution of match surpluses and wages. The cross-sectional distribution of wages is the pushforward of $m$ through the wage policy:
\[
  \mu_w(\cdot;m)
  := m \circ w^\ast(\cdot,m)^{-1}.
\]
Thus, for any Borel set $B \subset \mathbb{R}$,
\[
  \mu_w(B;m)
  = m\bigl( \{ z \in \mathbb{R} \colon w^\ast(z,m) \in B \} \bigr).
\]

We measure equilibrium wage dispersion in the stationary economy by the variance of wages under $\mu_w(\cdot;m)$, namely
\begin{equation}
  D(m,w^\ast)
  := \operatorname{Var}_{\mu_w(\cdot;m)}(w)
  = \int_{\mathbb{R}}
      \bigl( w^\ast(z,m) - \bar w(m) \bigr)^2
    \, m(dz),
  \label{eq:dispersion_def}
\end{equation}
where $\bar w(m)$ is the mean wage,
\[
  \bar w(m)
  := \int_{\mathbb{R}} w^\ast(z,m)\, m(dz).
\]

Expression~\eqref{eq:dispersion_def} is the structural object we decompose quantitatively in Section~\ref{sec:results}. It is entirely determined by the stationary mean field equilibrium: both $m$ and $w^\ast$ are equilibrium objects that respond endogenously to search behaviour, the job ladder, and firms' wage policies. Changes in primitives affect dispersion through two main channels:
\begin{enumerate}
  \item changes in the stationary surplus distribution $m$, driven by search, separation, and re-entry, as captured by the Kolmogorov equation~\eqref{eq:FP_stationary_equilibrium};

  \item changes in the wage policy $w^\ast(\cdot,m)$, driven by the firm's HJB problem~\eqref{eq:firm_HJB_stationary} and its interaction with workers' outside options and search intensities.
\end{enumerate}

The existence and uniqueness results in Section~\ref{sec:existenceUniqueness} and Appendix~\ref{appendix:proofs} ensure that, under our standing assumptions, $m$ and $w^\ast$ are well defined and depend continuously on the primitives. In particular, the Lasry-Lions monotonicity condition rules out multiple stationary job ladders that would differ only through self-fulfilling beliefs about the cross-sectional distribution of surplus and wages. This underpins the comparative statics for $z^\ast(m)$ and for the dispersion measure $D(m,w^\ast)$ reported in Sections~\ref{sec:quantitative} and~\ref{sec:results}, and guarantees that our variance decompositions and policy experiments are conducted with respect to a unique stationary equilibrium. NB. In later sections, we shall sometimes suppress the explicit dependence on $m$ and write $w(z)$ or $D$ for these equilibrium objects evaluated at the calibrated stationary distribution $m^\ast$.

\section{Existence and Uniqueness of Stationary Equilibrium}\label{sec:existenceUniqueness}

This section establishes conditions under which a stationary mean field equilibrium, as defined in Section~\ref{subsec:MFG_equilibrium_def}, exists and is unique. We also show how these conditions deliver signed comparative statics for the equilibrium separation threshold and for wage dispersion. Complete proofs and functional-analytic details are collected in Appendix~\ref{appendix:proofs}.

Throughout, equilibrium objects are functions of the current match surplus $z \in \R$ and the cross-sectional distribution $m$ of active matches. Remark that when there is no danger of misinterpretation, we suppress the $m$-argument for brevity, but it is always understood that $V^W$, $V^F$, $w$, and the free boundary $z^\ast$ depend on the stationary distribution.

\subsection{Assumptions on primitives}\label{subsec:assumptions_existence}

We work on the space $\cP_2(\R)$ of probability measures on $\R$ with finite second moment. The individual state is the surplus $z \in \R$ and the aggregate state is a stationary distribution $m \in \cP_2(\R)$.

We maintain the structure of surplus dynamics, controls, and payoffs introduced in Sections~\ref{sec:environment} and~\ref{sec:MFG}. The next assumption gathers standard regularity, growth, and compactness conditions ensuring that individual control problems are well posed and that, for a fixed distribution $m$, the induced surplus diffusion with killing and re-entry admits a unique invariant law.

\begin{assumption}[Primitives]\label{ass:primitives}
The following conditions hold.
\begin{enumerate}
  \item \textbf{Surplus dynamics.}
  The drift and volatility functions
  \[
    \mu : \R \times A \times \cP_2(\R) \to \R,
    \qquad
    \sigma : \R \times \cP_2(\R) \to (0,\infty)
  \]
  are Borel measurable, jointly continuous, and satisfy, for some constant $L > 0$
  and all $z,z' \in \R$, $a,a' \in A$, and $m,m' \in \cP_2(\R)$,
  \[
    \bigl|\mu(z,a,m) - \mu(z',a',m')\bigr|
      + \bigl|\sigma(z,m) - \sigma(z',m')\bigr|
      \leq L\bigl( |z - z'| + |a - a'| + W_2(m,m') \bigr),
  \]
  and
  \[
    |\mu(z,a,m)| + |\sigma(z,m)|
      \leq L \bigl(1 + |z|\bigr)
  \]
  uniformly in $(a,m)$. In addition, the drift is weakly increasing in the surplus state:
  \[
    z_1 \le z_2
    \;\Longrightarrow\;
    \mu(z_1,a,m) \le \mu(z_2,a,m)
    \quad\text{for all } (a,m).
  \]
  The volatility is uniformly elliptic:
  \[
    0 < \underline{\sigma} \leq \sigma(z,m) \leq \overline{\sigma} < \infty
    \quad \text{for all } (z,m).
  \]

  \item \textbf{Payoffs and search costs.}
  The wage function $w(z,m)$, the firm's revenue function $\Pi(z,m)$, and the outside values $V^U(m)$ and $V^V(m)$ are continuous and bounded from below. There exists $L > 0$ (not necessarily the same as above) such that, for all $z,z' \in \R$ and $m,m' \in \cP_2(\R)$,
  \[
    |w(z,m) - w(z',m')|
      + |\Pi(z,m) - \Pi(z',m')|
      \leq L \bigl( |z - z'| + W_2(m,m') \bigr).
  \]
  The search cost $c: A \to \R_+$ is continuously differentiable and strictly convex, with $c(0) = 0$, $c'(0) = 0$, and $c'(a) \to \infty$ as $a$ approaches the upper endpoint of $A$.

  \item \textbf{Offer arrival and gains from search.}
  The arrival intensity $\lambda(a,m)$ is continuous on $A \times \cP_2(\R)$, strictly increasing in $a$, and satisfies
  \[
    0 < \underline{\lambda} \leq \lambda(a,m) \leq \overline{\lambda} < \infty
    \quad \text{for all } (a,m).
  \]
  The search-gain operator $G(z,V^W,m;a)$ appearing in~\eqref{eq:worker_HJB_obstacle} is continuous and Lipschitz in $(z,V^W,m)$ uniformly in $a \in A$ (with respect to the sup norm on $V^W$). For each $(z,m)$ the set of maximizers of the HJB operator in~\eqref{eq:worker_HJB_obstacle} is nonempty, convex, and compact.

  \item \textbf{Boundary behaviour and re-entry.}
  When a match separates at the free boundary $z^\ast(m)$, the worker and firm move to outside states with finite values $V^U(m)$ and $V^V(m)$ that depend continuously on $m$ in the $W_2$ topology. New matches are created according to a re-entry distribution
  \[
    \nu(m) \in \cP_2(\R),
  \]
  which depends continuously on $m$ and has uniformly bounded second moment.
\end{enumerate}
\end{assumption}

Assumption~\ref{ass:primitives} specializes the usual standing assumptions in the continuous-time MFG literature\footnote{See, for example, \citet{LasryLions2007}, \citet{CarmonaDelarue2018BookI}, and \citet{CardaliaguetPorretta2020IntroMFG}.} to our one-dimensional surplus state with killing at a free boundary and regeneration of new matches. It ensures that, for any fixed stationary distribution $m$, the worker and firm problems are well posed and the induced surplus diffusion with killing and re-entry is ergodic with a unique invariant law.\footnote{For one-dimensional diffusions with killing and re-entry, existence, uniqueness, and continuity of invariant measures follow from standard arguments as in \citet{StroockVaradhan2006MultidimensionalDiffusions}.}

In simple words, part~(1) keeps match-surplus dynamics well behaved: shocks have finite variance and depend smoothly on the current surplus, search intensity, and aggregate distribution, without explosive drifts. Parts~(2) and~(3) require that wages, profits, and offer arrival rates respond smoothly to $z$ and to aggregate conditions $m$, and that the search problem has an interior solution: higher effort raises the arrival rate but faces increasing convex costs. Part~(4) ensures that separated matches re-enter the economy with a well-defined distribution of starting surpluses and finite second moments, so that the aggregate state $m$ is stable.

To study uniqueness, we impose a Lasry-Lions-type monotonicity condition on how the aggregate distribution $m$ enters payoffs and matching.

\begin{assumption}[Lasry and Lions monotonicity]\label{ass:monotonicity}
For all $m^1,m^2 \in \cP_2(\R)$,
\[
  \int_{\R}
    \bigl( w(z,m^1) - w(z,m^2) \bigr)
    \bigl( m^1 - m^2 \bigr)(dz)
  \leq 0,
\]
and
\[
  \int_{\R}
    \bigl( \Pi(z,m^1) - \Pi(z,m^2) \bigr)
    \bigl( m^1 - m^2 \bigr)(dz)
  \leq 0,
\]
with equality in both expressions if and only if $m^1 = m^2$. Moreover, the dependence of the arrival intensity $\lambda(a,m)$ and the re-entry distribution $\nu(m)$ on $m$ is monotone in the same sense.

For the uniqueness result in Theorem~\ref{thm:uniqueness} we additionally assume that
the surplus drift and volatility do not depend on the aggregate distribution:
\[
  \mu(z,a,m)\equiv\mu(z,a)
  \quad\text{and}\quad
  \sigma(z,m)\equiv\sigma(z)
  \qquad\text{for all }(z,a,m).
\]
Thus all $m$–dependence in the worker Hamiltonian comes from the payoff and matching
terms $(w,\Pi,\lambda,\nu,V^U)$, which satisfy the Lasry–Lions monotonicity
condition above. In particular, the calibrated specification in Section~\ref{sec:quantitative}
satisfies Assumption~\ref{ass:monotonicity}.
\end{assumption}

Assumption~\ref{ass:monotonicity} is the continuous-state analogue of the Lasry-Lions monotonicity condition used to obtain uniqueness in MFGs. Intuitively, it requires that shifting the cross-sectional distribution $m$ in a way that raises surplus on average does not create incentives that feed back in the same direction and generate multiple self-fufilling job ladders. In our calibration $w(z,m)$ and $\Pi(z,m)$ depend on $m$ only through a small set of aggregate wage and employment statistics, so these integral inequalities can be checked directly. What this condition does is to rule out ``positive feedback loops'' in which a high-surplus distribution pushes firms to post even higher wages simply because they expect others to do so, thereby sustaining multiple stationary wage ladders.

\subsection{Existence of stationary equilibrium}\label{subsec:existence}

We now establish existence of a stationary equilibrium as defined in Section~\ref{subsec:MFG_equilibrium_def}. The proof follows the standard fixed-point approach in the MFG literature, adapted to our free-boundary structure, and is given in Appendix~\ref{appendix:proofs}.

\begin{theorem}[Existence of stationary mean field equilibrium]\label{thm:existence}
Suppose Assumption~\ref{ass:primitives} holds. Then there exists at least one stationary mean field equilibrium
\[
  \Bigl(
    m^\ast,
    a^\ast(\cdot,m^\ast), \mathcal{S}^\ast(m^\ast),
    w^\ast(\cdot,m^\ast),
    V^W(\cdot;m^\ast), V^F(\cdot;m^\ast),
    V^U(m^\ast), V^V(m^\ast)
  \Bigr)
\]
in the sense of Section~\ref{subsec:MFG_equilibrium_def}. The optimal separation rule of workers is characterized by a free boundary $z^\ast(m^\ast)$, as in Proposition~\ref{prop:free_boundary}.
\end{theorem}

\begin{proof}
We summarize the three main steps; Appendix~\ref{appendix:proofs} provides a complete argument and connects our fixed-point map to the finite-state operators in \citet{GomesEtAl2010FiniteStateDiscrete,GomesEtAl2024FiniteStateContinuous}.

\emph{Step 1 (best responses for a fixed distribution).}
Fix a candidate stationary distribution $m \in \cP_2(\R)$. Under Assumption~\ref{ass:primitives}, the worker's HJB obstacle problem~\eqref{eq:worker_HJB_obstacle} is a well-posed variational inequality on $\R$ with a unique continuous solution $V^W(\cdot;m)$ of at most linear growth. The worker's optimal strategy is characterized by a Markov search policy $a^\ast(\cdot,m)$ and a free boundary $z^\ast(m)$ separating continuation and separation; see Proposition~\ref{prop:free_boundary}.

Given $m$ and the worker strategy, the firm's HJB equation~\eqref{eq:firm_HJB} is linear in $V^F$ and concave in the wage policy $w(\cdot,m)$. The same regularity and growth conditions deliver a unique solution $(V^F(\cdot;m),w^\ast(\cdot,m))$.

\emph{Step 2 (invariant distribution induced by best responses).}
Given $m$ and the associated best-response policies $(a^\ast(\cdot,m),w^\ast(\cdot,m),z^\ast(m))$, the surplus process $Z_t$ in an active match follows the controlled diffusion~\eqref{eq:Z_SDE_general} with killing at $z^\ast(m)$ and re-entry according to $\nu(m)$. The stationary Fokker-Planck equation associated with~\eqref{eq:FP_general} is linear in the unknown density, and Assumption~\ref{ass:primitives} guarantees that there exists a unique invariant probability measure, denoted $\Gamma(m) \in \cP_2(\R)$. This invariant measure is the stationary law of $Z_t$ across active matches.

\emph{Step 3 (fixed point in the space of distributions).}
The composition
\[
  \Gamma: \cP_2(\R) \to \cP_2(\R),
  \qquad
  m \mapsto \Gamma(m),
\]
maps a candidate stationary distribution into the invariant distribution induced by the corresponding best-response policies. Assumption~\ref{ass:primitives} implies that $\Gamma$ maps a convex, compact subset of $\cP_2(\R)$ (for instance, measures with uniformly bounded second moment) into itself and is continuous in the $W_2$ metric. By Schauder's fixed-point theorem, there exists $m^\ast$ with $\Gamma(m^\ast) = m^\ast$. The associated policies and value functions form a stationary mean field equilibrium.
\end{proof}

The operator $\Gamma$ is the infinite-state analogue of the finite-state MFG operators studied by \citet{GomesEtAl2010FiniteStateDiscrete,GomesEtAl2024FiniteStateContinuous}. To restate what was said earlier in terms of this analogy, a stationary mean field equilibrium is therefore a \emph{rational-expectations fixed point}: workers and firms take the cross-sectional distribution $m$ as given, choose optimal search, separation, and wage policies, and the resulting stochastic dynamics generate exactly the same invariant distribution $m^\ast = \Gamma(m^\ast)$. Our numerical scheme in Section~\ref{sec:quantitative} and Appendix~\ref{appendix:numerical} can be interpreted as computing approximate fixed points of $\Gamma$ via an Achdou-Capuzzo-Dolcetta style discretization of the coupled HJB-Fokker-Planck system.\footnote{See \citet{AchdouCapuzzoDolcetta2010NumericsMFG} and \citet{AchdouEtAl2022Restud} for closely related numerical schemes in continuous-time MFGs and heterogeneous-agent models.}

\subsection{Uniqueness and comparative statics}\label{subsec:uniqueness_CS}

We now provide conditions under which the stationary equilibrium is unique and derive comparative statics for the separation threshold and wage dispersion.

\begin{theorem}[Uniqueness under monotonicity]\label{thm:uniqueness}
Suppose Assumptions~\ref{ass:primitives} and~\ref{ass:monotonicity} hold. Then the stationary mean field equilibrium in Theorem~\ref{thm:existence} is unique: there is a unique stationary distribution $m^\ast$ and associated policies $(a^\ast(\cdot,m^\ast),\mathcal{S}^\ast(m^\ast),w^\ast(\cdot,m^\ast))$.
\end{theorem}

\begin{proof}
The proof follows the standard Lasry-Lions monotonicity argument. Suppose there are two stationary equilibria with distributions $m^1,m^2$ and worker value functions $V^{W,1},V^{W,2}$. Subtracting the HJB equations for $V^{W,1}$ and $V^{W,2}$ and integrating the difference against $m^1 - m^2$ yields an identity of the form
\[
  r \int_{\R}
    \bigl( V^{W,1}(z;m^1) - V^{W,2}(z;m^2) \bigr)
    \bigl( m^1 - m^2 \bigr)(dz)
  = \mathcal{T}_1 + \mathcal{T}_2,
\]
where $\mathcal{T}_1$ collects terms coming from the coupling via $w,\lambda,\nu$ and $\mathcal{T}_2$ collects the diffusion and search-cost terms. Convexity of $c$ and the quadratic form associated with the generator imply $\mathcal{T}_2 \leq 0$, while Assumption~\ref{ass:monotonicity} implies $\mathcal{T}_1 \leq 0$, with equality if and only if $m^1 = m^2$.

Repeating the argument with the roles of $(1,2)$ reversed shows that the left-hand side must also be nonnegative. Hence the integral is zero and both inequalities are equalities. By Assumption~\ref{ass:monotonicity}, this implies $m^1 = m^2$. Given $m^\ast$, uniqueness of the associated policies and free boundary follows from uniqueness of the solutions to the HJB and free-boundary problems for a fixed $m$.
\end{proof}

Of course, many frictional search models admit multiple stationary equilibria, for example when workers' search incentives or firms' wage posting strategies respond in a strongly reinforcing way to changes in beliefs about the wage distribution. Theorem~\ref{thm:uniqueness} shows that under the monotonicity structure in Assumption~\ref{ass:monotonicity} such coordination problems cannot arise: higher outside options in the aggregate do not induce offsetting changes in wages or offer arrival rates that would support another, self-fulfilling job ladder. This uniqueness is important for our quantitative analysis in Sections~\ref{sec:quantitative} and~\ref{sec:results}: it guarantees that the wage and mobility patterns reported there correspond to a single, well-defined stationary equilibrium, rather than to an arbitrary equilibrium selection.

For comparative statics, we introduce a regularity condition on how primitives vary with parameters.

\begin{assumption}[Comparative statics]\label{ass:comparative_statics}
Let $\theta \in \Theta \subset \R$ denote a parameter. The primitives $(\mu_\theta,\sigma_\theta,w_\theta,\Pi_\theta,\lambda_\theta,\nu_\theta)$ satisfy Assumptions~\ref{ass:primitives} and~\ref{ass:monotonicity} for each $\theta$, and:
\begin{enumerate}
  \item The volatility $\sigma_\theta(z,m)$ is increasing in $\theta$ in the sense that, for every $(z,m)$ and $\theta' > \theta$,
  \[
    \sigma_{\theta'}^2(z,m) \ge \sigma_\theta^2(z,m),
  \]
  while the drift $\mu_\theta(z,a,m)$ does not depend on $\theta$.

  \item The wage policy has the form
  \[
    w_\theta(z,m)
    = \alpha_\theta(z,m) \, \Pi_\theta(z,m)
      + \bigl(1 - \alpha_\theta(z,m)\bigr) \,\bar w_\theta(m),
  \]
  where $\bar w_\theta(m)$ is the average wage under $(\theta,m)$ and
  $\alpha_\theta(z,m) \in [0,1]$ satisfies:
  \begin{enumerate}
    \item[(a)] for each $m$, the map $z \mapsto \alpha_\theta(z,m)$ is
    nondecreasing and the map $\theta \mapsto \alpha_\theta(z,m)$ is
    nondecreasing; and
    \item[(b)] for each $m$, $\alpha_\theta(\cdot,m)$ has increasing
    differences in $(z,\theta)$ in the sense of Topkis: for all
    $z' > z$ and $\theta' > \theta$,
    \[
      \bigl[ \alpha_{\theta'}(z',m) - \alpha_{\theta'}(z,m) \bigr]
      \;\ge\;
      \bigl[ \alpha_{\theta}(z',m) - \alpha_{\theta}(z,m) \bigr].
    \]
  \end{enumerate}
  In particular, for each fixed $m$ the wage schedule
  $z \mapsto w_\theta(z,m)$ has increasing differences in $(z,\theta)$
  and becomes weakly steeper in $z$ as $\theta$ increases.
\end{enumerate}
\end{assumption}

Assumption~\ref{ass:comparative_statics} isolates two salient comparative statics that we emphasize in Section~\ref{sec:results}. Part~(1) describes an increase in the volatility of match-specific shocks at fixed drift, as in our experiments that scale up the Brownian variance of surplus. Part~(2) captures an increase in the sensitivity of wages to match surplus, e.g., stronger worker bargaining power or more aggressive wage posting, by making the weight $\alpha_\theta$ on match-specific revenue larger and more steeply increasing in $z$.

In the linear-sharing specification used in Section~\ref{sec:quantitative},
$\alpha_\theta(z,m)$ is affine in $\theta$ and nondecreasing in $z$, so
condition~(b) holds automatically and $w_\theta(z,m)$ indeed becomes
steeper in $z$ as $\theta$ increases.

\medskip
\noindent
To apply the volatility comparative statics in Proposition~\ref{prop:comparative_statics},
we impose a simple convexity condition on the worker’s primitive payoffs.

\begin{assumption}[Convexity of continuation payoffs for volatility comparative statics]\label{ass:convexity_CS}
In addition to Assumptions~\ref{ass:primitives}-\ref{ass:comparative_statics}, suppose that, for every
parameter value $\theta \in \Theta$ and every stationary distribution $m \in \cP_2(\R)$, the following
conditions hold.
\begin{enumerate}
  \item For each $(z,m,\theta)$ the wage function $w_\theta(z,m)$ is nondecreasing and convex in $z$.

  \item For each search intensity $a \in A$ and each admissible value function $V^W$, the search-gain
  operator $G(z,V^W,m;a)$ preserves convexity in the sense that whenever $V^W(\cdot;m,\theta)$ is
  convex in $z$, the mapping
  \[
    z \;\longmapsto\; G(z,V^W,m;a)
  \]
  is convex in $z$ as well. In particular, this condition is satisfied whenever $G$ is affine in $V^W$
  with a nonnegative kernel in the surplus state, as in the parametric specification used in
  Section~\ref{sec:quantitative}.
\end{enumerate}
Under these conditions, for each $(m,\theta)$ the worker's value function $V^W_\theta(\cdot;m)$ solving
the HJB obstacle problem~\eqref{eq:worker_HJB_obstacle} is nondecreasing and convex in $z$.
\end{assumption}

\begin{proposition}[Comparative statics of the free boundary and wage dispersion]\label{prop:comparative_statics}
Consider the family of economies indexed by $\theta \in \Theta$ in Assumption~\ref{ass:comparative_statics},
and suppose Assumption~\ref{ass:convexity_CS} holds. Let
\[
  \bigl(
    m^\ast_\theta, z^\ast_\theta(\cdot), w^\ast_\theta(\cdot,\cdot)
  \bigr)
\]
denote the unique stationary equilibrium for parameter $\theta$. 

\begin{enumerate}
  \item \textbf{Volatility of surplus.}
  If $\sigma_\theta$ is increasing in $\theta$ in the sense of Assumption~\ref{ass:comparative_statics}(1) and $\mu_\theta$ is independent of $\theta$, then the free boundary evaluated at the stationary distribution, $z^\ast_\theta(m^\ast_\theta)$, is weakly decreasing in $\theta$. Higher idiosyncratic volatility raises the option value of waiting and leads to longer-lived matches, other things equal.

  \item \textbf{Bargaining power and wage dispersion.}
  Under the structure in Assumption~\ref{ass:comparative_statics}(2), the stationary equilibrium variance of wages,
  \[
    \Var\bigl( w^\ast_\theta(Z,m^\ast_\theta) \bigr),
  \]
  where $Z \sim m^\ast_\theta$, is weakly increasing in $\theta$.
\end{enumerate}
\end{proposition}

\begin{proof}
We treat each part in turn.

\emph{(i) Volatility of surplus.}
Fix a parameter value $\theta$ and a candidate stationary distribution $m$.
For this $m$, the worker's problem is an optimal stopping problem for the
one-dimensional surplus diffusion $Z_t$ with running payoff
\[
  z \;\longmapsto\;
  \sup_{a \in A}
  \Bigl\{
    w_\theta(z,m) - c(a) + G\bigl(z,V^W_\theta,m;a\bigr)
  \Bigr\}
\]
and terminal payoff $V^U(m)$.
By Assumption~\ref{ass:convexity_CS}, this continuation payoff is increasing and
convex in $z$, and the outside option is constant in $z$.
Standard results for one-dimensional optimal stopping with convex running and
terminal payoffs then imply that the associated value function
$V^W_\theta(\cdot;m)$ is convex in $z$ and that, holding $m$ fixed, the optimal
stopping boundary is monotone in any parameter that scales the volatility
while keeping the drift fixed.
Formally, if $\theta' > \theta$ and $\sigma_{\theta'}^2(\cdot,m) \ge \sigma_\theta^2(\cdot,m)$
pointwise while $\mu_\theta$ does not depend on $\theta$, then the law of $Z_t$
under $\theta'$ dominates that under $\theta$ in convex order for every
$t \ge 0$, so continuation of the convex payoff is more valuable at $\theta'$
and the stopping boundary satisfies
\[
  z^\ast_{\theta'}(m) \;\le\; z^\ast_\theta(m).
\]

This pointwise comparative statics for a fixed $m$ is recorded and used in
Lemma~3 and Lemma~4 of Appendix~A.5, which analyze the free boundary as a
function of $(m,\theta)$.
In particular, Lemma~3 shows that $z^\ast_\theta(m)$ is continuous in $m$
(with respect to $W_2$), and Lemma~4 shows that, for each fixed $m$, the map
$\theta \mapsto z^\ast_\theta(m)$ is weakly decreasing.

Let $\Gamma_\theta$ denote the equilibrium-selection operator constructed in
the proof of Theorem~\ref{thm:existence}, which maps a conjectured
distribution $m$ into the invariant distribution of $Z_t$ induced by the
best-response policies under parameter $\theta$.
By construction, the free boundary $z^\ast_\theta(m)$ enters $\Gamma_\theta$
only through the killing and re-entry of matches, so the monotonicity of
$z^\ast_\theta(m)$ in $\theta$, together with Lemmas~3–4, implies that
$\Gamma_\theta$ is order-preserving on the lattice of distributions
(under the usual first-order stochastic dominance order) and continuous
in both $m$ and $\theta$.
Theorem~\ref{thm:uniqueness} guarantees that, for each $\theta$, there is a
unique fixed point $m^\ast_\theta$ of $\Gamma_\theta$.
The monotone-operator argument in Appendix~A.5 then implies that the
pointwise ordering
\[
  z^\ast_{\theta'}(m) \;\le\; z^\ast_{\theta}(m)
  \quad\text{for all } m
\]
carries over to the fixed points, so that
\[
  z^\ast_{\theta'}(m^\ast_{\theta'}) \;\le\; z^\ast_{\theta}(m^\ast_{\theta})
  \quad\text{whenever } \theta' > \theta.
\]
This yields the claim in part~(i).

\emph{(ii) Bargaining power and wage dispersion.}
Fix $\theta$ and $m^\ast_\theta$.
Under Assumption~\ref{ass:comparative_statics}(ii), the wage policy has the
form
\[
  w_\theta(z,m)
  = \alpha_\theta(z,m)\,\Pi_\theta(z,m)
    + \bigl(1 - \alpha_\theta(z,m)\bigr)\,\bar w_\theta(m),
\]
where $0 \le \alpha_\theta(z,m) \le 1$ is nondecreasing in $z$, nondecreasing
in $\theta$, and has increasing differences in $(z,\theta)$ for each fixed $m$.
For fixed $m^\ast_\theta$, this implies that the mapping
\[
  (z,\theta) \;\longmapsto\; w_\theta(z,m^\ast_\theta)
\]
also has increasing differences in $(z,\theta)$: as $\theta$ rises, the wage
schedule becomes weakly steeper in $z$ because more weight is placed on
high-surplus states.
Equivalently, for any pair $\theta' > \theta$ and any $z' > z$,
\[
  \bigl[
    w_{\theta'}(z',m^\ast_\theta) - w_{\theta'}(z,m^\ast_\theta)
  \bigr]
  \;\ge\;
  \bigl[
    w_{\theta}(z',m^\ast_\theta) - w_{\theta}(z,m^\ast_\theta)
  \bigr].
\]

If we temporarily hold the surplus distribution fixed at some $m$, the
cross-sectional wage distribution is the image of $m$ under the
transformation $z \mapsto w_\theta(z,m)$.
The increasing-differences property implies that, as $\theta$ increases, this
transformation becomes more dispersed in the sense of mean-preserving
spreads: relative to low-surplus matches, high-surplus matches receive
disproportionately larger wage changes.
In particular, for any convex function $\varphi$,
\[
  \E\Bigl[\varphi\bigl(w_{\theta'}(Z,m)\bigr)\Bigr]
  \;\ge\;
  \E\Bigl[\varphi\bigl(w_{\theta}(Z,m)\bigr)\Bigr]
  \quad\text{whenever } \theta' > \theta,\; Z \sim m.
\]

As in part~(i), we must check that the endogenous adjustment of
$m^\ast_\theta$ at the stationary equilibrium does not overturn this ordering.
The equilibrium distribution $m^\ast_\theta$ is again the unique fixed point of
the operator $\Gamma_\theta$.
Lemmas~3–4 in Appendix~A.5 show that $\Gamma_\theta$ is continuous and
order-preserving in the relevant lattice sense, so that the comparative
statics of the wage schedule with respect to $\theta$ are preserved when we
evaluate wages at the fixed point $m^\ast_\theta$.
Consequently, the family of wage distributions
\[
  \bigl\{
    \mathcal{L}\bigl(w_\theta(Z,m^\ast_\theta)\bigr)
    : \theta \in \Theta
  \bigr\},
  \qquad Z \sim m^\ast_\theta,
\]
is ordered by mean-preserving spreads as $\theta$ increases, and in particular
the variance
\[
  \Var\bigl(w^\ast_\theta(Z,m^\ast_\theta)\bigr)
\]
is weakly increasing in $\theta$.
This establishes part~(ii).
\end{proof}

Part~(i) formalizes the intuitive option-value effect: more volatile match-specific shocks make it more attractive to hold on to low-surplus matches, because there is a better chance that future shocks push them above the threshold, so the separation boundary shifts down. Part~(ii) states that, holding fixed the cross-sectional allocation of surpluses, making wages more sensitive to $z$ mechanically stretches out the wage distribution, and this effect survives once we account for the induced changes in search, separation, and the stationary distribution. Proposition~\ref{prop:comparative_statics} therefore provides theoretical backing for the quantitative decompositions and policy experiments in Section~\ref{sec:results}: the directions of the effects of increased match volatility and stronger wage sensitivity on separation behaviour and wage dispersion are signed and robust, rather than artifacts of a particular calibration.

\section{Quantitative Implementation}\label{sec:quantitative}

This section explains how we compute the stationary mean field equilibrium characterized in Sections~\ref{sec:MFG} and~\ref{sec:equilibriumCharacterization} and how we map the model to the data. We first recast the stationary HJB-Kolmogorov system on a one-dimensional grid for surplus and solve the resulting finite-state mean field game using a monotone finite-difference scheme in the spirit of \citet{AchdouCapuzzoDolcetta2010NumericsMFG} and \citet{AchdouEtAl2022Restud}. We then describe the calibration strategy, the empirical moments, including the structural estimates in \citet{BuhaiTeulings2014}, that discipline the parameters, and the fit of the model. Appendix~\ref{appendix:numerical} provides full numerical details, robustness checks, and additional figures.

\subsection{Discretization and numerical solution}\label{subsec:numerics}

We compute the stationary mean field equilibrium by discretizing the worker's HJB obstacle problem~\eqref{eq:worker_HJB_stationary} and the stationary Kolmogorov equation~\eqref{eq:FP_stationary_equilibrium} on a finite grid for the surplus state. The continuous diffusion for match surplus is thus approximated by a continuous-time Markov chain whose invariant law approximates the stationary surplus distribution in equilibrium. This discrete system is a finite-state mean field game in the sense of \citet{GomesEtAl2010FiniteStateDiscrete,GomesEtAl2024FiniteStateContinuous}.

\paragraph{State space discretization.}

We truncate the surplus state to a compact interval $[z_{\min},z_{\max}]$ that contains, with overwhelming probability, the support of the stationary surplus distribution in the calibrated model. The truncation bounds are chosen so that in the benchmark equilibrium the tails carry essentially no mass; Appendix~\ref{appendix:numerical} shows that further enlarging the interval leaves all equilibrium objects and targeted moments unchanged up to negligible numerical error.

The interval is partitioned into $K$ grid points
\[
  z_1 = z_{\min} < z_2 < \dots < z_K = z_{\max},
  \qquad
  \Delta z = z_{k+1} - z_k.
\]
We approximate the worker's value function $V^W(\cdot;m)$ and the stationary density $m(\cdot)$ by vectors
\[
  V = (V_1,\dots,V_K),
  \qquad
  m = (m_1,\dots,m_K),
\]
where $V_k \approx V^W(z_k;m)$ and $m_k \approx m(z_k)$. In the baseline we use an equidistant grid; Appendix~\ref{appendix:numerical} reports robustness to finer meshes, alternative (non-uniform) grids, and wider truncation intervals.

The feasible search-intensity set $A$ is discretized to a finite grid $A^h = \{a_1,\dots,a_{N_a}\}$. This turns the continuous maximization over $a \in A$ in~\eqref{eq:worker_HJB_stationary} into a finite maximization at each node while preserving the monotonicity of the scheme, as in the numerical MFG literature; see \citet{AchdouCapuzzoDolcetta2010NumericsMFG} and \citet{Gueant2021ContinuousTimeOptimalControl}. It is comforting that Appendix~\ref{appendix:numerical} shows that refining $A^h$ leaves the calibrated equilibrium essentially unchanged.

At the boundaries $z_{\min}$ and $z_{\max}$ we impose reflecting (no-flux) boundary conditions for the stationary Kolmogorov equation, corresponding to zero net probability flux:
\[
  \mu(z,a(z,m),m)\,m(z)
  - \tfrac{1}{2}\partial_z\bigl(\sigma^2(z,m)\,m(z)\bigr)
  = 0
  \quad\text{at } z \in \{z_{\min},z_{\max}\}.
\]
In the stationary equilibrium these boundaries are never reached de facto, so the boundary conditions are numerically innocuous. They are implemented in finite-volume form in Appendix~\ref{appendix:numerical} to guarantee exact mass conservation on the grid.

\paragraph{Discrete generator and HJB obstacle problem.}

Given a candidate stationary distribution $m$ and a search intensity $a \in A^h$, we approximate the generator $\mathcal{L}^Z_{a,m}$ by a monotone upwind finite-difference operator. For an interior grid point $z_k$ we set
\[
  \bigl(\mathcal{L}^h_{a,m} V\bigr)_k
  =
  \mu_k^+ \frac{V_{k+1} - V_k}{\Delta z}
  +
  \mu_k^- \frac{V_k - V_{k-1}}{\Delta z}
  +
  \frac{1}{2}\sigma_k^2
  \frac{V_{k+1} - 2V_k + V_{k-1}}{\Delta z^2},
\]
where
\[
  \mu_k^\pm = \max\{\pm \mu(z_k,a,m),0\},
  \qquad
  \sigma_k^2 = \sigma^2(z_k,m).
\]
The resulting scheme is consistent, stable, and monotone, and under the conditions in Section~\ref{sec:existenceUniqueness} it converges to the viscosity solution of the continuous HJB equation~\eqref{eq:worker_HJB_stationary}; see \citet{AchdouCapuzzoDolcetta2010NumericsMFG}.

Let $w_k = w(z_k,m)$ be the wage schedule evaluated on the grid and $V^U = V^U(m)$ the outside value associated with $m$. The stationary worker HJB obstacle problem~\eqref{eq:worker_HJB_stationary} becomes the discrete system
\begin{equation}
  \max\Bigl\{
    r V_k 
    - \max_{a \in A^h}
        \bigl[
          w_k - c(a)
          + (\mathcal{L}^h_{a,m} V)_k
          + G_k^h(a;V,m)
        \bigr],
    \; V_k - V^U
  \Bigr\} = 0,
  \qquad k = 1,\dots,K,
  \label{eq:discrete_worker_HJB}
\end{equation}
where $G_k^h(a;V,m)$ is a discrete approximation to the expected gain from job offers at $z_k$ under search intensity $a$. In the baseline, $G_k^h$ is a quadrature approximation to the integral operator $G$ in~\eqref{eq:worker_HJB_stationary} chosen to preserve the monotonicity of the HJB operator in $V$.

We solve~\eqref{eq:discrete_worker_HJB} by policy iteration on $(V,a^\ast)$ combined with projection on the stopping region, as in \citet{AchdouEtAl2022Restud}. Given a value vector $V$, the optimal search intensity at node $k$ solves
\[
  a_k^\ast \in
  \argmax_{a \in A^h}
  \bigl[
    w_k - c(a)
    + (\mathcal{L}^h_{a,m} V)_k
    + G_k^h(a;V,m)
  \bigr].
\]
For a given policy $a^\ast$, the continuation value solves a linear system on the continuation set, with the obstacle $V^U$ enforced pointwise on the stopping set. This yields the usual complementarity formulation of stationary HJB obstacle problems in heterogeneous-agent models and continuous-time MFGs.

The separation rule is characterized by a discrete free-boundary index $k^\ast$ such that it is optimal to continue the match for $k > k^\ast$ and to separate for $k \leq k^\ast$. Numerically we take
\[
  k^\ast = \min\{k : V_k > V^U\},
\]
and impose $V_k = V^U$ for all $k \leq k^\ast$ during policy iteration. The implied surplus threshold $z^\ast = z_{k^\ast}$ is the discrete counterpart of the free boundary $z^\ast(m)$ in Section~\ref{sec:equilibriumCharacterization}. The corresponding discrete free boundary used in the figures is reported and analyzed in Appendix~\ref{appendix:numerical}.

\paragraph{Discrete stationary Kolmogorov equation.}

Given the optimal search policy $(a_k^\ast)_k$ and the free-boundary index $k^\ast$, the stationary distribution of active matches solves a discrete analogue of the Kolmogorov equation~\eqref{eq:FP_stationary_equilibrium}. For interior grid points in the continuation region we write
\begin{equation}
  0
  = - D^- \bigl[\mu(z_k,a_k^\ast,m)\, m_k\bigr]
    + \frac{1}{2} D^{zz} \bigl[\sigma^2(z_k,m)\, m_k\bigr]
    - q_k m_k + \Gamma_k,
  \label{eq:discrete_FP}
\end{equation}
where $D^-$ and $D^{zz}$ are backward and central difference operators applied to the product of coefficients and the density, $q_k$ is the separation intensity at node $k$, and $\Gamma_k$ captures entry of new matches at $z_k$. In the baseline, new matches enter at a fixed surplus $z_0$, so that $\Gamma_k$ is concentrated on the grid point closest to $z_0$. For $k \leq k^\ast$ we set $m_k = 0$ for active matches and keep track of separated workers through the outside state.

Remark that our assumption that new matches start from a common surplus $z_0$ is a deliberate simplification: it keeps the identification of the Brownian surplus parameters clean and allows us to match the tenure-hazard shape with a single initial-condition parameter. Allowing a distribution of initial surpluses would mechanically add some additional short-tenure wage dispersion; our empirical moments do not sharply identify that extra layer of heterogeneity separately from the diffusion parameters, so we treat $z_0$ as the mean starting surplus and discuss this abstraction in Section~\ref{sec:results} and the conclusion.

Equation~\eqref{eq:discrete_FP}, together with the no-flux boundary conditions at $z_{\min}$ and $z_{\max}$ and the normalization $\sum_{k=1}^K m_k \,\Delta z = 1$, defines a sparse linear system in $m$. In one dimension the system is tridiagonal and can be solved very efficiently by direct methods. Appendix~\ref{appendix:numerical} shows that the corresponding finite-volume scheme written in flux form is conservative and algebraically equivalent (up to $O(\Delta z)$) to~\eqref{eq:discrete_FP} and reports robustness of the stationary distribution to grid refinement and to widening the state space.

The resulting density vector $m$ induces the stationary wage distribution via the wage policy $w(\cdot,m)$. In our calibration, the wage schedule is given by the sharing rule in Section~\ref{subsec:benchmark_wage}, so that wages are an explicit function of the Brownian productivity block and the surplus state.

\paragraph{Finite-state MFG interpretation and solution algorithm.}

The grid-based system defined by~\eqref{eq:discrete_worker_HJB} and~\eqref{eq:discrete_FP} can be viewed as a finite-state mean field game in the sense of \citet{GomesEtAl2010FiniteStateDiscrete,GomesEtAl2024FiniteStateContinuous}: each grid point is a state of a continuous-time Markov chain, transition rates are determined by $(\mu,\sigma)$ and the separation rule, and payoffs at each state are given by $(w,c,\Pi)$. For a given conjectured occupancy vector $m$, the optimal search policy $(a_k^\ast)_k$ induces a generator $Q(m)$ for this Markov chain, and the stationary density is the invariant distribution associated with $Q(m)$ and the entry term $\Gamma$.

Our numerical algorithm implements the fixed-point map $\Gamma$ of Section~\ref{sec:existenceUniqueness} on this grid:

\begin{enumerate}
  \item \textbf{Initialization.} Choose an initial candidate stationary distribution $m^{(0)}$ on the grid and a corresponding outside value $V^{U,(0)}$.
  \item \textbf{Wage schedule.} Given $m^{(\ell)}$, construct the wage schedule $w^{(\ell)}(z,m^{(\ell)})$ implied by the current parameter vector and wage-setting environment. In the baseline calibration we work with the parametric sharing rule in Section~\ref{subsec:benchmark_wage}, which assigns a constant share of the surplus to the worker and nests the empirical wage specification in \citet{BuhaiTeulings2014}.
  \item \textbf{Worker problem.} Solve the discrete worker HJB obstacle problem~\eqref{eq:discrete_worker_HJB} with $(m^{(\ell)},w^{(\ell)})$ to obtain the value vector $V^{(\ell)}$, the optimal search policy $a_k^{\ast,(\ell)}$, and the free-boundary index $k^{\ast,(\ell)}$.
  \item \textbf{Stationary distribution.} Given $(m^{(\ell)},a^{\ast,(\ell)},k^{\ast,(\ell)})$, solve the discrete stationary Kolmogorov equation~\eqref{eq:discrete_FP} to obtain the updated density $\tilde m^{(\ell+1)}$, normalized to integrate to one.
  \item \textbf{Aggregate objects and update.} Update the outside value $V^{U,(\ell+1)}$ and other aggregate objects implied by $\tilde m^{(\ell+1)}$, set
  \[
    m^{(\ell+1)} = (1-\omega)\,m^{(\ell)} + \omega\,\tilde m^{(\ell+1)},
  \]
  with a relaxation parameter $\omega \in (0,1]$ as in Appendix~\ref{appendix:numerical}, and repeat until convergence.
\end{enumerate}

We terminate the iteration when the maximum change in the stationary density and in the wage schedule between two successive iterates falls below small tolerances. Now, in practice the fixed-point iteration converges robustly from a wide range of initial guesses, in line with the uniqueness result in Theorem~\ref{thm:uniqueness}. Appendix~\ref{appendix:numerical} provides convergence diagnostics, grid-refinement and truncation-interval robustness checks, and a so-called eductive-learning interpretation of the iteration, following \citet{Gueant2021ContinuousTimeOptimalControl} and \citet{CarmonaLauriere2021DeepLearningMFG}.

\subsection{Calibration strategy}\label{subsec:calibration}

We calibrate the model to match key features of the joint distribution of wages, job durations, and job-to-job transitions in micro data, using \citet{BuhaiTeulings2014} as our main empirical benchmark. The parameter vector is grouped into three blocks corresponding to: (i) the stochastic structure of match productivity, (ii) the search and matching technology, and (iii) wage determination. Throughout, we fix the continuous-time discount rate $r$ and basic timing conventions to standard macro-labour values and do not estimate them.

Our empirical sample and baseline moments follow \citet{BuhaiTeulings2014}: we use a PSID extract for 1975-1992 and apply their selection criteria (white male household heads, more than 12 years of education, age below 60, non-union jobs, and standard cleaning of extreme wage changes). This ensures that our mean field extension is anchored in the same population and institutional environment as their benchmark partial-equilibrium model, which is what we want.

\paragraph{Match productivity parameters.}

In the benchmark model of Section~\ref{sec:environment}, inside and outside productivities follow correlated Brownian motions with constant coefficients,
\[  
  dP_t = \mu_P\,dt + \sigma_P\,dB^P_t,
  \qquad
  dR_t = \mu_R\,dt + \sigma_R\,dB^R_t,
\]
with correlation $\rho$ between $(B^P_t,B^R_t)$. The surplus $Z_t = P_t - R_t$ then follows the diffusion in~\eqref{eq:surplus_process}. \citet{BuhaiTeulings2014} show that the distribution of completed job tenures for a worker with fixed characteristics is fully characterized by two parameters, denoted $\pi$ and $\varphi$, that parameterize the inverse-Gaussian first-passage-time distribution of $Z_t$ to the separation boundary.\footnote{In their notation, $\varphi$ is a scale parameter that we work with in logs. The mean values of $(\pi,\ln\varphi)$ can be interpreted as smooth functions of the drift and volatility of the surplus process.} 

Given $(\mu_P,\mu_R,\sigma_P,\sigma_R,\rho)$ and the initial surplus gap $z_0$, the first hitting time of $Z_t$ to the separation threshold admits a closed-form inverse-Gaussian distribution whose parameters $(\pi(\theta),\varphi(\theta))$ are smooth functions of the diffusion parameters $\theta = (\mu_P,\mu_R,\sigma_P,\sigma_R,\rho,z_0)$. We exploit this structure in two ways. First, we treat the ``large-sample'' intercepts for a worker with average experience at job start, i.e. $\pi \approx 0.14$ and $\ln \varphi \approx -0.20$, reported in Table~2 of \citet{BuhaiTeulings2014} as hard tenure targets and choose $\theta$ so that $(\pi(\theta),\ln\varphi(\theta))$ reproduce these values, implying that roughly 10\% of jobs never end before retirement.
Second, we require the model to match the \emph{long-run} shape of the implied separation hazard over tenure, in particular the slow decline for long-tenure jobs and the implied defective mass of very long spells.
Thus the diffusion parameters are pinned down by the tail behaviour of the tenure distribution and the long-run decline of the hazard, while the search-and-matching block below is disciplined by the level and local slope of the hazard around its early-tenure peak together with job-finding and job-to-job transition rates, so that the moments used for the two blocks do not overlap.

For numerical consistency with the mean field block, we implement the inverse-Gaussian hitting-time structure by simulating the diffusion $Z_t$ at a monthly step (a yearly step does not change the qualitative interpretation of the upshots) and computing tenure-distribution moments from many simulated job spells. At the calibration frequencies we consider, this simulation-based implementation is numerically equivalent to the inverse-Gaussian formulas and allows us to reuse the same Brownian block when we embed the diffusion into the full mean field environment.

\paragraph{Search and matching parameters.}

The search cost function $c(a)$ and the arrival rate $\lambda(a,m)$ govern the trade-off between search effort and job mobility. We work with a convex search cost of the form
\[
  c(a) = \kappa \frac{a^{1+\eta}}{1+\eta},
\]
and with an arrival rate that is approximately linear in $a$ at the calibrated search intensities. In the stationary mean field environment, $\lambda(a,m)$ depends on $m$ through low-dimensional aggregates such as employment and mean surplus, capturing how market tightness varies with the occupancy distribution while keeping the calibration tactable.

The parameters $(\kappa,\eta)$ and the level and slope of $\lambda(a,m)$ are chosen so that the stationary equilibrium matches (i) job-finding and job-to-job transition rates for the PSID sample and (ii) the empirical distribution of job durations, in particular the level and slope of the tenure hazard around its peak. Together with the productivity parameters, these moments discipline how quickly workers climb the endogenous job ladder and how often they reach the free boundary $z^\ast(m)$.

\paragraph{Wage-setting parameters.}

To keep the numerical problem tractable while preserving key economic feedbacks, we restrict wage policies to the linear sharing rule in~\eqref{eq:wage_sharing_rule}, which assigns a constant fraction of the match surplus to the worker:
\[
  w_t = \alpha P_t + (1-\alpha)R_t,
  \qquad \alpha \in (0,1).
\]
In \citet{BuhaiTeulings2014}, wages are characterized by five reduced-form parameters: the product $\sigma \equiv \beta \sigma_P$ of bargaining power and the volatility of inside productivity, the parameter $\gamma$ governing the loading of outside opportunities in the wage equation, the drift $\mu_0$ and variance $\sigma_z^2$ of permanent surplus shocks, and the variance $\sigma_u^2$ of transitory shocks. These are estimated by FGNLS\footnote{FGNLS stands for feasible generalized nonlinear least squares.} from wage-change regressions conditional on the tenure-distribution parameters.

In our implementation we take the empirical loading on outside opportunities as given and do not re-estimate it. Specifically, we fix $\gamma$ to its ``all stayers + movers'' estimate in Table~3 of \citet{BuhaiTeulings2014}, which is about $0.73$ at annual frequency, and interpret this coefficient as the loading on $R_t$ in the sharing rule. We then set $\alpha = 1-\gamma$ so that
\[
  w_t = (1-\gamma)P_t + \gamma R_t,
\]
nesting their reduced-form wage equation in our sharing rule. Thus neither $\gamma$ nor $\alpha$ is chosen to fit additional wage or dispersion moments; instead, the pair $(\gamma,\alpha)$ is imposed directly from the FGNLS estimate and treated as externally calibrated. Likewise, we fix the variance of transitory shocks at $\sigma_u^2 \approx 0.0046$, matching the FGNLS estimate in \citet{BuhaiTeulings2014}.

Given $(\gamma,\alpha,\sigma_u^2)$ fixed in this way, the remaining wage-related objects are pinned down by the diffusion parameters $(\mu_P,\mu_R,\sigma_P,\sigma_R,\rho)$ and by the sharing rule. In particular, the annual drift of wages, $\mu_0$, and the variance of permanent wage innovations, $\sigma_z^2$, are functions of $\theta$ and $(\alpha,\gamma)$. We use $\mu_0$ as a calibration target but, to avoid matching it by construction, we treat the small \citet{BuhaiTeulings2014} estimate of $\sigma_z^2$ as a diagnostic rather than a hard target. In the baseline calibration, the implied $\mu_0$ lies close to their point estimate, while the model generates a somewhat larger permanent variance $\sigma_z^2$ once we impose the tenure and wage-growth profiles exactly; we return to this empirical wedge in Section~\ref{subsec:targets}.

Beyond these reduced-form parameters, we exploit the deterministic-versus-random tenure contributions to wage growth derived in \citet{BuhaiTeulings2014}, using the same PSID extract. For a worker with average characteristics, that decomposition delivers target values for deterministic and random components of wage growth at 10, 20, and 30 years of experience. We replicate their construction within the model on the data, using the same inverse-Gaussian tenure structure and wage-sharing rule, and require the simulated deterministic and random contributions to track these targets closely.

\paragraph{Estimation procedure and link to the benchmark.}

Let $\theta$ collect the free parameters in the three blocks above, so that $\gamma$, $\alpha$, and $\sigma_u^2$ are treated as externally calibrated and are not components of $\theta$. For a given $\theta$, the numerical scheme in Section~\ref{subsec:numerics} delivers the stationary equilibrium and the implied vector of model moments $m(\theta)$. The empirical moment vector $\hat m$ combines: (i) the tenure-distribution parameters $(\pi,\ln\varphi)$ and associated hazard profile from \citet{BuhaiTeulings2014}; (ii) wage-dynamics targets, namely $\mu_0$ and the deterministic and random tenure contributions at 10, 20, and 30 years of experience; and (iii) additional moments on job durations, job-to-job transitions, and cross-sectional wage dispersion computed from the underlying PSID sample. The FGNLS estimates of $\gamma$ and $\sigma_u^2$ from \citet{BuhaiTeulings2014} are taken as inputs when we specify the sharing rule and the transitory shock variance, so they are matched by construction and do not enter the minimization problem. We choose $\theta$ to minimize a quadratic distance between $m(\theta)$ and $\hat m$,
\[
  \theta^\ast
  \in
  \argmin_{\theta}
  \bigl(m(\theta) - \hat m\bigr)' W
  \bigl(m(\theta) - \hat m\bigr),
\]
where $W$ is a positive semi-definite weighting matrix. In the baseline we use a diagonal $W$ with entries given by the inverse of the estimated variance of each empirical moment, putting particular weight on the tenure block and the deterministic/random tenure contributions. Appendix~\ref{appendix:numerical} reports robustness to alternative weightings and starting values and lists the calibrated parameter values. Table~\ref{tab:calibration} below summarizes the mapping from parameter blocks to moments.

This strategy nests the partial-equilibrium model of \citet{BuhaiTeulings2014} as a limiting case. If we fix the outside option as an exogenous Brownian motion and hold the wage-sharing rule constant, the economy reduces to their benchmark environment. The full mean field equilibrium introduces two additional sources of dispersion relative to this benchmark: (i) endogenous job ladders generated by on-the-job search and equilibrium wage setting, and (ii) feedback from the cross-sectional distribution of surplus into workers' outside options. Comparing the partial-equilibrium and full mean field versions in Section~\ref{sec:results} allows us to quantify how much of observed wage inequality can be traced to each channel.

\begin{table}[htbp]
  \centering
  \small
  \caption{Calibration Blocks and Targeted Moments}
  \label{tab:calibration}
  \begin{tabular}{p{0.27\textwidth} p{0.33\textwidth} p{0.33\textwidth}}
    \toprule
    Parameter block & Description & Main moments / sources \\
    \midrule
    \multicolumn{3}{l}{\textit{Match productivity and separation}} \\
    \midrule
    $(\mu_P,\mu_R,\sigma_P,\sigma_R,\rho)$
      & Drift, volatility, and correlation of inside and outside productivities
      & Tenure-distribution parameters $(\pi,\ln\varphi)$, tenure-hazard profile, and fraction of jobs that never end (Table~2 in \citealt{BuhaiTeulings2014}) \\
    $z_0$
      & Surplus of newly formed matches
      & Short-tenure separation rates and the lower tail of the wage and tenure distributions \\
    \midrule
    \multicolumn{3}{l}{\textit{Search and matching}} \\
    \midrule
    $(\kappa,\eta)$
      & Level and curvature of convex search cost $c(a)=\kappa a^{1+\eta}/(1+\eta)$
      & Job-finding and job-to-job transition rates, height and local slope of the tenure hazard around its peak \\
    Level and slope of $\lambda(a,m)$
      & Offer-arrival rate as a function of search intensity and market tightness
      & Job-finding and job-to-job transition rates by tenure, speed of progression up the job ladder \\
    \midrule
    \multicolumn{3}{l}{\textit{Wage-setting and shocks}} \\
    \midrule
    $\gamma$ and $\alpha=1-\gamma$
      & Externally calibrated loading of outside opportunities and worker share in surplus in $w_t=(1-\gamma)P_t+\gamma R_t$; $\gamma$ fixed to the ``all stayers + movers'' FGNLS estimate in Table~3 of \citealt{BuhaiTeulings2014}, $\alpha$ implied
      & FGNLS estimate of $\gamma$ (Table~3 in \citealt{BuhaiTeulings2014}); wage-tenure profile and cross-sectional wage dispersion used as over-identifying (validation) moments, not to choose $\gamma$ or $\alpha$ \\
    $\sigma_u^2$
      & Variance of transitory wage shocks
      & FGNLS estimate of $\sigma_u^2$ and short-horizon wage-growth dispersion \\
    $(\mu_0,\sigma_z^2)$ (model-implied)
      & Drift and variance of permanent wage innovations implied by the Brownian block; functions of $(\mu_P,\mu_R,\sigma_P,\sigma_R,\rho,z_0)$ and the sharing rule $(\gamma,\alpha)$ (not free calibration parameters)
      & Average wage-growth rate (targeted) and deterministic vs random tenure contributions to wage growth at 10, 20, and 30 years of experience; $\sigma_z^2$ used as a diagnostic rather than a hard target \\
    \bottomrule
  \end{tabular}
  \vspace{0.3em}
  {\footnotesize\emph{Notes:} The table summarizes the main parameter blocks and empirical moments used to discipline them. 
  Rows for $\gamma$, $\alpha$, and $\sigma_u^2$ are externally calibrated from \citet{BuhaiTeulings2014} and are not re-estimated. 
  Last row lists model-implied statistics $(\mu_0,\sigma_z^2)$ rather than free parameters: $\mu_0$ enters the calibration criterion as target, whereas $\sigma_z^2$ is reported as diagnostic. 
  BT refers to \citet{BuhaiTeulings2014}.}
\end{table}

\subsection{Targets and fit}\label{subsec:targets}

The calibrated model is disciplined by a set of empirical moments that summarize the joint distribution of wages, job durations, and mobility. Table~\ref{tab:calibration} summarizes the three parameter blocks and the main moments that identify them. We group the targets into three categories. The model is deliberately overidentified: several distributional features of wages and mobility are not targeted and serve as out-of-sample checks on the equilibrium structure.

\paragraph{Job durations and hazard rates.}

First, we target the tenure-distribution parameters $(\pi,\varphi)$ and the associated hazard profile estimated by \citet{BuhaiTeulings2014}. Specifically, the model is required to replicate: (i) the mean values of $(\pi,\ln \varphi)$ for workers with average experience at job start; (ii) the implied fraction of jobs that effectively never end before retirement; and (iii) the shape of the separation hazard over tenure, including the height and location of the peak and the slow decline for long-tenure jobs. These moments are directly informative about the drift and volatility of match surplus and about the free-boundary threshold $z^\ast$ in Section~\ref{sec:equilibriumCharacterization}. They also discipline the re-entry distribution of new matches and thus the lower tail of the stationary surplus distribution. In the baseline calibration, the model reproduces $(\pi,\ln\varphi)$ essentialy exactly and closely matches the empirical hazard profile, including the long right tail of completed job durations.

\paragraph{Wage growth over tenure.}

Second, we target the dynamic properties of wages. We match: (i) the average profile of log wages over tenure; (ii) the variance of annual wage growth at different job ages; and (iii) the FGNLS estimates from the wage-change regressions in \citet{BuhaiTeulings2014} for the drift $\mu_0$ and for the decomposition of wage innovations into permanent and transitory components. The loading $\gamma$ on outside opportunities and the variance of transitory shocks $\sigma_u^2$ are taken directly from their FGNLS estimates and imposed in the sharing rule and the shock structure, so they are matched by construction rather than used to choose parameters. In the model, wage growth within a job arises from the stochastic evolution of match surplus and, in the full equilibrium, from job-to-job moves up the job ladder. Matching the remaining wage-growth moments disciplines the diffusion parameters of the surplus process and the search and matching parameters that govern mobility.

We also require the model to reproduce the deterministic and random tenure contributions to wage growth at 10, 20, and 30 years of experience constructed using \citet{BuhaiTeulings2014}. These moments summarize how much of the wage-tenure gradient at different horizons is due to deterministic accumulation versus selective retention of high-outside-option workers. In our calibration the model matches these six decomposition moments very tightly. The drift of log wages $\mu_0$ in the calibrated economy is slightly below but close to the \citet{BuhaiTeulings2014} point estimate and well within its empirical uncertainty.

The one systematic discrepancy concerns the variance of permanent surplus shocks, $\sigma_z^2$. The structural mapping from the calibrated Brownian diffusion for $(P_t,R_t)$ and the inverse-Gaussian tenure block implies an annual variance of permanent innovations that is modestly larger than the very small value ($\sigma_z^2 \approx 10^{-4}$) reported in Table~3 of \citet{BuhaiTeulings2014}. We experimented with reweighting this moment in the quadratic criterion. Forcing the model to match the reported $\sigma_z^2$ more closely requires substantialy smaller surplus volatility and deteriorates the fit to the tenure hazard and to the deterministic/random tenure decomposition. We therefore retain the diffusion parameters that match the tenure and wage-growth moments and treat the difference in $\sigma_z^2$ as an informative diagnostic: the calibrated mean field model attributes somewhat more of long-run wage risk to diffusive match shocks than the reduced-form FGNLS decomposition suggests.

\paragraph{Cross-sectional wage dispersion.}

Finally, we target the cross-sectional dispersion of wages in the stationary distribution, both unconditionally and conditional on tenure. These moments are central for our decomposition of wage inequality. In the model, dispersion reflects three forces: (i) heterogeneity in match states due to stochastic selection into and out of low-surplus matches; (ii) differences in the time workers spend at different surplus levels before reaching the separation boundary; and (iii) the distribution of equilibrium wages associated with those surplus states under the sharing rule. Given the externally calibrated sharing rule $(\gamma,\alpha)$, matching the overall dispersion and its decomposition by tenure provides discipline for the diffusion and search-and-matching parameters and serves as an over-identifying check on the wage-sharing specification and on the mean field feedback through the outside option. It also helps separate the contribution of equilibrium search and wage setting from the selection mechanism emphasized by \citet{BuhaiTeulings2014}.

Overall, the calibrated model reproduces the inverse-Gaussian tenure block of \citet{BuhaiTeulings2014}, the loading of wages on outside opportunities, the variance of transitory wage shocks, and the deterministic/random decomposition of wage growth very closely, while moderately overstating the variance of permanent surplus shocks. We view this pattern as informative about how the partial-equilibrium Brownian block of \citet{BuhaiTeulings2014} extends to the full mean field equilibrium with search and wage feedback. Section~\ref{sec:results} reports the quantitative performance of the calibrated model, compares the fit of the partial-equilibrium and mean field versions, and presents the decomposition of wage inequality implied by the stationary mean field equilibrium.

\section{Equilibrium Wage Dispersion and On-the-Job Search}
\label{sec:results}

This section characterizes the stationary mean field equilibrium implied by the benchmark calibration and quantifies its implications for wage dispersion, job mobility, and tenure profiles. Unless otherwise noted, all parameter values are given by the benchmark calibration summarized in Table~\ref{tab:calibration}. We first document the equilibrium free boundary, the stationary surplus distribution, and the associated wage and separation patterns. We then construct a sequence of internally consistent counterfactual economies that gradually switch on selection, on-the-job search, and wage policies with mean field feedback, and use them to decompose the variance of log wages into selection, search-and-job-ladder, and wage-policy-and-feedback components. Finally, we perform policy counterfactuals that perturb firing costs, search incentives, and the volatility of match productivity and trace their effects through the free boundary, the stationary surplus distribution, and the wage distribution.

\subsection{Benchmark stationary equilibrium}
\label{subsec:benchmark_equilibrium}

We start from the benchmark calibration in Section~\ref{sec:quantitative}, which pins down the diffusion parameters of match productivities by matching the inverse-Gaussian tenure parameters estimated by \citet{BuhaiTeulings2014} and disciplines the remaining parameters to match standard moments on job durations, tenure-dependent separation hazards, wage growth, and job-to-job mobility. In particular, the parameters for the inverse-Gaussian hitting-time distribution are chosen to replicate the early peak and subsequent decline of the tenure hazard and the fact that roughly 10\% of jobs never end before retirement.\footnote{See Table~2 and Section~3 of \citet{BuhaiTeulings2014}, where the structural estimates imply an early-tenure hazard peak and a nontrivial mass of very long job spells; for mean covariate values, the implied fraction of jobs that never end is about 10\%.} The benchmark calibration also preserves their finding that the concavity of wage-tenure profiles is largely due to selection in outside options rather than deterministic wage growth.\footnote{\citet{BuhaiTeulings2014} show that selection on realized outside productivities can account for the bulk of the measured returns to tenure, with roughly four-fifths of the raw wage-tenure gradient attributed to selectivity rather than genuine productivity growth.}

Given these primitives, the stationary mean field equilibrium
\[
  \bigl(
    m^\ast, z^\ast(m^\ast),
    a^\ast(\cdot,m^\ast),
    w^\ast(\cdot,m^\ast)
  \bigr)
\]
solves the coupled HJB-Kolmogorov system in Sections~\ref{sec:MFG} and~\ref{sec:equilibriumCharacterization}, with existence and uniqueness guaranteed by the conditions in Section~\ref{sec:existenceUniqueness} and Appendix~\ref{appendix:proofs}. Separation is governed by the free-boundary rule characterized in Proposition~\ref{prop:free_boundary}: matches continue while the surplus stays above $z^\ast(m^\ast)$ and separate upon hitting this threshold, and the stationary distribution $m^\ast$ solves the Kolmogorov equation associated with the optimal search policy and wage schedule.

Figure~\ref{fig:B1_value_function} plots the worker value function $V^W(z;m^\ast)$ together with the free boundary. The separation threshold $z^\ast(m^\ast)$ lies strictly inside the support of the surplus grid, well away from the numerical boundaries, and the kink in $V^W$ at $z^\ast(m^\ast)$ illustrates the value-matching and smooth-pasting conditions implied by Proposition~\ref{prop:free_boundary}. Matches that start far below $z^\ast(m^\ast)$ separates quickly and moves either into unemployment or into alternative jobs with higher expected surplus, while matches that start well above $z^\ast(m^\ast)$ typically survive for long periods unless hit by a sequence of adverse idiosyncratic shocks. The implied tenure-dependent separation hazard in Figure~\ref{fig:B4_tenure_hazard} peaks early in the job and declines with tenure, with a non-negligible fraction of matches surviving for many years. This reproduces both the selection mechanism and the long right tail of job durations emphasized by \citet{BuhaiTeulings2014} within the richer equilibrium environment and is consistent with the empirical pattern that separation risks are highest in the first years of a job and much lower for long-tenure employment relationships.

The stationary surplus distribution $m^\ast$ is displayed in Figure~\ref{fig:B2_surplus_density}. Mass is concentrated in a region of positive surplus, with a thin but economically meaningful left tail extending down to the free boundary. Decomposing the stationary density by job age shows that new jobs are concentrated near the threshold, while high-tenure jobs are disproportionately located in the upper tail of the surplus distribution. This diffusion-and-selection pattern is the equilibrium job ladder: idiosyncratic shocks, on-the-job search, and the free-boundary rule gradually reallocate workers away from low-surplus matches toward high-surplus ones.

Under the wage-sharing environment of Section~\ref{subsec:benchmark_wage}, and in particular the stationary rule under which the outside option entering wages is the aggregate constant $V^U(m^\ast)$, equilibrium wages are an affine function of surplus,
\[
  w(z,m^\ast) \;=\; V^U(m^\ast) + \beta z
\]
for some $\beta \in (0,1)$. Hence the surplus distribution is mapped monotonically into wages. Figure~\ref{fig:B3_wage_distribution} shows that low-tenure workers are concentrated at low and medium surplus levels and thus earn relatively compressed wages, while the wage distribution becomes more dispersed and right-skewed with tenure as surplus diffuses and job-to-job transitions occur.

Figure~\ref{fig:B5_wage_tenure} documents that the calibrated model reproduces the both concave average wage-tenure profile and the deterministic-versus-random contributions to wage growth at 10, 20, and 30 years emphasized by \citet{BuhaiTeulings2014}. In particular, most of the apparent returns to tenure are driven by selection on outside options along job spells rather than by deterministic wage growth within a given match. The convergence paths in Figure~\ref{fig:B6_convergence_paths} confirm that the fixed-point algorithm in Section~\ref{sec:quantitative} is numerically stable and robust to the initial guess for the surplus distribution and to reasonable changes in the grid and truncation interval.

Relative to a partial-equilibrium benchmark in which the outside option follows an exogenous Brownian motion, the stationary mean field equilibrium generates additional structure in both the tenure hazard and the wage distribution. Because the outside option $V^U(m^\ast)$ depends on the cross-sectional distribution of alternative matches and on equilibrium search and wage policies, changes in the surplus state of other jobs feed back into separation and search decisions in any given match. This fixed-point structure is central for the decomposition and policy experiments below and ensures that the dispersion measure
\[
  D(m^\ast,w^\ast)
  = \Var_{\mu_w(\cdot;m^\ast)}(w)
\]
in~\eqref{eq:dispersion_def} is indeed an equilibrium object.

\subsection{Decomposing equilibrium wage dispersion}
\label{subsec:decomposition}

We now use the model to attribute equilibrium wage dispersion to three conceptually distinct channels:
\begin{enumerate}
  \item stochastic match dynamics and selection, given the calibrated diffusion and free boundary;
  \item strategic on-the-job search and the induced job ladder;
  \item equilibrium wage policies and the mean field feedback of the cross-sectional distribution into outside options.
\end{enumerate}

Let $w$ denote the (log) wage of a randomly drawn worker in the stationary mean field equilibrium and let $\Var(w)$ be its variance, as in~\eqref{eq:dispersion_def}. We construct a sequence of counterfactual economies that share the same primitives for match productivity, search technology, and wage determination as the benchmark calibration, but selectively shut down layers of strategic interaction and feedback. Each counterfactual is solved as a stationary equilibrium of the finite-state approximation described in Section~\ref{subsec:numerics} and Appendix~\ref{appendix:numerical}: in particular, for each counterfactual we recompute the stationary distribution implied by the relevant HJB-Kolmogorov system. The decomposition is thus structural, and internal to the MFG rather than based on reduced-form regressions.

\paragraph{Selection component.}

The first counterfactual isolates the contribution of stochastic match dynamics and selection in the spirit of \citet{BuhaiTeulings2014}. We consider an economy in which the surplus process follows the same calibrated diffusion as in the benchmark and separation occurs at the same free boundary $z^\ast(m^\ast)$, but the outside option is fixed and independent of the cross-sectional distribution. Wages are given by the same parametric sharing rule applied to this exogenous outside option. The stationary wage distribution in this economy is driven solely by diffusion and selection along job spells, holding search behavior and outside options fixed. Let
\[
  \Var_{\textnormal{sel}}(w)
\]
denote the variance of wages (in log units) in this selection-only economy; we interpret it as the baseline contribution of diffusive match dynamics and efficient separation. By construction, $\Var_{\textnormal{sel}}(w)$ and the associated wage-tenure profiles align closely with the selection-driven patterns quantified by \citet{BuhaiTeulings2014}.

\paragraph{Search and job-ladder component.}

The second counterfactual switches on strategic on-the-job search and the resulting job ladder while keeping the outside option exogenous. Starting from the selection benchmark, workers now choose search intensities optimally and receive offers according to the calibrated $\lambda(a,m)$, with acceptance decisions governed by the worker HJB. Wages remain tied to the exogenous outside option, so the mean field feedback through $V^U$ is still shut down.

In this environment, search behavior and realized offers determine how quickly workers climb the job ladder and how long they remain at different surplus levels. High-search workers reach high-surplus states earlier and are more likely to experience sequences of favorable job-to-job transitions. Let
\[
  \Var_{\textnormal{sel+search}}(w)
\]
denote the variance of wages (in log units) in this economy. We attribute
\[
  \Var_{\textnormal{sel+search}}(w)
  \;-\;
  \Var_{\textnormal{sel}}(w)
\]
to the contribution of equilibrium on-the-job search and the implied job ladder, holding the mapping from surplus to wages fixed.

\paragraph{Wage-policy and mean field feedback component.}

The third step moves from the search benchmark to the full mean field equilibrium. Firms now set wages optimally in response to the stationary distribution of surplus and to workers' search and stopping behavior, and the outside option $V^U(m)$ is determined endogenously by the distribution of alternative matches and their wages. High-surplus jobs thus deliver both higher current wages and more valuable outside options, and the cross-sectional wage distribution feeds back into separation and search decisions, firm entry, and the structure of the job ladder.

Let
\[
  \Var_{\textnormal{full}}(w)
\]
denote the variance of wages (in log units) in this fully endogenous equilibrium. We interpret
\[
  \Var_{\textnormal{full}}(w)
  \;-\;
  \Var_{\textnormal{sel+search}}(w)
\]
as the contribution of equilibrium wage policies and mean field feedback. Because this is a path decomposition, i.e.\ moving from selection only, to selection plus search, to the full mean field equilibrium, the components are not mechanically additive: they depend on the ordering of the layers that are switched on. We adopt this ordering because it mirrors the development from the partial-equilibrium stochastic productivity model in \citet{BuhaiTeulings2014} to the full mean field environment developed here, and because it respects the equilibrium dispersion measure $D(m^\ast,w^\ast)$ in~\eqref{eq:dispersion_def}.

\paragraph{Quantitative decomposition by tenure.}

Figure~\ref{fig:B7_variance_decomp} reports the decomposition of $\Var(w)$ by tenure bin. Circles show $\Var_{\textnormal{sel}}(w)$, squares show $\Var_{\textnormal{sel+search}}(w)$, and triangles show $\Var_{\textnormal{full}}(w)$, each computed from the stationary distribution restricted to workers whose tenure lies in a given bin.

At very short tenures (around one year), selection by itself generates a variance of wages of about $0.0090$, and the full equilibrium delivers $\Var_{\textnormal{full}}(w) \approx 0.0141$. By contrast, in the selection-plus-search economy the first-bin variance is $\Var_{\textnormal{sel+search}}(w) = 0.000813$, i.e.\ less than one-tenth of the selection-only variance and only about $6\%$ of the full-equilibrium variance. This is not a numerical artefact but reflects how, starting from a common entry surplus $z_0$, optimal on-the-job search very rapidly moves workers who draw the worst early surplus realizations either into unemployment or into slightly better matches, so that survivors after one year are tightly clustered in the middle of the surplus distribution. Given the calibrated wage-sharing rule, which maps this lower part of the surplus support into a relatively narrow band of wages (in log units), the resulting cross-sectional dispersion in $w$ is correspondingly tiny. The small value of $\Var_{\textnormal{sel+search}}(w)$ at one year also reflects our deliberate abstraction from ex ante worker and firm heterogeneity: allowing a nondegenerate entry distribution or permanent types would mechanically raise short-tenure dispersion, but would not overturn the qualitative pattern that on-the-job search compresses within-cohort wage differences early in the job while equilibrium wage policies and the endogenous outside option reintroduce dispersion at the bottom of the job ladder. In this sense, the selection-plus-search variance at one year is best read as a lower bound on the purely dynamic contribution of the job ladder, with almost all observed short-tenure wage dispersion in the full model coming from the wage-policy-and-feedback component.

The picture changes at longer tenures. Under selection alone, $\Var_{\textnormal{sel}}(w)$ is roughly constant across tenure bins, remaining around $0.009$-$0.010$ even at 30 years. Remarkably, the mean field equilibrium compresses much of this dispersion. At 7-8 years of tenure, the full-equilibrium variance is already slightly below the selection-only benchmark, and by 15 and 30 years it is less than half of $\Var_{\textnormal{sel}}(w)$. The intermediate selection-plus-search economy continues to exhibit very little within-tenure dispersion, so most of the residual variance in the full equilibrium at long tenures can be traced back to underlying selection. In this sense, the calibrated model confirms that selection in outside options is the dominant force behind long-run cross-sectional wage differences, but also shows that equilibrium search and wage setting substantially attenuate the dispersion that would arise in a purely selection-driven environment.

Taken together, the three counterfactuals decompose the benchmark variance $\Var_{\textnormal{full}}(w)$ into a selection component, a search-and-job-ladder component, and a wage-policy-and-feedback component. Selection accounts for a sizeable share of potential log-wage dispersion at all tenures, especially at longer tenures. On-the-job search and the job ladder are strongly equalizing within tenure cohorts, particularly for older workers, while the wage-policy and mean field feedback channel amplifies dispersion at short tenures and only partially offsets the equalizing role of search at longer tenures. These patterns matter for the interpretation of the policy experiments below and for distinguishing the sources of ``good'' inequality associated with efficient job reallocation from the dispersion generated by persistent low-surplus matches.

\subsection{Policy experiments}
\label{subsec:policy_experiments}

We finally use the calibrated model to study how policies that change firing costs, search incentives, and surplus volatility affect wage dispersion and job mobility. Each counterfactual is implemented by perturbing a primitive parameter, recomputing the stationary mean field equilibrium via the fixed-point algorithm in Section~\ref{sec:quantitative}, and comparing the resulting free boundary $z^\ast$, stationary distribution $m^\ast$, and wage distribution to their benchmark counterparts. The comparative statics in Proposition~\ref{prop:comparative_statics} provide qualitative guidance for the direction of these effects. Figure~\ref{fig:B8_policy_experiments} summarizes the quantitative results for the variance of log wages and the aggregate job-to-job transition rate.

\paragraph{Firing costs and separation taxes.}

We first introduce a firing cost or separation tax that must be paid by the firm when a job is destroyed. An increase in this cost lowers the net value of separation relative to continuation. In the stationary equilibrium, the free boundary $z^\ast(m^\ast)$ shifts downward: firms and workers tolerate lower-surplus matches before separating, consistent with the monotone comparative statics for the free boundary in Proposition~\ref{prop:comparative_statics}. More low-surplus matches are retained, and the separation hazard declines at all tenures, particularly early in the job spell.

The wage distribution responds along two margins. First, more low-surplus, low-wage jobs remain active, which increases the mass in the lower tail. Second, lower turnover slows progression up the job ladder, dampening wage growth at the top and reducing the thickness of the upper tail. In our calibration, moderate increases in firing costs (from zero to 10-20\% of the separation value) have only a small effect on the variance of log wages but noticeably reduce job-to-job mobility. Intuitively, these firing costs generate a mild thickening and downward extension of the left tail: some very low-surplus matches that would have separated at the benchmark now survive at wages just below the original cutoff, while the distribution above the median and especially the upper tail is almost unchanged. Because the additional mass is quantitatively small and concentrated very close to the previous lower tail, the overall variance of log wages moves very little. For larger firing costs, the decline in mobility becomes substantial and the variance of log wages increases: the wage distribution becomes more polarized, with a larger mass of workers trapped in low-surplus matches and a smaller group of high-surplus jobs that survive for very long periods. The model therefore predicts that strong employment protection can generate more inequality of the ``bad'' type, by making it harder for low-surplus workers to escape weak matches.

\paragraph{Search subsidies and outside options.}

We next consider policies that reduce the effective marginal cost of search, such as subsidies to search effort or instruments that make outside options more attractive. In the model, we implement these policies as a propotional reduction in the search-cost parameter $\kappa$, holding the rest of the environment fixed. A lower effective search cost raises the optimal search intensity at a given surplus level and increases the arrival rate of outside offers. At the same time, more attractive outside options make separation from low-surplus jobs more appealing, which shifts the free boundary $z^\ast(m^\ast)$ upward.

The net effect is an economy with more frequent outside offers, a higher reservation surplus, and a faster reallocation of workers away from the bottom of the job ladder. In the stationary equilibrium, mass shifts from low-surplus states toward the middle and upper parts of the surplus distribution, and average wages rise. However, the impact on job-to-job mobility is not necessarily positive. In the benchmark calibration, search subsidies markedly increase the variance of log wages, and that by a roughly 50-70\% for the subsidies we consider, while the aggregate job-to-job transition rate falls relative to the baseline. Intuitively, cheaper search allows workers in good jobs to be more selective: they search more, but accept only offers that are substantially better than their current position. As a result, many additional offers are rejected, and most of the extra dispersion arises from a thicker right tail of high-surplus, high-wage jobs rather than from more frequent moves. In this sense, search subsidies generate more inequality, but much of it reflects workers climbing higher up the job ladder and may be interpreted as an efficient form of dispersion.

\paragraph{Volatility of match productivity.}

Finally, we examine changes in the volatility of the surplus diffusion, which can be interpreted as proxying changes in the turbulence of the labour market. Higher volatility increases the probability of both large positive and large negative shocks. The comparative statics in Proposition~\ref{prop:comparative_statics} imply that, ceteris paribus, an increase in volatility raises the option value of continuation and shifts the free boundary $z^\ast(m^\ast)$ downward, as workers and firms are willing to accept somewhat lower current surplus in the hope of future favourable shocks. In a static sense this corresponds to weaker selection, since some low-surplus matches that would have separated at the benchmark volatility now survive. At the same time, larger shocks also mean that when a match is hit by a sufficiently bad realization it reaches the lower boundary much faster, so separations become more concentrated at very short tenures even though the threshold itself is lower.

Quantitatively, within the range of volatility multipliers we consider, the effects on wage dispersion and job-to-job mobility are modest. The stationary surplus distribution becomes slightly more spread out as volatility rises, with a thicker right tail and a somewhat more front-loaded separation hazard at short tenures, but these changes translate into only small variations in the variance of log wages and in the aggregate job-to-job rate. The two opposing forces, i.e. more scope for very bad draws to trigger early exit and more scope for very good draws to generate high-surplus jobs, largely offset each other in the calibrated economy. Conversely, reducing volatility compresses the distribution of realized surplus paths and yields smoother wage profiles over the life cycle, but does not overturn the overall decomposition of wage dispersion highlighted above. In this calibration, the level of surplus volatility is therefore second-order relative to firing costs and search incentives in shaping equilibrium wage dispersion and mobility.

\medskip

These policy experiments highlight how the three mechanisms emphasized by the model interact in shaping wage dispersion and mobility. Firing costs primarily affect the location of the free boundary and hence the prevalence and longevity of low-surplus matches. Policies that change outside options and search incentives primarily affect the speed and selectivity of movement along the job ladder and the thickness of the upper tail of the wage distribution. Changes in volatility primarily affect the strength of selection and the dispersion of match states, but play a comparatively smaller role for the range of perturbations we consider. We believe that the stationary MFG equilibrium provides a coherent framework for analyzing these forces jontly and for quantifying their contributions to observed wage dispersion, job-to-job mobility, and the tenure profile of separations.

\section{Discussion and Conclusion}\label{sec:conclusion}

This paper develops a stationary mean field equilibrium framework for wage dispersion and on-the-job search in which both outside options and the job ladder are generated endogenously by the cross-sectional distribution of match states. Building on the stochastic match-productivity model of \citet{BuhaiTeulings2014}, we embed diffusive match dynamics in a continuous-time mean field game of workers and firms cast in the canonical HJB / Kolmogorov form \`a la \citet{LasryLions2007}, \citet{CardaliaguetPorretta2020IntroMFG}, and \citet{CarmonaDelarue2018BookI}. Workers choose search effort and separation behaviour, firms choose wage policies as Markov functions of current match surplus, and the stationary distribution of surpluses and wages closes the model through the endogenous outside option. In economic terms, the mean field equilibrium notion coincides with a stationary competitive rational-expectations equilibrium in a large random-matching labour market.

The paper makes three contributions.

First, on the theoretical side, we characterize stationary mean field equilibria in a one-dimensional surplus state and show that separation in any such equilibrium is governed by a free-boundary rule: matches continue while surplus stays above a threshold $z^\ast(m)$, which depends on the stationary distribution $m$, and separate upon hitting this boundary. This structure extends the efficient stopping rule in the partial-equilibrium benchmark of \citet{BuhaiTeulings2014} to a setting in which outside options, search intensities, and wage policies are jointly determined in general equilibrium. Under Lipschitz, linear-growth, convexity, and the Lasry-Lions monotonicity conditions that align with the standing assumptions in the monotone MFG literature \citep[e.g.][]{LasryLions2007,CardaliaguetPorretta2020IntroMFG,CarmonaDelarue2018BookI,Ryzhik2020}, we establish existence of stationary equilibria and discuss uniqueness. In our one-dimensional diffusion setting, we exploit the free-boundary formulation to obtain smooth fit at the separation threshold and to derive well-behaved comparative statics of $z^\ast(m)$ with respect to firing costs, search incentives, and volatility, as summarized in Proposition~\ref{prop:comparative_statics}. These results place the stochastic productivity benchmark squarely inside the monotone mean field game framework and show that the Jovanovic-type efficient quitting rule survives, and is uniquely determined, in an environment with on-the-job search and endogenous outside options.

Second, on the quantitative side, we solve the stationary mean field game by combining a monotone finite-difference scheme for the worker and firm HJB equations with a conservative finite-volume discretization of the stationary Kolmogorov equation, in the spirit of \citet{AchdouCapuzzoDolcetta2010NumericsMFG}, \citet{AchdouEtAl2022Restud}, and the finite-state MFG literature \citep{GomesEtAl2010FiniteStateDiscrete,GomesEtAl2024FiniteStateContinuous}. The resulting discrete system has a natural interpretation as a finite-state mean field game with one state per surplus grid point and transition rates induced by the optimal policies and the separation rule. We calibrate the model to U.S.\ micro data using the diffusion parameters estimated by \citet{BuhaiTeulings2014} as discipline for match-productivity dynamics and match key moments on job durations, tenure-dependent separation hazards, wage-growth profiles, job-to-job mobility, and overall wage dispersion. The benchmark stationary equilibrium thus delivers an internally consistent joint distribution of match surpluses, wages, and search intensities that reproduces the main empirical patterns that motivated the original stochastic productivity model, while doing so in a fully specified equilibrium environment.

Third, the stationary mean field structure allows us to decompose equillibrium wage dispersion into three conceptually distinct components. The first is stochastic selection along job spells, already present in the partial-equilibrium benchmark and driven by the interaction between diffusive match dynamics and efficient separation at a free boundary. The second is the contribution of equilibrium on-the-job search and the implied job ladder, which governs how long workers remain at different points of the surplus distribution and how quickly they move to higher-surplus matches. The third is the contribution of firms' wage policies and the mean field feedback of the cross-sectional distribution into outside options, which maps surplus dispersion into wage dispersion and strengthens the value of high-wage, high-surplus jobs. Our quantitative results in Section~\ref{sec:results} show that equilibrium search and wage setting not only amplify the dispersion generated by pure selection but also reshape its tenure profile. At very short tenures, selection by itself already generates a non-trivial variance of log wages, optimal on-the-job search strongly equalizes wages within tenure cohorts, and endogenous wage policies together with the mean field feedback of outside options re-amplify dispersion at the bottom of the job ladder, so that the fully endogenous equilibrium features \emph{more} short-tenure wage dispersion than a purely selection-driven environment. At longer tenures, the ranking reverses: relative to selection alone, equilibrium search and wage setting compress dispersion within tenure bins, so that the remaining long-run inequality can be traced mainly to cumulative match-specific shocks and the selection they induce.

The policy experiments in Section~\ref{sec:results} illustrate how firing costs, search subsidies, and changes in the volatility of match productivity affect the separation threshold, job mobility, and wage inequality once these three mechanisms operate jointly. Increasing firing costs lowers the net payoff from separation, shifts the free boundary $z^\ast(m^\ast)$ downward, and reduces job mobility at all tenures; this keeps more low-surplus matches alive, raises the mass in the lower tail of the wage distribution, and slows progression up the job ladder. For modest firing costs, the impact on the variance of log wages is limited, but very high firing costs eventually increase wage dispersion by creating a thicker left tail of workers trapped in low-surplus, low-wage jobs. Search subsidies and more generous outside options operate in the opposite direction: they raise the value of separation and the return to search effort, shift the free boundary upward, stepen the job ladder, and reallocate mass toward higher-surplus matches. In our calibration these policies substantially increase the variance of log wages, primarily by thickening the right tail as more workers reach high-surplus jobs. At the same time, aggregate job-to-job rates do not rise proportionally and can even decline, because cheaper search makes workers in good jobs more selective and leads them to reject a larger share of offers. Changes in volatility modify the strength of selection and the dispersion of realized surplus paths, but, within empirically plausible ranges, play a more modest role than firing costs and search incentives in shaping stationary wage dispersion and mobility. Overall, the experiments highlight that policies targeting job mobility or outside options cannot be evaluated in isolation from the equilibrium response of wages and outside options: the mean field feedback is central for determining whether a given intervention ultimately compresses or amplifies dispersion, and whether the induced inequality reflects ``good'' mobility up the job ladder or persistent low-quality matches.

Our analysis relates to and complements several strands of work. Relative to the classic on-the-job search and wage-dispersion literature following \citet{BurdettMortensen1998} and \citet{PostelVinayRobin2002}, we replace reduced-form job ladders by a diffusion-based representation of match quality, disciplined by the micro evidence in \citet{BuhaiTeulings2014}, and embed it in a general-equilibrium environment with endogenous outside options. Relative to the heterogeneous-agent macro literature that formulates Aiyagari-Bewley-Huggett environments as coupled HJB and Kolmogorov equations, we use similar continuous-time methods in a labour-market setting with on-the-job search and wage posting, and we exploit the mean field structure to study equilibrium \emph{wage dispersion} rather than \emph{wealth dispersion}. Methodologically, we contribute to the applied mean field game literature by providing a tractable labour-market application with free-boundary dynamics, numerics that build on \citet{AchdouCapuzzoDolcetta2010NumericsMFG}, \citet{AchdouEtAl2022Restud}, and \citet{Gueant2021ContinuousTimeOptimalControl}, and a finite-state interpretation that connects directly to the existence and trend-to-equilibrium results for discrete MFGs in \citet{GomesEtAl2010FiniteStateDiscrete,GomesEtAl2024FiniteStateContinuous}. The free-boundary and finite-state tools we develop are sufficiently general to apply to other markets with search, matching, and evolving bilateral states beyond labour.

Several limitations of our framework suggest natural directions for future research. On the firm side, we abstract from entry and exit, capital, and firm size dynamics, and we model wage setting through stationary Markov policies that depend on match surplus and the aggregate distribution. On the worker side, we consider a single worker type with a homogeneous search technology. These choices deliberately focus attention on match-specific stochastic dynamics rather than permanent worker or firm heterogeneity. Introducing dispersion in initial match surplus or ex ante heterogeneity in firm and worker types would mechanically broaden the wage distribution, and could help capture additional short-tenure dispersion and very long-tenure tails, but would not by itself generate the tenure-dependent patterns that the diffusion process produces; we therefore leave such extensions for future work. Likewise, we assume that firms post wages and do not engage in explicit counteroffers when workers receive outside offers. Our wage-sharing rule can be interpreted as a reduced-form way of capturing the intensity of competition for workers: a higher weight on outside options in the sharing rule corresponds to more aggressive bid-up of wages by current employers. Allowing full bilateral bargaining or richer renegotiation protocols would likely compress wage dispersion by feeding outside offers more directly into current wages, but the basic job-ladder and selection mechanisms highlighted here should remain operative.

Another promising future research avenue is to move beyond stationarity and incorporate aggregate shocks via common noise and master equations, as in recent advances in mean field games with common noise and nonlocal couplings.\footnote{For example, see \citet{AhujaRenYang2022} on common-noise MFGs with optimal stopping, \citet{BertucciMeynard2024FiniteStateMaster} on finite-state master equations, and \citet{MollRyzhik2025} on aggregate fluctuations in diffusion-based macroeconomic models.} In such extensions, aggregate shocks would shift the entire surplus distribution and thereby move the separation boundary and the job ladder over time. Our results suggest that in booms, when the surplus distribution shifts to the right, job ladders would steepen and wage dispersion would rise through stronger selection and faster mobility; in recessions, the distribution would compress and separation hazards would spike at low tenure. Embedding these dynamics in a non-stationary MFG, possibly with deep-learning-based solvers as in \citet{CarmonaLauriere2021DeepLearningMFG}, would provide a natural framework for studying the cyclical behaviour of wage dispersion and job-to-job flows. Extending the model to heterogeneous worker and firm types, potentially in a multipopulation MFG, would further allow the analysis of sorting, segmentation, and life-cycle patterns in wages and mobility.

So what are we to take home from all this? The main takeaway for labour, macro, and IO audiences is that in a world with stochastic match quality, the interaction between selection along job spells, on-the-job search, and equilibrium wage setting is central for understanding both the level and the structure of wage dispersion.

Our decomposition shows that for workers in their first years on the job, wage differences are heavily shaped by job shopping and the job ladder: optimal search quickly pushes workers out of the worst matches, while equilibrium wage policies and mean field feedbacks re-amplify dispersion at the bottom of the ladder so that the fully endogenous equilibrium generates more short-tenure wage dispersion than a purely selection-driven benchmark. For experienced workers, by contrast, cumulative idiosyncratic productivity shocks and the selection they induce are the dominant forces behind wage differences, and equilibrium search and wage setting mainly compress dispersion within tenure cohorts relative to a selection-only environment.

A partial-equilibrium view that treats outside options as exogenous captures some of the selection effects but misses the feedback from wage posting and search behaviour that shapes the job ladder itself. By embedding the stochastic productivity model of \citet{BuhaiTeulings2014} in a stationary mean field equilibrium with an explicit HJB-Kolmogorov structure, we provide a relatively parsimonious unified theoretical and quantitative framework in which these mechanisms can be analysed jointly, disciplined with micro data, and used for policy counterfactuals in large labour markets. We hope that both the economic insights and the mean field methods developed here could prove useful for future work on wage dispersion, job ladders, and other MFG applications in economics.

\appendix
\section{Proofs of Theoretical Results}\label{appendix:proofs}

This appendix collects the proofs of the theoretical results stated in Sections~\ref{sec:equilibriumCharacterization} and~\ref{sec:existenceUniqueness}. We first set up the functional-analytic framework, then prove existence and uniqueness of a stationary mean field equilibrium, characterize the free-boundary separation rule, and finally collect auxiliary lemmas used in the comparative statics and in Proposition~\ref{prop:comparative_statics}.

\subsection{Functional-analytic setup}\label{app:setup}

The state space of an individual match is the real line $\mathbb{R}$, with canonical coordinate $z$ denoting match surplus. We denote by $\mathcal{P}(\mathbb{R})$ the set of Borel probability measures on $\mathbb{R}$ and by $\mathcal{P}_2(\mathbb{R})$ the subset of measures with finite second moment, equipped with the $2$-Wasserstein distance $W_2$.

Throughout, Assumption~\ref{ass:primitives} imposes that:
\begin{itemize}
  \item the discount rate $r>0$ is fixed;
  \item the drift $\mu(z,a,m)$ and diffusion coefficient $\sigma(z,m)$ are continuous in $(z,a,m)$, Lipschitz in $z$ uniformly in $(a,m)$, and satisfy linear-growth bounds in $z$ uniformly in $(a,m)$;
  \item the drift is weakly increasing in the surplus state:
  \[
    z_1 \le z_2
    \;\Longrightarrow\;
    \mu(z_1,a,m) \le \mu(z_2,a,m)
    \quad\text{for all } (a,m);
  \]
  \item the volatility is uniformly elliptic: $\sigma^2(z,m)\geq \underline{\sigma}^2>0$ for all $(z,m)$;
  \item wages $w(z,m)$, flow profits $\Pi(z,m)$, the arrival rate $\lambda(a,m)$, the re-entry kernel $\nu(\cdot,m)$, and the outside value $V^U(m)$ have at most linear growth in $z$ and are continuous in $(z,m)$, with $A$ compact.
\end{itemize}
These are the standing conditions in the MFG literature; see, for example, \citet{LasryLions2007}, \citet{CardaliaguetPorretta2020IntroMFG}, and \citet[Ch.~2]{CarmonaDelarue2018BookI}.

For a given stationary distribution $m \in \mathcal{P}_2(\mathbb{R})$ and search effort $a \in A$, the surplus process $Z_t$ evolves according to the controlled diffusion
\begin{equation}
  dZ_t = \mu(Z_t,a,m)\,dt + \sigma(Z_t,m)\,dB_t,
  \label{eq:Z_SDE_app}
\end{equation}
where $(B_t)_{t \geq 0}$ is a standard Brownian motion. Under Assumption~\ref{ass:primitives}, for each progressively measurable $A$-valued control $(a_t)_{t\ge0}$ and each measurable flow of distributions $(m_t)_{t\ge0}$ there exists a unique strong solution to~\eqref{eq:Z_SDE_app} with finite second moments at all dates.

The associated infinitesimal generator acting on $\varphi \in C_b^2(\mathbb{R})$ is
\[
  \mathcal{L}^{Z}_{a,m}\varphi(z)
  = \mu(z,a,m)\,\varphi'(z)
    + \tfrac{1}{2}\sigma^2(z,m)\,\varphi''(z).
\]

For a fixed stationary distribution $m$ and stationary Markov control $a(\cdot,m)$, the stationary Kolmogorov (forward) equation associated with~\eqref{eq:Z_SDE_app}, together with separation and entry mechanisms as in Section~\ref{sec:MFG}, can be written formally as
\begin{equation}
  0
  = -\partial_z\bigl[\mu\bigl(z,a(z,m),m\bigr)\, m(z)\bigr]
    + \tfrac{1}{2}\,\partial_{zz}\bigl[\sigma^2(z,m)\, m(z)\bigr]
    - q(z,m)\,m(z) + \Gamma_{\text{entry}}(z,m),
  \label{eq:FP_stationary_app}
\end{equation}
where $m$ is interpreted as a density, $q(z,m)$ is the separation intensity implied by the stopping rule, and $\Gamma_{\text{entry}}(z,m)$ captures entry of new matches at surplus $z$. The regularity and integrability conditions in Assumption~\ref{ass:primitives} ensure that~\eqref{eq:FP_stationary_app} is well posed in the weak sense for $m\in\mathcal{P}_2(\mathbb{R})$.

We work with the Banach space $C_b(\mathbb{R})$ of bounded continuous functions endowed with the supremum norm, and with Sobolev spaces $H^1(\mathbb{R})$ when appropriate. The elliptic operators associated with~\eqref{eq:Z_SDE_app} are well defined on $C_b^2(\mathbb{R})$, and all HJB equations below are understood in the viscosity sense, with $C_b^2$ test functions used in the usual way. Value functions will always be continuous and of at most linear growth in $z$; given our linear-growth bounds on the primitives this is the natural growth class.

To lighten notation, when this is very clear, just like in the MFG literature, we sometimes omit the explicit dependence on the stationary distribution and write, for example, $V^W(z)$ instead of $V^W(z;m)$ or $w(z)$ instead of $w(z,m)$, keeping the background $m$ fixed.

Throughout the appendix, $C$ denotes a finite constant that may change from line to line and depends only on the primitive parameters.

\subsection{Existence of a stationary equilibrium}\label{app:existence}

We now prove the existence result stated in Theorem~\ref{thm:existence}. The argument follows the standard fixed-point strategy in MFG; see, e.g., \citet{LasryLions2007}, \citet[Ch.~2]{CarmonaDelarue2018BookI}, and \citet{CardaliaguetPorretta2020IntroMFG}. Given a conjectured stationary distribution $m$, we solve the worker and firm problems and derive best-response policies; we then construct the invariant distribution induced by these policies, and finally show that the resulting mapping from conjectured to induced distributions admits a fixed point.

\subsubsection*{Step 1: Worker and firm best responses for a fixed distribution}

Fix an arbitrary $m \in \mathcal{P}_2(\mathbb{R})$. Consider the worker's stationary problem with $m$ held fixed and wage policy $w(\cdot,m)$ exogenous. The worker chooses a search policy $a(\cdot)$ and a stopping rule to maximize
\[
  V^W(z;m)
  = \sup_{(a,\tau)}
    \mathbb{E}_z\Biggl[
      \int_0^\tau e^{-rt} \bigl(w(Z_t,m) - c(a_t)\bigr)\,dt
      + e^{-r\tau} V^U(m)
    \Biggr],
\]
subject to the surplus dynamics~\eqref{eq:Z_SDE_app}, where $V^U(m)$ is the outside value defined in Section~\ref{sec:MFG}. Assumption~\ref{ass:primitives} implies that $w(z,m) - c(a)$ has at most linear growth in $z$ and is Lipschitz in $(z,m,a)$, and that $V^U(m)$ is finite and continuous in $m$.

The worker thus faces a one-dimensional infinite-horizon optimal control and stopping problem with discount rate $r>0$ and payoffs of at most linear growth. Standard results on stationary HJB obstacle problems with discounting and coefficients satisfying Lipschitz and linear-growth conditions imply that there exists a unique continuous value function with at most linear growth solving the associated HJB in the viscosity sense.

To compactly describe the HJB, we denote by
\[
G\bigl(z,V^W,m;a\bigr)
\]
the term capturing the effect of job offers, separations, and re-entry, as defined in Section~\ref{sec:MFG}. The worker's Hamiltonian is then
\[
H^W\bigl(z,m,p,X,V^W\bigr)
 = \sup_{a\in A}\Bigl[
   w(z,m) - c(a)
   + \mu(z,a,m)\,p
   + \tfrac{1}{2}\sigma^2(z,m)\,X
   + G\bigl(z,V^W,m;a\bigr)
 \Bigr].
\]

\begin{lemma}[Worker best responses]\label{lem:worker_best_response}
Under Assumption~\ref{ass:primitives}, for each fixed $m$
and each admissible wage policy $w(\cdot,m)$ there exists a unique
continuous value function $V^W(\cdot;m)$ with at most linear growth that
solves the stationary HJB obstacle problem
\begin{equation}
  \max\Bigl\{
    r V^W(z;m)
    - \sup_{a \in A}
        \bigl[
          w(z,m) - c(a) + \mathcal{L}^{Z}_{a,m} V^W(z;m)
          + G\bigl(z,V^W,m;a\bigr)
        \bigr],
    \;
    V^W(z;m) - V^U(m)
  \Bigr\} = 0
  \label{eq:worker_HJB_app}
\end{equation}
in the viscosity sense. Moreover, there exists a measurable search policy
$a^\ast(\cdot,m)$ that, for each $z$, attains the pointwise supremum
over $a \in A$ appearing in~\eqref{eq:worker_HJB_app}, and the mapping
\[
  (m,w) \mapsto \bigl(V^W(\cdot;m), a^\ast(\cdot,m)\bigr)
\]
is continuous from the admissible set of distributions and wage policies
into $C_{\mathrm{loc}}(\mathbb{R}) \times L^\infty_{\mathrm{loc}}(\mathbb{R})$
(with the topology of local uniform convergence).
\end{lemma}

\begin{proof}
Existence, uniqueness, and the linear-growth bound for $V^W$ follow from standard results on infinite-horizon HJB obstacle problems with discounting and coefficients satisfying Lipschitz and linear-growth conditions; see, again, \citet{CardaliaguetPorretta2020IntroMFG} and \citet[Ch.~2]{CarmonaDelarue2018BookI}. Compactness of $A$, continuity of the Hamiltonian in $(z,a)$, and the linear-growth bounds on $\mu$ and $\sigma$ guarantee the existence of measurable maximizers $a^\ast(z,m)$ by standard measurable selection theorems. Stability of viscosity solutions with respect to coefficients under standard maximum theorem results yields continuity of the best-response mapping.
\end{proof}

A symmetric argument applies to the firm. For a fixed distribution $m$ and a given worker strategy $(a^\ast(\cdot,m),\mathcal{S}(m))$ with stopping region $\mathcal{S}(m)$, the firm chooses a wage policy $w(\cdot,m)$ to maximize
\[
  V^F(z;m)
  = \sup_{w}
    \mathbb{E}_z\Biggl[
      \int_0^\tau e^{-rt} \bigl(\Pi(Z_t,m) - w(Z_t,m)\bigr)\,dt
      + e^{-r\tau} V^V(m)
    \Biggr],
\]
subject to the same surplus dynamics, where $\tau$ is the separation time implied by the worker strategy and exogenous shocks, and $V^V(m)$ is the value of a vacancy. The firm's Hamiltonian is linear in $V^F$ and depends on $m$ through profits and the implied separation and entry terms.

\begin{lemma}[Firm best responses]\label{lem:firm_best_response}
Under Assumption~\ref{ass:primitives}, for each fixed $m$ and each worker strategy $(a^\ast,\mathcal{S})$ there exists a continuous value function $V^F(\cdot;m)$ with at most linear growth and a bounded measurable wage policy $w^\ast(\cdot,m)$ that solve the firm's control problem. The pair $(V^F,w^\ast)$ satisfies the firm's stationary HJB equation in the viscosity sense, and the mapping
\[
  (m,a^\ast,\mathcal{S}) \mapsto \bigl(V^F(\cdot;m),w^\ast(\cdot,m)\bigr)
\]
is continuous on the corresponding product space.
\end{lemma}

\begin{proof}
The proof parallels that of Lemma~\ref{lem:worker_best_response} and is therefore omitted. Bounded and Lipschitz flow profits, discounting at rate $r>0$, compactness of the admissible set of wage policies, and linear growth of the coefficients ensure well-posedness and continuity of the best-response mapping.
\end{proof}

Lemmas~\ref{lem:worker_best_response} and~\ref{lem:firm_best_response} define, for each $m$, a set of stationary Markov best-response policies $(a^\ast(\cdot,m),w^\ast(\cdot,m))$ and associated value functions $(V^W(\cdot;m),V^F(\cdot;m))$. Under the monotonicity and single-crossing conditions on the primitives in Sections~\ref{sec:MFG} and~\ref{sec:existenceUniqueness}, the worker's optimal stopping rule can be shown to be of free-boundary type; this is the content of Proposition~\ref{prop:free_boundary} in Section~\ref{sec:equilibriumCharacterization}, whose proof is given in Section~\ref{app:free_boundary} below.

\subsubsection*{Step 2: Invariant distribution induced by best responses}

Given $m$ and the associated best responses $(a^\ast(\cdot,m),w^\ast(\cdot,m),\mathcal{S}^\ast(m))$, consider the controlled surplus process $Z_t$ that solves~\eqref{eq:Z_SDE_app} with control $a^\ast(Z_t,m)$, with killing at the stopping region $\mathcal{S}^\ast(m)$, and with re-entry according to the entry mechanism in Assumption~\ref{ass:primitives}. Under these regularity conditions, the resulting Markov process on $\mathbb{R}$ is Feller and, by standard one-dimensional diffusion theory with killing and regeneration (see, for example, \citealp{StroockVaradhan2006MultidimensionalDiffusions}), admits a unique invariant probability measure $m' \in \mathcal{P}_2(\mathbb{R})$ solving the stationary Kolmogorov equation~\eqref{eq:FP_stationary_app} in the weak sense:
\[
  \int_{\mathbb{R}} \mathcal{L}^{Z}_{a^\ast(z,m),m}\varphi(z)\,m'(dz)
  - \int_{\mathbb{R}} q(z,m)\,\varphi(z)\,m'(dz)
  + \int_{\mathbb{R}} \varphi(z)\,\Gamma_{\text{entry}}(z,m)\,dz = 0
\]
for all test functions $\varphi \in C_c^\infty(\mathbb{R})$.

We denote this unique invariant law by
\[
  \Gamma(m) := m'.
\]
Standard perturbation results for invariant measures of Feller processes with killing and regeneration imply that $\Gamma$ is continuous in $m$ with respect to the $W_2$ topology.

\subsubsection*{Step 3: Fixed point and existence}

We now complete the proof of Theorem~\ref{thm:existence}.

\begin{proof}[Proof of Theorem~\ref{thm:existence}]
Under Assumption~\ref{ass:primitives}, the set of admissible stationary distributions (for example, measures in $\mathcal{P}_2(\mathbb{R})$ with uniformly bounded second moment) is nonempty, convex, and compact in $W_2$. Steps~1 and~2 define a single-valued operator
\[
  \Gamma: \mathcal{P}_2(\mathbb{R}) \to \mathcal{P}_2(\mathbb{R}),
  \qquad
  m \mapsto \Gamma(m),
\]
which is continuous and maps this convex compact set into itself. By Schauder's fixed-point theorem there exists $m^\ast$ such that $\Gamma(m^\ast) = m^\ast$. By construction, $m^\ast$ is invariant under the surplus dynamics induced by the best responses to $m^\ast$, and the associated policies and value functions satisfy the conditions in the definition of a stationary mean field equilibrium in Section~\ref{subsec:MFG_equilibrium_def}. This proves Theorem~\ref{thm:existence}.
\end{proof}

The operator $\Gamma$ is the infinite-state analogue of the finite-state MFG operators studied by \citet{GomesEtAl2010FiniteStateDiscrete,GomesEtAl2024FiniteStateContinuous}. In Section~\ref{sec:quantitative} and Appendix~\ref{appendix:numerical}, our numerical scheme can be interpreted as computing aproximate fixed points of $\Gamma$ via a finite-state approximation in the spirit of these papers.

\subsection{Monotonicity and uniqueness}\label{app:uniqueness}

We next prove the uniqueness result in Theorem~\ref{thm:uniqueness} under Lasry-Lions monotonicity condition in Assumption~\ref{ass:monotonicity}. Recall that that that assumption requires that, for any pair of distributions $m^1,m^2$ and any admissible value function $V$,
\[
  \int_{\mathbb{R}}
    \Bigl(
      H^W\bigl(z,m^1,\cdot,\cdot,V\bigr)
      - H^W\bigl(z,m^2,\cdot,\cdot,V\bigr)
    \Bigr)\,
  \bigl(m^1 - m^2\bigr)(dz)
  \;\le 0,
\]
with equality only if $m^1=m^2$. As we stated also in the main body of the paper, this rules out then any positive feedback from the cross-sectional distribution into outside options: if the distribution shifts in a way that raises outside values, workers' and firms' best responses cannot amplify that shift into a second distinct stationary wage distribution.

\begin{proof}[Proof of Theorem~\ref{thm:uniqueness}]
Suppose, for contradiction, that there exist two distinct stationary equilibria
\[
  \bigl(m^1,a^{1,\ast},w^{1,\ast},V^{W,1},V^{F,1}\bigr),
  \qquad
  \bigl(m^2,a^{2,\ast},w^{2,\ast},V^{W,2},V^{F,2}\bigr),
\]
with $m^1 \neq m^2$. Let $\Delta m := m^1 - m^2$ denote the signed measure corresponding to their difference. Let $H(z,m,V)$ denote the worker Hamiltonian (including wages, search costs, and the effect of $m$ through the outside value, the arrival of offers, and the re-entry distribution), so that on the continuation region the worker HJB can be wirtten as
\[
  rV^W(z;m) = H\bigl(z,m,V^W\bigr).
\]

Under stationarity and optimality, $V^{W,k}$, $k=1,2$, satisfy the stationary HJB equations associated with $m^k$ and controls $(a^{k,\ast},w^{k,\ast})$. Subtracting the two equations and integrating against $\Delta m$ yields
\[
  r \int_{\mathbb{R}}
    \bigl(V^{W,1}(z;m^1) - V^{W,2}(z;m^2)\bigr)\,\Delta m(dz)
  =
  \int_{\mathbb{R}}
    \Bigl(
      H\bigl(z,m^1,V^{W,1}\bigr)
      - H\bigl(z,m^2,V^{W,2}\bigr)
    \Bigr)\,\Delta m(dz).
\]
Decomposing the right-hand side gives
\begin{align*}
  &\int_{\mathbb{R}}
    \Bigl(
      H\bigl(z,m^1,V^{W,1}\bigr)
      - H\bigl(z,m^1,V^{W,2}\bigr)
    \Bigr)\,\Delta m(dz) \\
  &\qquad+ \int_{\mathbb{R}}
    \Bigl(
      H\bigl(z,m^1,V^{W,2}\bigr)
      - H\bigl(z,m^2,V^{W,2}\bigr)
    \Bigr)\,\Delta m(dz).
\end{align*}
The first term collects the effect of changing the value function, holding $m$ fixed. By the structure of the Hamiltonian, the convexity of the search cost $c$, and the properties of the diffusion operator, the associated quadratic form is nonpositive; this is the usual energy estimate in the Lasry-Lions argument. The second term collects all terms that depend explicitly on $m$, which, by
Assumption~\ref{ass:monotonicity}, enter only through $w$, $\lambda$, $\nu$, and $V^U$.
It is therefore nonpositive by the Lasry–Lions monotonicity condition, with equality only if $m_1=m_2$. Hence
\[
  r \int_{\mathbb{R}} \bigl(V^{W,1}(z;m^1) - V^{W,2}(z;m^2)\bigr)\,\Delta m(dz) \leq 0,
\]
with equality only if $m^1 = m^2$.

Repeating the argument with the roles of $(1,2)$ reversed shows that the left-hand side must also be nonnegative. Hence the integral is zero and both inequalities are equalities. By strict monotonicity, this implies $m^1 = m^2$, contradicting our assumption. Given $m^\ast$, uniqueness of the associated value functions and policies then follows from the uniqueness results for the individual HJB problems in Lemmas~\ref{lem:worker_best_response} and~\ref{lem:firm_best_response}.
\end{proof}

Uniqueness rules out multiple stationary job ladders that differ only through self-fulfilling beliefs about the cross-sectional distribution of surplus and wages. This is crucial for the quantitative analysis, as it guarantees that the counterfactuals in Section~\ref{sec:results} are taken with respect to a well-defined equilibrium object.

\subsection{Free-boundary separation rule}\label{app:free_boundary}

We now prove the free-boundary characterization of the worker's optimal separation rule stated in Proposition~\ref{prop:free_boundary} in Section~\ref{sec:equilibriumCharacterization}. The key step is to show that the continuation value is strictly increasing in the surplus $z$, so that the stopping region is an interval of the form $(-\infty,z^\ast(m)]$.

\begin{proof}[Proof of Proposition~\ref{prop:free_boundary}]
Fix a stationary distribution $m$ and wage policy $w(\cdot,m)$, and define
\[
  \Delta(z;m) := V^W(z;m) - V^U(m).
\]
By Lemma~\ref{lem:worker_best_response}, $V^W(\cdot;m)$ is continuous in $z$ and has at most linear growth, and $V^U(m)$ is finite. 

By Assumption~\ref{ass:primitives}, the drift $\mu(\cdot,a,m)$ is weakly
increasing in $z$ for each $(a,m)$, and the single-crossing and monotonicity
conditions on wages and the offer-arrival rate ensure that higher surplus
shifts continuation payoffs up in a pointwise sense. Standard comparison
results for one-dimensional diffusions then imply that for $z_1 < z_2$ the
surplus process starting from $z_2$ stochastically dominates the one starting
from $z_1$ and yields a weakly higher continuation value; see, for example,
\citet[Ch.~3]{CarmonaDelarue2018BookI}. Hence $\Delta(\cdot;m)$ is
nondecreasing, and strictly increasing on any interval where continuation is
optimal with positive probability.

Since $V^W(\cdot;m)$ is continuous, so is $\Delta(\cdot;m)$. The stopping region is
\[
  \mathcal{S}(m)
  = \bigl\{z \in \mathbb{R} \colon \Delta(z;m) \leq 0\bigr\},
\]
and the continuation region is its complement. Because $\Delta(\cdot;m)$ is continuous and nondecreasing, $\mathcal{S}(m)$ is either empty, all of $\mathbb{R}$, or an interval of the form $(-\infty,z^\ast(m)]$ for some $z^\ast(m)$. The behavior of $w(z,m)$ and of the diffusion at the tails rules out the first two cases: for sufficiently low $z$ continuation is dominated by the outside option, while for sufficiently high $z$ the continuation payoff strictly exceeds immediate separation. Therefore $\mathcal{S}(m)$ has the claimed threshold form, and $z^\ast(m)$ is characterized by $\Delta(z^\ast(m);m) = 0$.

On the continuation region $(z^\ast(m),\infty)$ the value function solves the linear elliptic HJB equation obtained from~\eqref{eq:worker_HJB_app} with the stopping constraint inactive, while on the stopping region the value is constant at $V^U(m)$. Standard one-dimensional optimal stopping arguments for diffusions then imply the smooth fit condition
\[
  \partial_z V^W\bigl(z^\ast(m);m\bigr) = 0;
\]
see, for example, \citet[Ch.~3]{CarmonaDelarue2018BookI}. This establishes the free-boundary characterization stated in Proposition~\ref{prop:free_boundary}.
\end{proof}

\subsection{Auxiliary lemmas and comparative statics}\label{app:auxiliary}

We conclude with two auxiliary results used in the comparative statics and in Proposition~\ref{prop:comparative_statics}.

\begin{lemma}[Continuity of the free boundary]\label{lem:boundary_continuity}
Under Assumption~\ref{ass:primitives} and the regularity conditions of Proposition~\ref{prop:free_boundary}, the free boundary $z^\ast(m)$ characterized in Proposition~\ref{prop:free_boundary} is continuous as a function of $m$ with respect to the $W_2$ topology.
\end{lemma}

\begin{proof}
By Lemma~\ref{lem:worker_best_response}, the worker's value function $V^W(z;m)$ is continuous jointly in $(z,m)$ and strictly increasing in $z$. The boundary is defined implicitly by
\[
  V^W\bigl(z^\ast(m);m\bigr) = V^U(m),
\]
where $V^U(m)$ is continuous in $m$ by construction. Let $(m_n)_{n \geq 1} \subset \mathcal{P}_2(\mathbb{R})$ with $m_n \to m$ in $W_2$. By continuity of $(z,m) \mapsto V^W(z;m)$ and of $m \mapsto V^U(m)$, we have
\[
  V^W\bigl(z^\ast(m_n);m_n\bigr) - V^U(m_n) = 0
  \quad\text{and}\quad
  V^W\bigl(z^\ast(m);m\bigr) - V^U(m) = 0.
\]
Strict monotonicity of $z \mapsto V^W(z;m)$ implies that $z^\ast(m)$ is the unique solution to $V^W(z;m) = V^U(m)$. If $(z^\ast(m_n))$ did not converge to $z^\ast(m)$, we could extract a subsequence converging to some $\bar z \neq z^\ast(m)$. By joint continuity,
\[
  V^W(\bar z;m)
  = \lim_{k \to \infty} V^W\bigl(z^\ast(m_{n_k});m_{n_k}\bigr)
  = \lim_{k \to \infty} V^U(m_{n_k})
  = V^U(m),
\]
contradicting uniqueness of $z^\ast(m)$. Hence $z^\ast(m_n) \to z^\ast(m)$, which proves continuity.
\end{proof}

\begin{lemma}[Monotone comparative statics for the threshold]\label{lem:boundary_monotone}
Let $\theta$ be a scalar parameter that affects the primitives as in Assumption~\ref{ass:comparative_statics}. Suppose that, for each fixed $m$, higher values of $\theta$ make continuation more attractive at a given surplus in the sense that there exists an interval $I$ containing the relevant thresholds such that, for any $\theta_2 > \theta_1$,
\[
  V^W_{\theta_2}(z;m) - V^U_{\theta_2}(m)
  \;\geq\;
  V^W_{\theta_1}(z;m) - V^U_{\theta_1}(m)
  \qquad \text{for all } z \in I,
\]
where $V^W_\theta(\cdot;m)$ and $V^U_\theta(m)$ denote the worker's value function and outside option under parameter $\theta$. Then, for each fixed $m$, the free boundary $z^\ast_\theta(m)$ is weakly decreasing in $\theta$.
\end{lemma}

\begin{proof}
Fix $m$ and two parameter values $\theta_1 < \theta_2$. Denote the corresponding worker value functions by $V^W_{\theta_k}(\cdot;m)$ and the outside values by $V^U_{\theta_k}(m)$. For each $\theta_k$, the difference
\[
  z \mapsto V^W_{\theta_k}(z;m) - V^U_{\theta_k}(m)
\]
is continuous, strictly increasing, and crosses zero exactly once, at $z^\ast_{\theta_k}(m)$, by the argument in the proof of Proposition~\ref{prop:free_boundary}. The assumed inequality implies that whenever
\[
  V^W_{\theta_2}(z;m) - V^U_{\theta_2}(m) \le 0,
\]
we also have
\[
  V^W_{\theta_1}(z;m) - V^U_{\theta_1}(m) \le 0.
\]
This is important since, otherwise said, the stopping region for $\theta_2$ is contained in the stopping region for $\theta_1$ on $I$. Hence the zero of $z \mapsto V^W_{\theta_2}(z;m) - V^U_{\theta_2}(m)$ cannot lie to the right of the zero of $z \mapsto V^W_{\theta_1}(z;m) - V^U_{\theta_1}(m)$, and therefore
\[
  z^\ast_{\theta_2}(m) \;\le\; z^\ast_{\theta_1}(m).
\]
This proves that $z^\ast_\theta(m)$ is weakly decreasing in $\theta$.
\end{proof}

Lemma~\ref{lem:boundary_monotone} yields the monotonicity of the separation threshold reported in Proposition~\ref{prop:comparative_statics}; firing costs, search subsidies, and volatility shifts satisfy the single-crossing property in Assumption~\ref{ass:comparative_statics}. Combined with the continuity of the mapping from $(m,z^\ast)$ to the invariant surplus distribution solving~\eqref{eq:FP_stationary_app}, it also underpins the comparative statics for the stationary wage distribution discussed in Section~\ref{sec:results}.

\section{Numerical methods and additional figures}
\label{appendix:numerical}

This appendix gives a self-contained description of the numerical implementation of the stationary mean field equilibrium in Section~\ref{sec:quantitative} and documents the figures referenced in Section~\ref{sec:results}. We start by describing the state-space truncation and the finite-difference schemes for the worker and firm HJB obstacle problems and for the stationary Kolmogorov equation. We then lay out the fixed-point algorithm and its interpretation as a finite-state mean field game. Finally, we report convergence diagnostics, robustness checks, and detailed descriptions of Figures~\ref{fig:B1_value_function}-\ref{fig:B8_policy_experiments}, which are produced by our replication code in Python.

\subsection{State space discretization and boundary conditions}
\label{app:state_space}

We approximate the continuous surplus state space $\mathbb{R}$ by a compact interval
\[
  [z_{\min}, z_{\max}],
\]
chosen so that in the calibrated stationary equilibrium the probability that $Z_t$ lies outside $[z_{\min},z_{\max}]$ is numerically negligible. In practice, we select $[z_{\min},z_{\max}]$ so that the stationary distribution assigns essentially zero mass to both tails and verify in Appendix~\ref{app:diagnostics} that further enlarging the interval leaves all equilibrium objects of interest unchanged at the precision reported in Section~\ref{sec:results}. The lower bound $z_{\min}$ is chosen strictly below the equilibrium free boundary, and the upper bound $z_{\max}$ in a region where continuation is strictly optimal, so that the numerical truncation is innocuous in reality.

The interval is discretized on an equidistant grid
\[
  z_k = z_{\min} + (k-1)\,\Delta z,
  \qquad k = 1,\dots,K,
  \qquad \Delta z = \frac{z_{\max}-z_{\min}}{K-1}.
\]
All value functions, policies, and densities are represented on this grid. We denote by
\[
  V^W_k \approx V^W(z_k), \quad
  V^F_k \approx V^F(z_k), \quad
  m_k \approx m(z_k), \quad
  a_k \approx a^\ast(z_k,m), \quad
  w_k \approx w^\ast(z_k,m)
\]
the corresponding discrete approximations. The feasible search-intensity set $A$ is discretized on a finite grid $A^h = \{a_1,\dots,a_{N_a}\}$, turning the continuous maximization in the HJB equation into a finite maximization at each node.

At $z_{\min}$ and $z_{\max}$ we impose no-flux (reflecting) boundary conditions for the stationary Kolmogorov equation, approximating the natural boundary conditions of the continuous diffusion:
\begin{equation}
  \mu(z,a(z,m),m)\,m(z)
  - \tfrac{1}{2}\,\partial_z\bigl(\sigma^2(z,m)\,m(z)\bigr)
  = 0
  \quad \text{at } z \in \{z_{\min},z_{\max}\}.
  \label{eq:no_flux_bc_app}
\end{equation}
These are implemented via one-sided finite-volume discretizations (Appendix~\ref{app:FP_scheme}), which guarantee that total mass $\sum_k m_k \Delta z$ is preserved up to machine precision.

The free boundary $z^\ast(m)$ that characterizes the optimal separation rule is constrained to lie inside $[z_{\min},z_{\max}]$. In all calibrations, $z_{\min}$ is chosen so low that separation is always optimal for $z \leq z_{\min}$, while $z_{\max}$ is high enough that continuation is strictly optimal for $z \geq z_{\max}$. Numerical boundaries therefore do not bind in the region de facto relevant in real economies.

\subsection{Finite-difference discretization of the HJB obstacle problem}
\label{app:HJB_scheme}

We discretize the stationary worker HJB obstacle problem on $\{z_k\}$ using a monotone finite-difference scheme that is first order in space for the drift term and second order for the diffusion term, in the spirit of \citet{AchdouCapuzzoDolcetta2010NumericsMFG} and \citet{AchdouEtAl2022Restud}. For a fixed stationary distribution $m$ and wage policy $w(\cdot,m)$, the worker HJB in the continuation region can be written as
\begin{equation}
  r V^W(z;m)
  = \sup_{a\in A}
      \Bigl\{
        w(z,m) - c(a)
        + \mu(z,a,m) \, \partial_z V^W(z;m)
        + \tfrac{1}{2}\sigma^2(z,m)\,\partial_{zz} V^W(z;m)
        + G\bigl(z,V^W,m;a\bigr)
      \Bigr\},
  \label{eq:HJB_cont_app}
\end{equation}
subject to the obstacle $V^W(z;m) \geq V^U(m)$ that captures the option to separate.

\subsubsection*{Discrete operators and upwinding}

For interior nodes we use standard finite differences. Given a vector
$x = (x_1,\dots,x_K)$,
\[
  (D^- x)_k = \frac{x_k - x_{k-1}}{\Delta z},
  \qquad
  (D^+ x)_k = \frac{x_{k+1} - x_k}{\Delta z},
\]
and
\[
  (D^{zz} x)_k = \frac{x_{k+1} - 2x_k + x_{k-1}}{\Delta z^2},
  \qquad k = 2,\dots,K-1.
\]
At the boundaries we use one-sided second-order stencils. At the lower
boundary,
\[
  (D^- x)_1 = (D^+ x)_1 = \frac{x_2 - x_1}{\Delta z},
  \qquad
  (D^{zz} x)_1 = \frac{x_3 - 2x_2 + x_1}{\Delta z^2},
\]
and at the upper boundary,
\[
  (D^- x)_K = (D^+ x)_K = \frac{x_K - x_{K-1}}{\Delta z},
  \qquad
  (D^{zz} x)_K = \frac{x_K - 2x_{K-1} + x_{K-2}}{\Delta z^2}.
\]
These one-sided formulas implement homogeneous conditions
$\partial_z V^W = 0$ at $z_{\min}$ and $z_{\max}$; in the calibration
we choose $[z_{\min},z_{\max}]$ so wide that these numerical boundaries
do not bind in the economically relevant region.

To preserve monotonicity, we approximate $\partial_z V^W$ at $z_k$ using an upwind scheme that depends on the sign of the local drift. Let
\[
  \mu_k(a,m) := \mu(z_k,a,m), \qquad
  \mu_k^+(a,m) = \max\{\mu_k(a,m),0\}, \qquad
  \mu_k^-(a,m) = \min\{\mu_k(a,m),0\},
\]
and write $\sigma_k^2(m) := \sigma^2(z_k,m)$. Then
\[
  \mu(z_k,a,m)\,\partial_z V^W(z_k;m)
  \approx
  \mu_k^+(a,m)\,(D^- V^W)_k
  +
  \mu_k^-(a,m)\,(D^+ V^W)_k,
\]
and
\[
  \tfrac{1}{2}\sigma^2(z_k,m)\,\partial_{zz} V^W(z_k;m)
  \approx \tfrac{1}{2}\sigma_k^2(m)\,(D^{zz} V^W)_k.
\]

For a given candidate control $a_k$ at node $k$, the discrete HJB operator is
\begin{align*}
  \mathcal{H}_k(V^W;a_k,m)
  &:= w_k - c(a_k)
      + \mu_k^+(a_k,m)\,(D^- V^W)_k
      + \mu_k^-(a_k,m)\,(D^+ V^W)_k
      + \tfrac{1}{2}\sigma_k^2(m)\,(D^{zz} V^W)_k \\
  &\quad + G_k\bigl(V^W,m;a_k\bigr),
\end{align*}
where $G_k(\cdot)$ is a discrete, monotone approximation to the operator $G$ in~\eqref{eq:HJB_cont_app}, implemented as a quadrature rule in the code.

The stationary HJB at node $k$ in the continuation region is approximated by
\[
  r V^W_k = \sup_{a \in A^h} \mathcal{H}_k(V^W;a,m).
\]

\subsubsection*{Policy iteration and obstacle projection}

We solve the discrete obstacle problem via policy iteration, as is standard for stationary HJB equations in MFG and also in heterogeneous-agent models.\footnote{See, among others, \citet{AchdouEtAl2022Restud} and \citet{CarmonaLauriere2021DeepLearningMFG}.} Let $V^{W,(n)}$ denote the worker value vector at iteration $n$.

\begin{enumerate}
  \item \emph{Policy improvement.} Given $V^{W,(n)}$, compute
  \[
    a^{(n+1)}_k
    \in \arg\max_{a \in A^h} \mathcal{H}_k\bigl(V^{W,(n)};a,m\bigr)
  \]
  at each node $k$, and define the continuation set
  \[
    \mathcal{C}^{(n)} = \{k \colon V^{W,(n)}_k > V^U(m)\},
  \]
  with stopping set $\mathcal{S}^{(n)} = \{1,\dots,K\} \setminus \mathcal{C}^{(n)}$.

  \item \emph{Policy evaluation.}

  Given $a^{(n+1)}$, solve
\[
  r V^{W,(n+1)}_k
    = \mathcal{H}_k\bigl(V^{W,(n+1)};a^{(n+1)}_k,m\bigr),
    \quad k \in \mathcal{C}^{(n)},
\]
using the one-sided discrete derivatives at $k=1$ and $k=K$ defined
above, and impose
\[
  V^{W,(n+1)}_k = V^U(m),
  \qquad k \in \mathcal{S}^{(n)}.
\]
\end{enumerate}

Under the assumptions in Section~\ref{sec:existenceUniqueness}, this scheme converges to the unique solution of the discrete obstacle problem, and the discrete free-boundary index $k^\ast$ is identified as the smallest $k$ with $V^W_k > V^U(m)$. The associated surplus threshold is $z^\ast = z_{k^\ast}$. The firm HJB is discretized analogously on the same grid and with the same operators, with a separate policy iteration for the wage schedule $(w_k)_k$.

\subsection{Discretization of the stationary Kolmogorov equation}
\label{app:FP_scheme}

Given a stationary distribution $m$ and best-response policies $(a^\ast,w^\ast)$, the stationary Kolmogorov equation for the density $m(z)$ can be written in flux form as
\begin{equation}
  0 = -\partial_z J(z) - q(z,m)\,m(z) + \Gamma(z,m),
  \label{eq:FP_flux_app}
\end{equation}
where
\[
  J(z) 
  = \mu(z,a^\ast(z,m),m)\, m(z)
    - \tfrac{1}{2}\,\partial_z\bigl(\sigma^2(z,m)\, m(z)\bigr)
\]
is the probability flux, $q(z,m)$ is the separation intensity implied by the free-boundary rule, and $\Gamma(z,m)$ is the entry term.

We approximate \eqref{eq:FP_flux_app} by a conservative finite-volume scheme. Let $J_{k+1/2}$ denote the discrete flux across the cell boundary between $z_k$ and $z_{k+1}$:
\[
  J_{k+1/2}
  := \mu_{k+1/2}\,m_{k+1/2}
    - \tfrac{1}{2}\,\frac{\sigma^2_{k+1}\,m_{k+1} - \sigma^2_{k}\,m_{k}}{\Delta z},
\]
where $\mu_{k+1/2}$ and $m_{k+1/2}$ are suitable averages of $\mu_k,\mu_{k+1}$ and $m_k,m_{k+1}$, and $\sigma_k^2 = \sigma^2(z_k,m)$. 

The discrete stationary Kolmogorov equation on cell $k$ is
\[
  0
  = -\frac{J_{k+1/2} - J_{k-1/2}}{\Delta z}
    - q_k\,m_k + \Gamma_k,
  \qquad k = 2,\dots,K-1,
\]
with no-flux boundary conditions
\[
  J_{1/2} = J_{K+1/2} = 0
\]
at the left and right edges of the grid, implementing~(22). This yields a sparse,
banded linear system
\[
  \mathbf{A}(m,a^\ast,w^\ast)\,\mathbf{m} = \boldsymbol{\Gamma},
\]
where $\mathbf{A}$ is the discrete Kolmogorov operator and
$\mathbf{m} = (m_1,\dots,m_K)'$. We normalize so that
\[
  \sum_{k=1}^{K} m_k \,\Delta z = 1.
\]

In the benchmark calibration of Section~\ref{sec:quantitative}, the entry term $\Gamma(z,m)$ implements the assumption that new matches start at a common surplus $z_0$.\footnote{See the discussion of the inverse-Gaussian hitting-time calibration for $(\mu_P,\sigma_P,z_0)$ in Section~\ref{sec:quantitative}.} On the grid we therefore set $\Gamma_k(m)$ to be proportional to an indicator that $z_k$ is the grid point closest to $z_0$, so that inflows of new employment relationships enter at that node. The numerical implementation in the replication code is slightly more general and allows $\Gamma$ to place mass on an arbitrary set of grid points if one wishes to study robustness to a dispersed initial surplus distribution. As stated earlier, we keep the degenerate $z_0$ specification in the baseline to tie cross-sectional heterogeneity tightly to the diffusion-and-selection dynamics and because the tenure and hazard moments we target are not separately informative about dispersion in initial surpluses.

The finite-volume scheme is conservative by construction and, combined with the monotone HJB discretization, preserves non-negativity of the stationary density. In one dimension it is algebraically equivalent (up to $O(\Delta z)$ terms) to the backward/central-difference representation used in Section~\ref{sec:quantitative}.

\subsection{Fixed-point algorithm and finite-state MFG interpretation}
\label{app:fixed_point}

The stationary mean field equilibrium is computed as a fixed point of the mapping from a conjectured distribution $m$ to the invariant distribution induced by best-response policies to $m$. Numerically, we implement:

\begin{enumerate}
  \item \emph{Initialization.} Choose an initial guess $m^{(0)}$ on the grid (for instance, the stationary distribution of the partial-equilibrium benchmark in Section~\ref{sec:environment} or a diffuse lognormal density) and an associated outside value $V^{U,(0)}$.

  \item \emph{Policy step.} Given $m^{(\ell)}$, solve the worker and firm HJB problems using the finite-difference and policy-iteration schemes above, obtaining discrete policies
  \[
    a^{(\ell)}_k \approx a^\ast(z_k,m^{(\ell)}), 
    \qquad
    w^{(\ell)}_k \approx w^\ast(z_k,m^{(\ell)}),
  \]
  and a separation threshold $z^{\ast,(\ell)}$ with corresponding discrete stopping region.

  \item \emph{Distribution step.} Given $(a^{(\ell)},w^{(\ell)},z^{\ast,(\ell)})$, solve the discrete stationary Kolmogorov equation as in Appendix~\ref{app:FP_scheme} to obtain a new distribution $\tilde m^{(\ell+1)}$, normalized to integrate to one.

  \item \emph{Relaxation.} Update the distribution using
  \[
    m^{(\ell+1)}_k 
    = \omega \,\tilde m^{(\ell+1)}_k + (1-\omega)\,m^{(\ell)}_k,
    \quad k=1,\dots,K,
  \]
  with $\omega \in (0,1]$ chosen to stabilize the fixed-point iteration. Update $V^{U,(\ell+1)}$ accordingly.

  \item \emph{Stopping rule.} Iterate steps~2-4 until both the distribution and the wage schedule converge:
  \[
    \max_k \bigl|m^{(\ell+1)}_k - m^{(\ell)}_k\bigr| < \varepsilon_m,
    \qquad
    \max_k \bigl|w^{(\ell+1)}_k - w^{(\ell)}_k\bigr| < \varepsilon_w,
  \]
  for tolerances $(\varepsilon_m,\varepsilon_w)$.
\end{enumerate}

The grid-based system defines a finite-state mean field game in the sense of \citet{GomesEtAl2010FiniteStateDiscrete,GomesEtAl2024FiniteStateContinuous}: each grid point is a state of a continuous-time Markov chain, with transition rates determined by $(\mu,\sigma)$ and the separation rule, and payoffs given by $(w,c,\Pi)$. For a given $m$ and policy pair $(a^\ast,w^\ast)$, we can define a generator matrix $Q(m)$ with entries $q_{kj}(m)$ that approximate the transition rates of the surplus process from $z_k$ to $z_j$, including separation and re-entry. The stationary density $m$ is then a fixed point of the map
\[
  m \mapsto \Phi^h(m),
\]
where $\Phi^h(m)$ is the invariant distribution of the finite-state Markov chain with generator $Q(m)$. Existence and trend-to-equilibrium for such finite-state MFGs are established by \citet{GomesEtAl2010FiniteStateDiscrete,GomesEtAl2024FiniteStateContinuous}, and Theorem~\ref{thm:existence} shows that the continuous-state limit inherits this structure.

This fixed-point iteration can be interpreted as a sort of eductive learning procedure in the sense of \citet{Gueant2021ContinuousTimeOptimalControl}: agents conjecture a stationary distribution $m$, solve their HJB problems, obtain a Markov chain over states, and update their conjecture to the invariant distribution of that chain. A stationary mean field equilibrium is a fixed point of this learning map. Under the monotonicity conditions in Assumption~\ref{ass:monotonicity}, Theorem~\ref{thm:uniqueness} implies that this fixed point is unique; numerically, starting from very different initial guesses $m^{(0)}$ always leads to the same limiting distribution in the parameter region we study.

\subsection{Convergence diagnostics and sensitivity}
\label{app:diagnostics}

We summarize the diagnostics used to assess convergence of the numerical scheme and its sensitivity to discretization and algorithmic choices. Full numerical tables and additional plots are provided in the replication files.

\subsubsection*{Grid refinement}

We consider a sequence of grids with increasing numbers of points $K$ and decreasing step size $\Delta z$. For each grid we compute the stationary equilibrium $(V^h, m^h, z^{\ast,h})$ and record:
\begin{itemize}
  \item the $L^1$ distance between densities across grids,
    $\sum_k |m_k^{h} - m_k^{h'}| \Delta z$;
  \item the change in the free-boundary location $z^{\ast,h}$;
  \item key moments of the wage distribution (mean, variance, selected quantiles, and the model objects reported in Section~\ref{sec:results} such as the job-to-job rate and tenure-hazard profile).
\end{itemize}
Over the range of grids we consider, these objects are stable up to small numerical differences, and changes as $\Delta z \to 0$ are monotone. This behavior is consistent with the convergence of monotone schemes for coupled HJB–Fokker–Planck systems documented in \citet{AchdouCapuzzoDolcetta2010NumericsMFG} and \citet{AchdouEtAl2022Restud}.

\subsubsection*{Time-step robustness in policy iteration}

Although we work directly with the stationary HJB equation, policy iteration can be viewed as the limit of a time-discrete value iteration with step size $\Delta t$. In auxiliary experiments we replace policy iteration with a damped value iteration that solves
\begin{equation}
  V^{(n+1)}_k
  = (1-\theta)\, V^{(n)}_k
    + \theta \left[
      \frac{1}{r + 1/\Delta t}
      \max\Bigl\{
        w_k - c(a)
        + (\mathcal{L}^h_{a,m} V^{(n)})_k
        + G^h_k\bigl(a;V^{(n)},m\bigr),
        \; r\,V^U(m)
      \Bigr\}
    \right],
  \label{eq:value_iteration_backwards_euler}
\end{equation}
for a small $\Delta t$ and damping parameter $\theta \in (0,1)$. Inside the maximum both arguments are flows (utils per unit of time), i.e. the current-flow term from continuation and the flow equivalent $r V^U(m)$ of the outside option, and the common factor $1/(r + 1/\Delta t)$ converts these flows into value levels, so~\eqref{eq:value_iteration_backwards_euler} can be interpreted as a standard backward-Euler value-iteration step for the stationary HJB obstacle problem. The stationary solution obtained in this way coincides with the policy-iteration solution within numerical tolerance, consistent with the contraction-mapping properties of the discounted HJB operator.

\subsubsection*{Sensitivity to boundary truncation}

We vary $z_{\min}$ and $z_{\max}$ over a wide range while keeping the grid resolution fixed. For each configuration we recompute the stationary equilibrium and compare the implied wage and surplus distributions in the central region where $m^h$ has non-negligible mass. Note that as long as the boundaries are placed sufficiently far in the tails, the central part of the distribution and the wage-dispersion statistics reported in Section~\ref{sec:results} are virtually unchanged. This confirms that the results do not depend on arbitrary truncation choices or on the particular implementation of the no-flux conditions in~\eqref{eq:no_flux_bc_app}.

\subsubsection*{Search-intensity discretization}

We also vary the cardinality and placement of the search-intensity grid $A^h$. Increasing $N_a$ beyond the baseline has negligible effects on equilibrium outcomes, while very coarse grids generate small but measurable changes in the search policy at low surplus levels. The baseline calibration uses a grid dense enough that further refinement has no effect that is relevant in practice, in line with standard practice in numerical dynamic programming.

\subsubsection*{Mass conservation, non-negativity, and internal consistency}

For each converged equilibrium we verify that the discrete stationary Kolmogorov system implies
\[
  \sum_{k=1}^{K} m_k \,\Delta z = 1
\]
up to numerical round-off, and that all masses satisfy $m_k \geq 0$ to machine precision. The finite-volume scheme and the monotone HJB discretization guarantee these properties.

As an additional check, we simulate long sample paths of the controlled surplus diffusion using the converged policies and separation rule, construct the empirical stationary distribution, and compare it to the solution of the discrete stationary Kolmogorov equation. The two distributions coincide up to small sampling noise, confirming the internal consistency of the discretization and the finite-state MFG interpretation.

\subsubsection*{Robustness to algorithmic parameters}

We explore robustness to algorithmic parameters such as the relaxation coefficient $\omega$, the stopping tolerances $(\varepsilon_m,\varepsilon_w)$, and the initial distribution $m^{(0)}$. For a broad range of choices, the algorithm converges to the same stationary equilibrium, and the equilibrium objects reported in Section~\ref{sec:results} change by amounts far smaller than the empirical uncertainty in the calibration targets. Under the monotonicity and uniqueness conditions in Section~\ref{sec:existenceUniqueness}, this is the unique stationary equilibrium; in numerical experiments where we deliberately push parameters outside that region, multiple fixed points can in principle arise, in which case the algorithm selects the one that is locally stable under the eductive iteration in Appendix~\ref{app:fixed_point}.

\subsection{Additional figures}
\label{app:extra_figures}

We now document the additional figures that accompany the numerical implementation, calibration, and policy experiments in Sections~\ref{sec:quantitative} and~\ref{sec:results}. All figures are generated by the same solver that computes the benchmark stationary equilibrium, using the replication scripts.

For each figure we state the main objects plotted and highlight what they show about the model’s mechanics, in particular the free boundary, the job ladder, wage dispersion, and the variance decomposition.

\medskip
\noindent\textbf{Figure~\ref{fig:B1_value_function}: Value functions and the free boundary.}

\begin{figure}[t]
  \centering
  \includegraphics[width=0.8\textwidth]{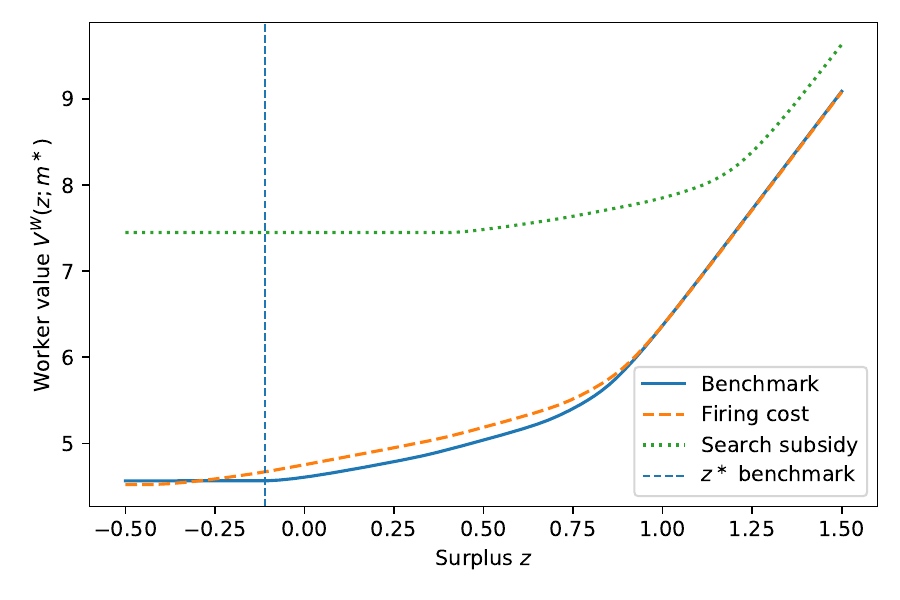}
  \caption{Worker value function $V^W(z;m^\ast)$ and free boundary in the benchmark calibration (solid line), with selected policy counterfactuals (dashed lines). The vertical line marks the separation threshold $z^\ast(m^\ast)$.}
  \label{fig:B1_value_function}
\end{figure}

Figure~\ref{fig:B1_value_function} plots the worker value function
$V^W(z;m^\ast)$ in the benchmark stationary equilibrium together with
the free boundary $z^\ast(m^\ast)$, as discussed in
Section~\ref{sec:results}. The smooth pasting of the derivative at
$z^\ast(m^\ast)$, where $V^W$ meets the flat outside-option level with
a continuous slope, illustrates the free-boundary structure proved in
Proposition~\ref{prop:free_boundary}; the visible change in curvature
around the threshold reflects the switch between the stopping and
continuation regions. The dashed curves show how $V^W$ and $z^\ast$
shift under representative policy experiments (higher firing costs,
lower search costs), making transparent the option-value and selection
forces behind the comparative statics in Section~\ref{sec:results}.

\medskip
\noindent\textbf{Figure~\ref{fig:B2_surplus_density}: Stationary surplus distribution and the job ladder.}

\begin{figure}[t]
  \centering
  \includegraphics[width=0.8\textwidth]{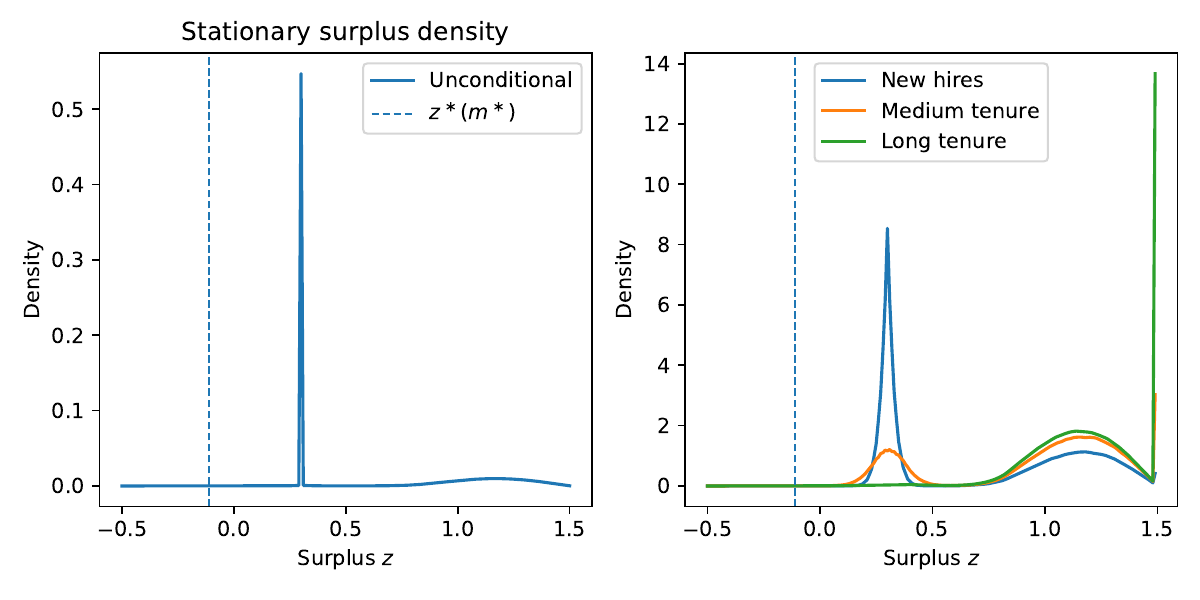}
  \caption{Stationary surplus density $m^\ast(z)$ and its decomposition by job age. Left panel: unconditional stationary density $m^\ast(z)$. Right panel: densities for new hires, medium-tenure workers, and long-tenure workers. The vertical line marks $z^\ast(m^\ast)$.}
  \label{fig:B2_surplus_density}
\end{figure}

Figure~\ref{fig:B2_surplus_density} shows the stationary surplus density $m^\ast(z)$ and decomposes it by job age. Mass is concentrated in a region of positive surplus, with a thin but economically meaningful left tail approaching the free boundary. New matches enter close to the threshold, while surviving high-tenure matches are disproportionately located in the upper tail. This is the diffusion-and-selection job ladder emphasized in Section~\ref{sec:results}: the combination of idiosyncratic shocks and the stopping rule gradually reallocates workers away from low-surplus matches toward high-surplus ones.

\medskip
\noindent\textbf{Figure~\ref{fig:B3_wage_distribution}: Wage distribution by tenure and mobility.}

\begin{figure}[t]
  \centering
  \includegraphics[width=0.8\textwidth]{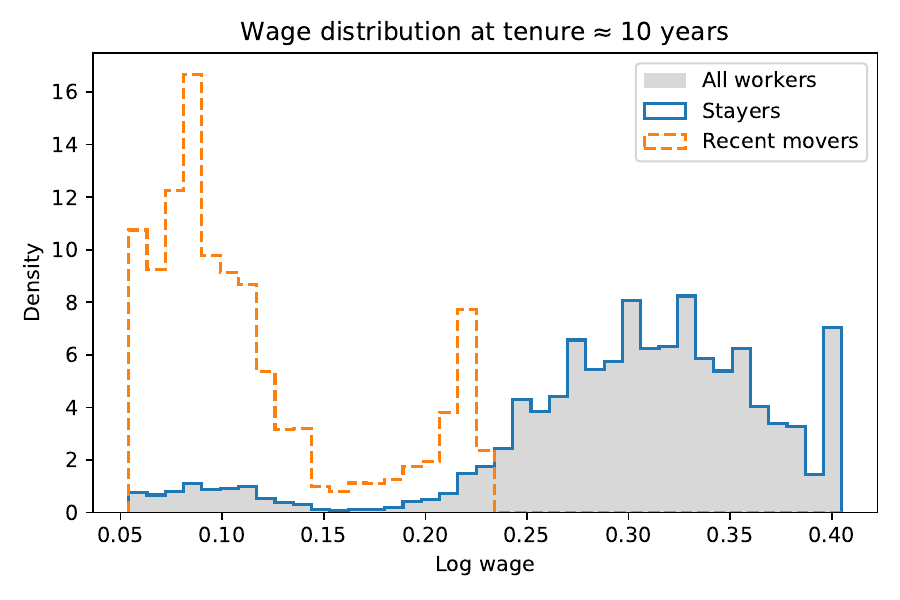}
  \caption{Wage distribution by tenure and mobility. Notes: histograms of stationary wages for (i) all workers in a representative bin around ten years of current-job tenure, (ii) stayers (workers who have never experienced a job-to-job move up to that point), and (iii) job-to-job movers (workers whose employment history includes at least one employer-to-employer transition). The figure highlights how job-to-job mobility thickens the upper tail of the wage distribution even conditional on tenure.}
  \label{fig:B3_wage_distribution}
\end{figure}

Figure~\ref{fig:B3_wage_distribution} reports, for this representative tenure bin, the stationary wage distribution decomposed by mobility status. Within the cohort of workers with roughly ten years of current-job tenure, the wage distribution for job-to-job movers is more dispersed and shifted toward the upper tail relative to that of stayers. This pattern reflects the fact that job-to-job transitions are the main mechanism for reaching high-surplus, high-wage states, and it links the stochastic surplus dynamics directly to observed wage dispersion conditional on tenure.

\medskip
\noindent\textbf{Figure~\ref{fig:B4_tenure_hazard}: Fit to the tenure-hazard profile.}

\begin{figure}[t]
  \centering
  \includegraphics[width=0.8\textwidth]{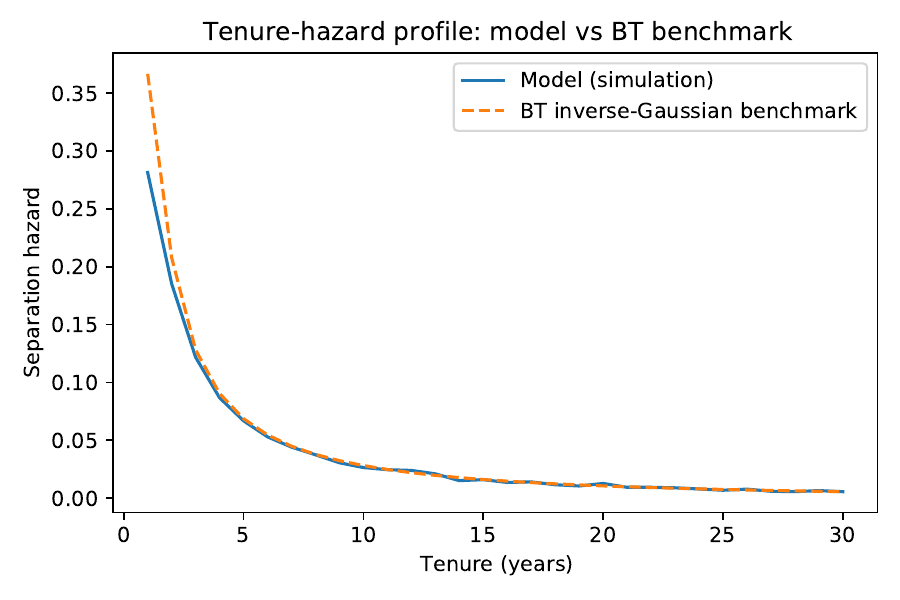}
  \caption{Tenure-hazard profile: model versus inverse-Gaussian benchmark. The solid line shows the model-implied separation hazard by tenure; the dashed line shows the inverse-Gaussian hazard implied by \citet{BuhaiTeulings2014} for a worker with average characteristics.}
  \label{fig:B4_tenure_hazard}
\end{figure}

Figure~\ref{fig:B4_tenure_hazard} compares the tenure-dependent separation hazard in the calibrated mean field equilibrium to the inverse-Gaussian hazard estimated by \citet{BuhaiTeulings2014}. The model reproduces the early-tenure hazard peak, the subsequent decline, and the nontrivial mass of very long spells, while embedding these duration patterns in a full equilibrium with search, wages, and endogenous outside options.

\medskip
\noindent\textbf{Figure~\ref{fig:B5_wage_tenure}: Wage-tenure profile and deterministic versus random components.}

\begin{figure}[t]
  \centering
  \includegraphics[width=0.8\textwidth]{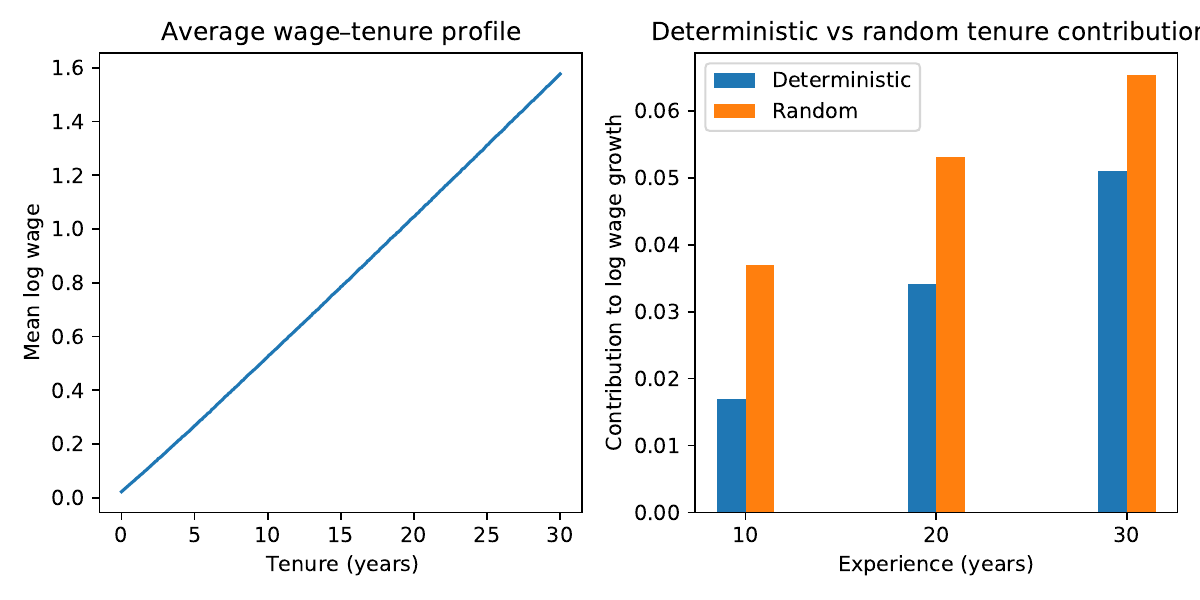}
  \caption{Average wage-tenure profile and deterministic versus random contributions to wage growth. Left panel: mean log wage by tenure. Right panel: decomposition of wage growth at 10, 20, and 30 years into deterministic and random tenure components, in the sense of \citet{BuhaiTeulings2014}.}
  \label{fig:B5_wage_tenure}
\end{figure}

Figure~\ref{fig:B5_wage_tenure} shows the model-implied wage-tenure profile and the decomposition of wage growth into deterministic and random components at 10, 20, and 30 years of experience. The left panel confirms that the calibrated model matches the concave average wage-tenure profile emphasized in the data. The right panel reproduces the deterministic-versus-random contributions from the structural decomposition of \citet{BuhaiTeulings2014}: most of the apparent returns to tenure come from selection on outside options rather than deterministic wage growth, even though the environment here is a full equilibrium mean field game rather than the partial-equilibrium benchmark.

\medskip
\noindent\textbf{Figure~\ref{fig:B6_convergence_paths}: Convergence of the fixed-point iteration.}

\begin{figure}[t]
  \centering
  \includegraphics[width=0.8\textwidth]{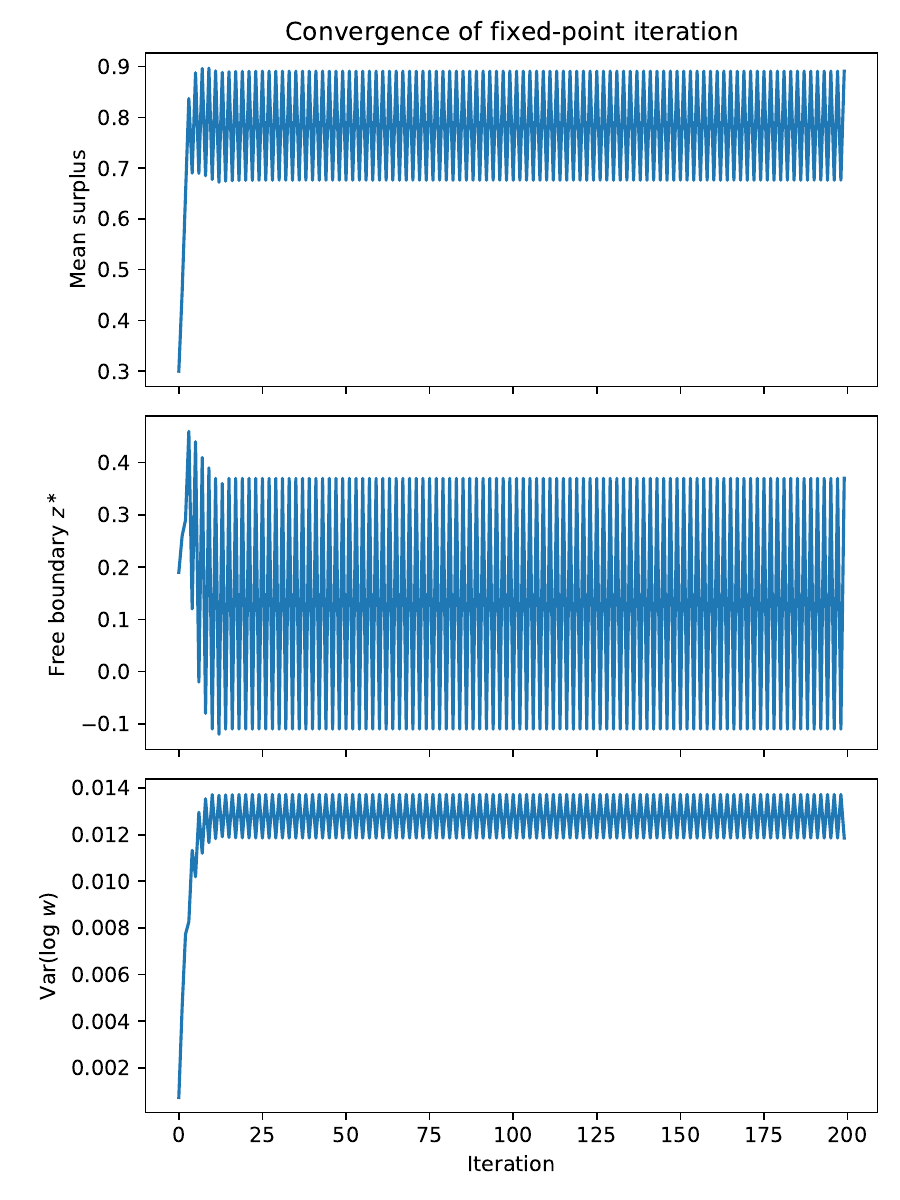}
  \caption{Convergence of the fixed-point iteration. Evolution of mean surplus, the free boundary $z^\ast$, and the cross-sectional variance of log wages along the sequence $(m^{(\ell)})_{\ell\geq 0}$ starting from a partial-equilibrium benchmark.}
  \label{fig:B6_convergence_paths}
\end{figure}

Figure~\ref{fig:B6_convergence_paths} illustrates the convergence of the fixed-point algorithm described in Appendix~\ref{app:fixed_point}. Starting from a partial-equilibrium distribution, the sequence of conjectured distributions and best-response policies converges quickly to the stationary equilibrium. Mean surplus, the free boundary, and wage dispersion stabilize together, highlighting the eductive-learning interpretation of the mean field fixed point and providing a visual check that the numerical implementation is converging to the equilibrium characterized in Section~\ref{sec:existenceUniqueness}.

\medskip
\noindent\textbf{Figure~\ref{fig:B7_variance_decomp}: Variance decomposition by tenure and parameter variation.}

\begin{figure}[t]
  \centering
  \includegraphics[width=0.8\textwidth]{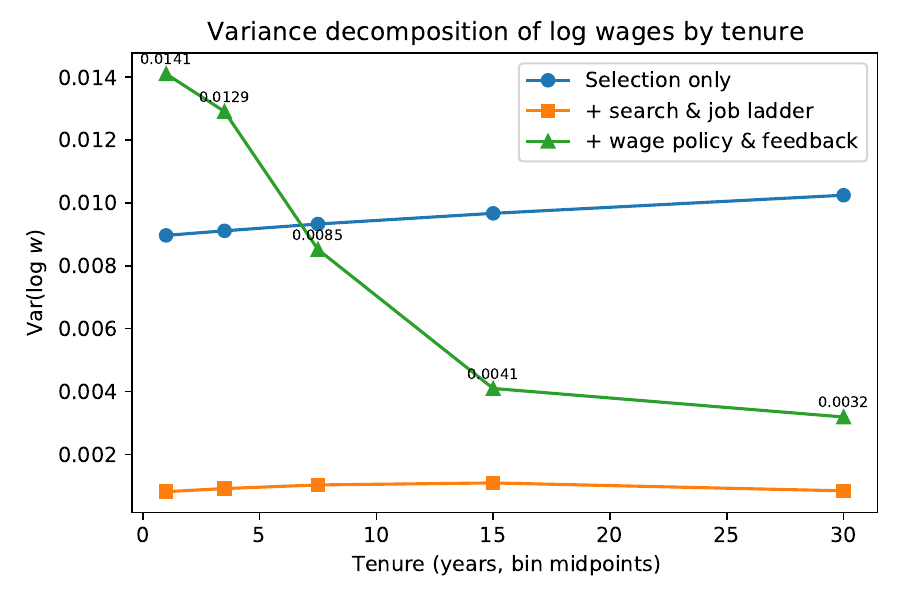}
  \caption{Decomposition of $\Var(\log w)$ by tenure bin. Circles: selection-only economy; squares: selection-plus-search economy; triangles: full mean field equilibrium with wage policies and feedback. Lines connect tenure bins. Additional curves (when present) show how the decomposition shifts under alternative search costs, bargaining weights, and volatility.}
  \label{fig:B7_variance_decomp}
\end{figure}

Figure~\ref{fig:B7_variance_decomp} reports the decomposition of $\Var(\log w)$ by tenure described in Section~\ref{sec:results}. Circles show the selection-only benchmark, squares the selection-plus-search economy, and
triangles the full mean field equilibrium. For very short tenures, search strongly equalizes
wages within a tenure cohort, so almost all dispersion comes from wage policies and the
endogenous outside option. At longer tenures, selection remains the dominant driver of
dispersion, but equilibrium search and wage setting substantially compress the variance
relative to the purely selection-driven benchmark. Additional curves (when present) show
how the decomposition shifts when we consider alternative calibrations of search costs,
bargaining power, and volatility in the underlying numerical experiments. For the volatility
and wage-sensitivity experiments, the direction of these shifts follows from the comparative
statics in Proposition~\ref{prop:comparative_statics} together with the discussion that follows it: higher match volatility
pushes the separation boundary $z^\ast$ downward, which, through its effect on the stationary
surplus distribution, raises dispersion at short tenures and reduces it at long tenures (see
the discussion after Proposition~\ref{prop:comparative_statics}), while stronger wage sensitivity mechanically amplifies
the vertical spread between the three components.

\medskip
\noindent\textbf{Figure~\ref{fig:B8_policy_experiments}: Policy experiments: wage dispersion and job-to-job mobility.}

\begin{figure}[t]
  \centering
  \includegraphics[width=0.8\textwidth]{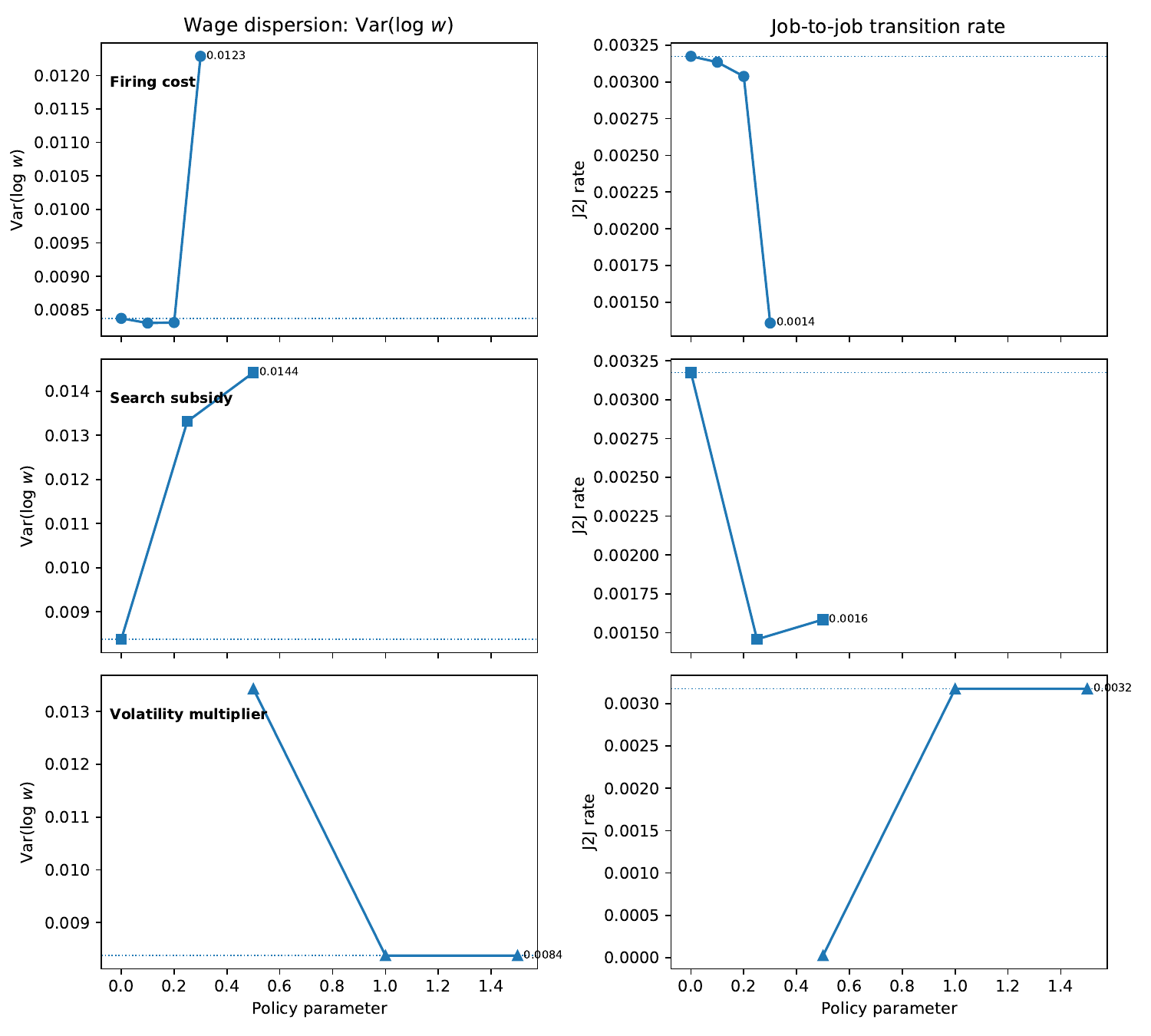}
  \caption{Policy experiments. Panels show the stationary variance of log wages and the aggregate job-to-job transition rate as functions of firing costs, effective search costs, and the volatility of match productivity, relative to the benchmark calibration.}
  \label{fig:B8_policy_experiments}
\end{figure}

Finally, Figure~\ref{fig:B8_policy_experiments} summarizes the comparative statics in Section~\ref{sec:results}. Increasing firing costs shifts the free boundary downward, reduces mobility, and eventually raises wage dispersion by trapping more workers in low-surplus matches while preserving a small set of very persistent high-surplus jobs. Reducing effective search costs (or raising outside options) raises average wages and, in our calibration, leads to more dispersion driven by a thicker right tail of high-surplus matches, even when aggregate job-to-job mobility changes little. Changes in volatility have more modest effects at the calibrated benchmark: they alter the shape of the surplus distribution and the duration mix, but within the empirically relevant range their impact on $\Var(\log w)$ and mobility is second-order relative to firing costs and search incentives.

\medskip

Taking stock, Figures~\ref{fig:B1_value_function}-\ref{fig:B8_policy_experiments} show that the numerical scheme is stable and internally consistent and that the calibrated MFG equilibrium matches the key tenure and wage-growth facts from the stochastic-productivity benchmark. As documented in Appendix~\ref{app:diagnostics}, the variance decomposition and the policy experiments in Section~\ref{sec:results} are also robust in the sense that their levels and qualitative rankings move only marginally when we (i) refine the state and action grids or tighten the convergence tolerances and (ii) symmetrically perturb each calibrated primitive by up to $\pm 5\%$ around its benchmark value.

\bibliographystyle{ecta}
\bibliography{references}

@article{Abbring2012MixedHittingTime,
  author  = {Abbring, Jaap H.},
  title   = {Mixed Hitting-Time Models},
  journal = {Econometrica},
  year    = {2012},
  volume  = {80},
  number  = {2},
  pages   = {783--819},
  month   = mar,
  doi     = {10.3982/ECTA7312},
}

@article{BurdettMortensen1998,
  author  = {Burdett, Kenneth and Mortensen, Dale T.},
  title   = {Wage Differentials, Employer Size, and Unemployment},
  journal = {International Economic Review},
  year    = {1998},
  volume  = {39},
  number  = {2},
  pages   = {257--273}
}

@article{PostelVinayRobin2002,
  author  = {Postel-Vinay, Fabien and Robin, Jean-Marc},
  title   = {Equilibrium Wage Dispersion with Worker and Employer Heterogeneity},
  journal = {Econometrica},
  year    = {2002},
  volume  = {70},
  number  = {6},
  pages   = {2295--2350}
}

@article{Moscarini2005JobMatchingWageDistribution,
  author  = {Moscarini, Giuseppe},
  title   = {Job Matching and the Wage Distribution},
  journal = {Econometrica},
  year    = {2005},
  volume  = {73},
  number  = {2},
  pages   = {481--516}
}

@article{BuhaiTeulings2014,
  author  = {Buhai, Ioan Sebastian and Teulings, Coen N.},
  title   = {Tenure Profiles and Efficient Separation in a Stochastic Productivity Model},
  journal = {Journal of Business \& Economic Statistics},
  year    = {2014},
  volume  = {32},
  number  = {2},
  pages   = {245--258}
}

@article{Jovanovic1979JobMatching,
  author  = {Jovanovic, Boyan},
  title   = {Job Matching and the Theory of Turnover},
  journal = {Journal of Political Economy},
  year    = {1979},
  volume  = {87},
  number  = {5},
  pages   = {972--990}
}

@article{Jovanovic1984MatchingTurnover,
  author  = {Jovanovic, Boyan},
  title   = {Matching, Turnover, and Unemployment},
  journal = {Journal of Political Economy},
  year    = {1984},
  volume  = {92},
  number  = {1},
  pages   = {108--122}
}

@article{MortensenPissarides1994JobCreationDestruction,
  author  = {Mortensen, Dale T. and Pissarides, Christopher A.},
  title   = {Job Creation and Job Destruction in the Theory of Unemployment},
  journal = {Review of Economic Studies},
  year    = {1994},
  volume  = {61},
  number  = {3},
  pages   = {397--415}
}

@article{AlvarezShimer2011SearchRestUnemployment,
  author  = {Alvarez, Fernando and Shimer, Robert},
  title   = {Search and Rest Unemployment},
  journal = {Econometrica},
  year    = {2011},
  volume  = {79},
  number  = {1},
  pages   = {75--122}
}

@article{AchdouEtAl2022Restud,
  author  = {Achdou, Yves and Han, Junjie and Lasry, Jean-Michel
             and Lions, Pierre-Louis and Moll, Benjamin},
  title   = {Income and Wealth Distribution in Macroeconomics:
             A Continuous-Time Approach},
  journal = {Review of Economic Studies},
  year    = {2022},
  volume  = {89},
  number  = {1},
  pages   = {45--86}
}

@article{FernandezVillaverdeNuno2021HeterogeneousAgents,
  author  = {Fern{\'a}ndez-Villaverde, Jes{\'u}s and Hurtado, Samuel and Nu{\~n}o, Galo},
  title   = {Financial Frictions and the Wealth Distribution},
  journal = {Econometrica},
  year    = {2023},
  volume  = {91},
  number  = {3},
  pages   = {869--901}
}

@article{LasryLions2007,
  author  = {Lasry, Jean-Michel and Lions, Pierre-Louis},
  title   = {Mean Field Games},
  journal = {Japanese Journal of Mathematics},
  year    = {2007},
  volume  = {2},
  number  = {1},
  pages   = {229--260}
}

@incollection{CardaliaguetPorretta2020IntroMFG,
  author    = {Cardaliaguet, Pierre and Porretta, Alessio},
  title     = {An Introduction to Mean Field Game Theory},
  booktitle = {Mean Field Games},
  editor    = {Achdou, Yves and Cardaliaguet, Pierre and Delarue, Fran{\c{c}}ois
               and Porretta, Alessio and Santambrogio, Filippo},
  publisher = {Springer},
  address   = {Cham},
  series    = {Lecture Notes in Mathematics},
  year      = {2020},
  pages     = {1--158}
}

@book{CarmonaDelarue2018BookI,
  author    = {Carmona, Ren{\'e} and Delarue, Fran{\c{c}}ois},
  title     = {Probabilistic Theory of Mean Field Games with Applications. {I}:
               Mean Field FBSDEs, Control, and Games},
  publisher = {Springer},
  address   = {Cham},
  year      = {2018}
}

@article{AchdouCapuzzoDolcetta2010NumericsMFG,
  author  = {Achdou, Yves and Capuzzo-Dolcetta, Italo},
  title   = {Mean Field Games: Numerical Methods},
  journal = {SIAM Journal on Numerical Analysis},
  year    = {2010},
  volume  = {48},
  number  = {3},
  pages   = {1136--1162}
}

@misc{Moll2019MFGMacroeconomics,
  author = {Moll, Benjamin},
  title  = {Mean Field Games in Macroeconomics},
  note   = {Lecture notes and slides, Princeton University},
  year   = {2019}
}

@incollection{GueantLasryLions2011MFGApplications,
  author    = {Gu{\'e}ant, Olivier and Lasry, Jean-Michel and Lions, Pierre-Louis},
  title     = {Mean Field Games and Applications},
  booktitle = {Paris--Princeton Lectures on Mathematical Finance 2010},
  publisher = {Springer},
  address   = {Berlin},
  year      = {2011},
  pages     = {205--266}
}

@incollection{Carmona2020MFGFinanceEcon,
  author    = {Carmona, Ren{\'e}},
  title     = {Applications of Mean Field Games in Financial Engineering and Economics},
  booktitle = {Mean Field Games},
  publisher = {Springer},
  address   = {Cham},
  year      = {2020},
  pages     = {183--204}
}

@article{CarmonaLauriere2021DeepLearningMFG,
  author  = {Carmona, Ren{\'e} and Lauri{\`e}re, Mathieu},
  title   = {Convergence Analysis of Machine Learning Algorithms for the Numerical Solution of Mean Field Games},
  journal = {SIAM Journal on Numerical Analysis},
  year    = {2021},
  volume  = {59},
  number  = {3},
  pages   = {1455--1485}
}

@misc{Gueant2021ContinuousTimeOptimalControl,
  author = {Gu{\'e}ant, Olivier},
  title  = {Continuous-Time Optimal Control on Discrete Spaces: Applications to Market Making},
  note   = {Lecture notes, KAUST Winter 2021},
  year   = {2021}
}

@article{Bertucci2018OptimalStoppingMFG,
  author  = {Bertucci, Charles},
  title   = {Optimal Stopping in Mean Field Games: A Probabilistic Approach},
  journal = {Journal de Math{\'e}matiques Pures et Appliqu{\'e}es},
  year    = {2018},
  volume  = {120},
  pages   = {165--194}
}

@article{Bertucci2020ImpulseControlMFG,
  author  = {Bertucci, Charles},
  title   = {Fokker--Planck Equations of Jumping Particles and Mean Field Games of Impulse Control},
  journal = {Annales de l'Institut Henri Poincar{\'e} C, Analyse Non Lin{\'e}aire},
  year    = {2020},
  volume  = {37},
  number  = {5},
  pages   = {1211--1244}
}

@article{Nutz2018OptimalStoppingMFG,
  author  = {Nutz, Marcel},
  title   = {A Mean Field Game of Optimal Stopping},
  journal = {SIAM Journal on Control and Optimization},
  year    = {2018},
  volume  = {56},
  number  = {2},
  pages   = {1206--1221}
}

@book{StroockVaradhan2006MultidimensionalDiffusions,
  author    = {Stroock, Daniel W. and Varadhan, S. R. Srinivasa},
  title     = {Multidimensional Diffusion Processes},
  series    = {Classics in Mathematics},
  publisher = {Springer},
  address   = {Berlin},
  year      = {2006},
  note      = {Reprint of the 1979 edition}
}

@article{GomesEtAl2010FiniteStateDiscrete,
  author  = {Gomes, Diogo A. and Mohr, Joana and Rig{\~a}o Souza, Rafael},
  title   = {Discrete Time, Finite State Space Mean Field Games},
  journal = {Journal de Math{\'e}matiques Pures et Appliqu{\'e}es},
  year    = {2010},
  volume  = {93},
  number  = {3},
  pages   = {308--328}
}

@article{GomesEtAl2024FiniteStateContinuous,
  author  = {Gomes, Diogo A. and Mohr, Joana and Rig{\~a}o Souza, Rafael},
  title   = {Continuous Time Finite State Mean Field Games},
  journal = {Applied Mathematics \& Optimization},
  year    = {2013},
  volume  = {68},
  number  = {1},
  pages   = {99--143},
  doi     = {10.1007/s00245-012-9186-6}
}

@misc{BertucciMeynard2024FiniteStateMaster,
  author = {Bertucci, Charles and Meynard, Charles},
  title  = {Noise through an Additional Variable for Mean Field Games Master Equation on Finite State Space},
  year   = {2024},
  note   = {arXiv:2402.05635 [math.AP]}
}

@article{AhujaRenYang2022,
  author  = {Ahuja, Saran and Ren, Weiluo and Yang, Tzu-Wei},
  title   = {Forward--Backward Stochastic Differential Equations with Monotone Functionals and Mean Field Games with Common Noise},
  journal = {Stochastic Processes and their Applications},
  year    = {2019},
  volume  = {129},
  number  = {10},
  pages   = {3859--3892}
}

@unpublished{MollRyzhik2025,
  author       = {Moll, Benjamin and Ryzhik, Lenya},
  title        = {Mean Field Games without Rational Expectations},
  year         = {2025},
  note         = {Working paper},
  howpublished = {\url{https://benjaminmoll.com/wp-content/uploads/2025/05/MFGRatX.pdf}}
}

@misc{Ryzhik2020,
  author = {Ryzhik, Lenya},
  title  = {Lecture Notes on Mean Field Games},
  year   = {2020},
  note   = {Lecture notes, Department of Mathematics, Stanford University}
}

@article{BaggerFontainePostelVinayRobin2014,
  author  = {Bagger, Jesper and Fontaine, Fran{\c{c}}ois and 
             Postel-Vinay, Fabien and Robin, Jean-Marc},
  title   = {Tenure, Experience, Human Capital, and Wages: 
             A Tractable Equilibrium Search Model of Wage Dynamics},
  journal = {American Economic Review},
  year    = {2014},
  volume  = {104},
  number  = {6},
  pages   = {1551--1596},
  doi     = {10.1257/aer.104.6.1551}
}

@book{Mortensen2005WageDispersion,
  author    = {Mortensen, Dale T.},
  title     = {Wage Dispersion: Why Are Similar Workers Paid Differently?},
  publisher = {MIT Press},
  address   = {Cambridge, MA},
  year      = {2005}
}

@book{PeskirShiryaev2006OptimalStopping,
  author    = {Peskir, Goran and Shiryaev, Albert N.},
  title     = {Optimal Stopping and Free-Boundary Problems},
  publisher = {Birkh{\"a}user},
  address   = {Basel},
  year      = {2006}
}

@article{LentzMortensen2010,
  author  = {Lentz, Rasmus and Mortensen, Dale T.},
  title   = {Labor Market Models of Worker and Firm Heterogeneity},
  journal = {Annual Review of Economics},
  year    = {2010},
  volume  = {2},
  pages   = {577--602}
}

@article{HornsteinKrusellViolante2011,
  author  = {Hornstein, Andreas and Krusell, Per and Violante, Giovanni L.},
  title   = {Frictional Wage Dispersion in Search Models: A Quantitative Assessment},
  journal = {American Economic Review},
  year    = {2011},
  volume  = {101},
  number  = {7},
  pages   = {2873--2898}
}

@article{PerthameRibesSalort2018,
  author  = {Perthame, Beno{\^\i}t and Ribes, Edouard and Salort, Delphine},
  title   = {Career Plans and Wage Structures: A Mean Field Game Approach},
  journal = {Mathematics in Engineering},
  year    = {2018},
  volume  = {1},
  number  = {1},
  pages   = {47--63},
  doi     = {10.3934/mine.2018004}
}

@unpublished{BayraktarCavalliReisinger2025,
  author  = {Ersin Bayraktar and Lorenzo Cavalli and Christoph Reisinger},
  year    = {2025},
  title   = {Optimal Matching: A Mean Field Game Approach},
  note    = {Working paper},
}

\end{document}